\documentclass{article}

\usepackage{arxiv}
\usepackage{amsfonts,amsthm,amssymb}

\usepackage[utf8]{inputenc}
\usepackage[T1]{fontenc}
\usepackage{newpxtext,newpxmath}

\usepackage{hyperref}
\usepackage{url}
\usepackage{booktabs}
\usepackage{nicefrac}
\usepackage{microtype}
\usepackage{lipsum}
\usepackage{graphicx}

\usepackage{enumitem}
\usepackage{mathtools}
\usepackage{subfigure}
\usepackage{xcolor}
\usepackage[numbers,sort&compress]{natbib}

\usepackage[linesnumbered,ruled,vlined]{algorithm2e}
\SetKwInOut{Input}{Input}
\SetKwInOut{Output}{Output}
\SetKwProg{Fn}{Function}{}{}

\newtheorem{condition}{Condition}
\newtheorem{definition}{Definition}
\newtheorem{example}{Example}
\newtheorem{lemma}{Lemma}
\newtheorem{theorem}{Theorem}

\DeclareMathOperator*{\argmin}{arg\,min}

\newcommand{\grmr}{GRMR\xspace}
\newcommand{\exact}{E-GRMR\xspace}
\newcommand{\heuristic}{H-GRMR\xspace}
\newcommand{\builddom}{\textsc{BuildDomGraph}\xspace}
\newcommand{\np}{\textnormal{NP}\xspace}
\newcommand{\extremes}{\ensuremath{X}\xspace}
\newcommand{\nphard}{{{\np}-hard}\xspace}
\newcommand{\delaunay}{IPDG\xspace}
\newcommand{\delgraph}{\ensuremath{\mathcal{G}}\xspace}
\newcommand{\vecangle}{\ensuremath{{\theta}}\xspace}
\newcommand{\domgraph}{\ensuremath{\mathcal{H}}\xspace}

\title{GRMR: Generalized Regret-Minimizing Representatives}

\author{Yanhao Wang\textsuperscript{1}, Michael Mathioudakis\textsuperscript{2}, Yuchen Li\textsuperscript{3}, Kian-Lee Tan\textsuperscript{1}\\
	\textsuperscript{1}National University of Singapore \qquad \textsuperscript{2}University of Helsinki \qquad \textsuperscript{3}Singapore Management University\\
	\textsuperscript{1}\small{\texttt{\{yanhao90,tankl\}@comp.nus.edu.sg}} \qquad \textsuperscript{2}\small{\texttt{michael.mathioudakis@helsinki.fi}} \qquad \textsuperscript{3}\small{\texttt{yuchenli@smu.edu.sg}}
}

\renewcommand{\headeright}{Technical Report}
\renewcommand{\undertitle}{Technical Report}

\hypersetup{
  pdftitle={GRMR: Generalized Regret-Minimizing Representatives},
  pdfsubject={cs.DB},
  pdfauthor={Yanhao Wang, Michael Mathioudakis, Yuchen Li, Kian-Lee Tan},
}

\begin{document}
\maketitle

\begin{abstract}
  Extracting a small subset of representative tuples from a large database is
  an important task in multi-criteria decision making.
  The \emph{regret-minimizing set} (RMS) problem
  is recently proposed for representative discovery from databases. Specifically,
  for a set of tuples (points) in $d$ dimensions, an RMS problem finds the smallest subset
  such that, for any possible ranking function, the relative difference in scores
  between the top-ranked point in the subset and the top-ranked point in the entire 
  database is within a parameter $\varepsilon \in (0,1)$.
  Although RMS and its variations have been extensively investigated in the literature,
  existing approaches only consider the class of
  nonnegative (monotonic) linear functions for ranking,
  which have limitations in modeling user preferences and decision-making processes.
  
  To address this issue, we define the \emph{generalized regret-minimizing representative}
  (GRMR) problem that extends RMS by taking into account all linear functions including
  non-monotonic ones with negative weights. For two-dimensional databases,
  we propose an optimal algorithm for GRMR via a transformation into 
  the shortest cycle problem in a directed graph.
  Since GRMR is proven to be NP-hard even in three dimensions,
  we further develop a polynomial-time heuristic algorithm for GRMR on databases in
  arbitrary dimensions. Finally, we conduct extensive experiments on real and synthetic
  datasets to confirm the efficiency, effectiveness, and scalability of our proposed algorithms.
\end{abstract}

\section{Introduction}\label{sec:intro}

Nowadays, a database usually contains millions of tuples and is beyond any user's
capability to explore it entirely. Hence, it is an important task to extract a small subset
of representative tuples from a large database.
In multi-criteria decision making, a common method for identifying representatives
is the \emph{top-$k$ query}~\cite{DBLP:journals/csur/IlyasBS08},
which selects $k$ tuples with the highest scores in a database based on
a utility (ranking) function to model user preferences. For databases with multiple
numeric attributes, the ranking function is often expressed in the form of
a linear combination of attributes w.r.t.~a utility vector.
Finding players with the best (or poorest) performance based on a linear
combination of their statistics~\cite{DBLP:journals/pvldb/ChesterTVW14,DBLP:conf/sigmod/AsudehN0D17}
and evaluating credit risks according to
a linear combination of criteria such as income and credit history~\cite{DBLP:conf/sigmod/YuAY12,DBLP:journals/kais/LuoWY09}
are a few examples.

However, the ranking function could be unknown a-priori in many cases as the preference may vary from user to user,
and it is impossible to design a unified ranking function to model a variety of user preferences.
In the absence of explicit ranking functions, the
\emph{maxima representation}~\cite{DBLP:conf/icde/BorzsonyiKS01,DBLP:journals/pvldb/NanongkaiSLLX10,DBLP:conf/sigmod/ChangBCLLS00,DBLP:journals/vldb/XieWL20}
is used for representative discovery from a large database.
Specifically, the maxima representation is
a subset that consists of the maxima (i.e., tuple with the highest score) of the database
for any possible ranking function. If tuples are viewed as points in Euclidean space,
the \emph{convex hull}~\cite{DBLP:books/sp/PreparataS85,DBLP:conf/sigmod/ChangBCLLS00} of 
the point set is a maxima representation for all linear functions.
As another example, the \emph{skyline}~\cite{DBLP:conf/icde/BorzsonyiKS01} of the database
containing all Pareto-optimal tuples is a maxima representation for all nonnegative monotonic
ranking functions. A large body of work has been done in these areas
(see~\cite{DBLP:journals/csur/IlyasBS08,DBLP:journals/corr/KalyvasT17} for extensive surveys).
However, a major issue with maxima representations
is that they can contain a large portion of tuples in the database. For example,
the number of vertices of the convex hull is $O(n^{1/3})$ for a set of $n$ random points
inside a unit circle~\cite{DBLP:journals/corr/abs-1111-5340}.
The problem would become worse in higher dimensions.
The convex hull often contains $O(n)$ vertices even for five-dimensional
databases~\cite{DBLP:conf/sigmod/AsudehN0D17,DBLP:conf/sigmod/AsudehN00J19}.

Consequently, it is necessary to find a smaller subset to approximate maxima representations.
One such method that has attracted much attention recently is the
\emph{\underline{R}egret-\underline{M}inimizing \underline{S}et}
(RMS) problem~\cite{DBLP:journals/pvldb/NanongkaiSLLX10,DBLP:conf/icde/PengW14,DBLP:conf/sigmod/AsudehN0D17,DBLP:conf/wea/AgarwalKSS17,DBLP:conf/sigmod/XieW0LL18,DBLP:journals/pvldb/ShetiyaAAD19,DBLP:journals/vldb/XieWL20}.
In the setting of RMS, we define the \emph{regret ratio} as the relative score difference
between the top-ranked tuple in a subset and the top-ranked tuple in the whole database
for a specific ranking function. Then, we use the \emph{maximum regret ratio}, i.e.,
the maximum of the regret ratios of a subset over all possible ranking functions,
to measure how well the subset approximates the maxima representation of the database.
Given an error parameter $ \varepsilon \in (0,1) $, an RMS problem\footnote{There is a dual formulation of RMS which asks for a subset of size $r$ that minimizes the maximum regret ratio. Considering the equivalence of both formulations, we do not elaborate on the dual problem anymore in this paper.}
asks for the smallest subset whose maximum regret ratio is at most $\varepsilon$.
However, a significant limitation of RMS and its variations~\cite{DBLP:journals/pvldb/ChesterTVW14,DBLP:conf/icde/ZeighamiW19,DBLP:conf/aaai/StorandtF19,DBLP:conf/sigmod/AsudehN00J19,DBLP:conf/sigmod/NanongkaiLSM12,DBLP:conf/sigmod/XieWL19,DBLP:conf/icde/XieW0T20}
is that they only consider the class of nonnegative (monotonic) linear ranking functions, i.e.,
the utility vector is restricted to be nonnegative in each dimension, which may not be suitable
for modeling user preferences and decision-making processes in many cases.
Specifically, they are unable to express the ``liking'' and ``dislike'' of an attribute
at the same time. For example, when evaluating players' performance, we can find that
a criterion is positive for one player but negative for another, depending on their positions
and play styles~\cite{doi:10.1080/24748668.2013.11868632}.
As another example, to find preferable houses in a real estate database,
some users favor the ones close to traffic as they are convenient while others favor
the ones far away from traffic as they are quiet~\cite{DBLP:conf/sigmod/XinHC07}.
In such cases, nonnegative linear functions
cannot express conflicting user preferences on the same attribute simultaneously.
Furthermore, in many applications such as epidemiological and financial analytics~\cite{DBLP:journals/kais/LuoWY09,DBLP:conf/sigmod/YuAY12,DBLP:conf/sigmod/ChangBCLLS00},
the attribute values of tuples and weights of utilities are permitted to be negative,
and both the maxima and minima of a database are considered as representatives.
Obviously, RMS is not applicable in such scenarios.

To address the limitations of RMS, we propose the
\emph{\underline{G}eneralized \underline{R}egret-\underline{M}inimizing \underline{R}epresentative}
(\grmr) problem in this paper. We first extend the notion of
\emph{maximum regret ratio} by considering the class of all linear functions
other than merely nonnegative ones. In this way, we can capture
the ``liking'' (resp.~positive weight) and ``dislike'' (resp.~negative weight)
of an attribute as well as the maxima and minima of a database all at once.
Then, \grmr is defined to find the smallest subset with a regret ratio
of at most $\varepsilon$ for any possible linear function.
Due to the extension of ranking function spaces, \grmr becomes more challenging
than RMS. Most existing algorithms for RMS such as \textsc{Greedy}~\cite{DBLP:journals/pvldb/NanongkaiSLLX10}
and its variants~\cite{DBLP:conf/icde/PengW14,DBLP:conf/sigmod/XieW0LL18}
cannot be used for \grmr because they heavily rely on the monotonicity
and non-negativity of ranking functions for computation. A few RMS algorithms, e.g.,
\textsc{$\varepsilon$-Kernel}~\cite{DBLP:conf/wea/AgarwalKSS17,DBLP:conf/icdt/CaoLWWWWZ17}
and \textsc{HittingSet}~\cite{DBLP:conf/wea/AgarwalKSS17,DBLP:conf/alenex/KumarS18}
can be adapted to \grmr, but unfortunately, they suffer from low efficiency
and/or inferior solution quality.

Therefore, it is essential to design more efficient and effective algorithms for \grmr.
First of all, we prove that \grmr is \nphard on databases of three or higher dimensionality.
Then, we provide an exact polynomial-time algorithm \exact for \grmr on two-dimensional
databases. The basic idea of \exact is to transform \grmr in 2D into an equivalent problem
of finding the shortest cycle of a directed graph. \exact has $O(n^3)$
running time in the worst case where $n$ is the number of tuples in the database
while always providing an optimal result for \grmr.
Furthermore, we develop a polynomial-time heuristic algorithm \heuristic for \grmr
on databases in arbitrary dimensions. Based on a geometric interpretation of
\grmr using \emph{Voronoi diagram} and
\emph{Delaunay graph}~\cite{DBLP:journals/csur/Aurenhammer91},
\heuristic simplifies \grmr to a \emph{minimum dominating set} problem on
a directed graph. Theoretically, the result of \heuristic always has a maximum regret ratio
of at most $\varepsilon$. But \heuristic cannot guarantee the minimality of result size.
Moreover, we discuss two practical issues on the implementation of \heuristic,
including building an approximate Delaunay graph for \heuristic
and improving the solution quality of \heuristic via graph materialization and reuse.

Finally, we conduct extensive experiments on real and synthetic datasets to evaluate
the performance of our proposed algorithms. The experimental results confirm that
\exact always provides optimal results for \grmr on two-dimensional data within reasonable time.
In addition, \heuristic not only is more efficient (up to $10^4$x faster in running time)
but also achieves better solution quality (up to $50$x smaller in result sizes)
than the baseline methods.

In summary, we make the following contributions in this paper.
\begin{itemize}
  \item We propose the \emph{generalized regret-minimizing representative} problem (\grmr) to
  find a small subset as an approximate maxima representation of a database for all possible
  linear functions. We prove that \grmr is \nphard in $\mathbb{R}^d$ when $d \geq 3$.
  (Section~\ref{sec:def})
  \item In the case of two-dimensional databases, we provide an exact polynomial-time algorithm
  \exact for \grmr. (Section~\ref{sec:alg:2d})
  \item For databases of higher dimensionality ($d \geq 3$), we develop a polynomial-time
  heuristic algorithm \heuristic for \grmr. We also discuss the practical
  implementation of \heuristic. (Section~\ref{sec:alg:hd})
  \item We conduct extensive experiments on real and synthetic datasets to demonstrate
  the efficiency, effectiveness, and scalability of our proposed algorithms. (Section~\ref{sec:exp})
\end{itemize}

\section{Preliminaries}\label{sec:def}

\subsection{Model and Problem Formulation}\label{subsec:def}

\textbf{Database Model:}
We consider a database $P$ of $n$ tuples, each of which has $d$ numeric attributes,
and so we represent a tuple in $P$ as a $d$-dimensional point $p=(p[1],\ldots,p[d])$
and view $P$ as a point set in $\mathbb{R}^d$. Without loss of generality, we assume
that attribute values are normalized to the range $[-1,1]$ so that $1$ corresponds to
the maximum possible value and $-1$ corresponds to the minimum possible value.
Tuples are ranked by scores according to
a \emph{ranking function} $f : \mathbb{R}^d \rightarrow \mathbb{R}$.
A tuple $p_i$ outranks $p_j$ based on $f$ if $f(p_i) > f(p_j)$;
when $f(p_i) = f(p_j)$, any consistent rule can be used for tie-break.
In this paper, we focus on the class of \emph{linear} ranking functions as
they are widely used both in practical settings as well as the literature on
regret minimization problems~\cite{DBLP:journals/pvldb/NanongkaiSLLX10,DBLP:conf/sigmod/AsudehN0D17,DBLP:conf/alenex/KumarS18,DBLP:conf/wea/AgarwalKSS17,DBLP:conf/sigmod/XieW0LL18}.
The score of a tuple $p$ based on a linear function $f$
with a utility vector $x=(x[1],\ldots,x[d]) \in \mathbb{R}^d$
is computed as the inner product of $p$ and $x$, i.e.,
$f_x(p) = \langle p,x \rangle = \sum_{i=1}^{d} p[i] \cdot x[i]$.
Notice that, from a geometric perspective, 
we consider the set of all utility vectors $x \in \mathbb{R}^d$
corresponding to the unit $(d-1)$-sphere
$\mathbb{S}^{d-1}=\{x \in \mathbb{R}^d : \|x\|=1 \}$.
Focusing on unit vectors incurs no loss of generality since rankings based on
linear functions are norm-invariant\footnote{Obviously, the \emph{zero vector} $0$ is out of consideration. Because $\langle p,0 \rangle = 0$ for any $p \in \mathbb{R}^d$, $0$ is not meaningful for regret minimization problems.}.
We use $\omega(x,P) = \max_{p \in P} f_x(p)$ to denote the score of
the top-ranked tuple in $P$ based on vector $x$.

\textbf{Maxima Representation:}
For a database $P$ and a class $\mathcal{F}$ of ranking functions,
the \emph{maxima representation} of $P$ w.r.t.~$\mathcal{F}$ is defined as
the subset of tuples in $P$ that are top-ranked (i.e., have the highest score)
for some function $f \in \mathcal{F}$.
As discussed in Section~\ref{sec:intro}, the \emph{convex hull}~\cite{DBLP:books/sp/PreparataS85}
and \emph{skyline}~\cite{DBLP:conf/icde/BorzsonyiKS01} of $P$ are
maxima representations when $\mathcal{F}$ is the set of all \emph{linear functions}
and the set of all \emph{nonnegative monotonic functions}, respectively. 
One problem with maxima representations is that they can include a large portion of tuples
in the database.
To address this issue, we propose to use a relaxed definition of \emph{maxima representation}
that would trade an approximate representation for a smaller size.
Towards this end, we define the \emph{regret ratio} $l_x(Q,P)$ of a subset $Q$ of $P$
w.r.t.~a utility vector $x$ as $l_x(Q,P) = 1 - \frac{\omega(x,Q)}{\omega(x,P)}$.
In other words, for a given linear function, the regret ratio denotes the relative loss of
approximating the top-ranked tuple in $P$ by the top-ranked tuple in $Q$.
Since the maxima representation takes all possible ranking functions into account,
we define the \emph{maximum regret ratio} $l(Q,P)$ of $Q$ over $P$
as the maximum of the regret ratios over all linear ranking functions, i.e.,
$l(Q,P)=\max_{x \in \mathbb{S}^{d-1}} l_x(Q,P)$.
Based on the above measures, we formally define a relaxed notion of
\emph{maxima representation} for all linear functions,
to which we refer as an ``$\varepsilon$-regret set''.
\begin{definition}[$\varepsilon$-Regret Set]
  A subset $Q \subseteq P$ is an $\varepsilon$-regret set of $P$
  if and only if $l(Q,P) \leq \varepsilon$.
\end{definition}
Intuitively, an $\varepsilon$-regret set contains at least one tuple with score
at least $(1-\varepsilon) \cdot \omega(x,P)$ for every utility vector $x$.
In what follows, when the context is clear, we will drop $P$ from the notations
of \emph{regret ratio} and \emph{maximum regret ratio}.

Note that, for the notions of \emph{regret ratio} and \emph{maximum regret ratio}
to be well-formulated, we require that the top-ranked tuple has a positive score.
Or formally,
\begin{condition}
  The database $P$ satisfies that $\omega(x,P) > 0$
  for every vector $x \in \mathbb{S}^{d-1}$.
\end{condition}
This condition guarantees $l_x(Q)$ and $l(Q)$ to be nonnegative for any $Q \subseteq P$.
In fact, it is equivalent to the condition that the convex hull $\mathcal{CH}(P)$ of $P$
contains the origin of axes as an interior point~\cite{DBLP:conf/nips/ZhouTX019}.
This condition is often mild since the attribute values are normalized
to $[-1,1]$ and $P$ is ``centered'' around the origin of axes.
In all that follows, we will assume that $P$ satisfies this condition.

\textbf{Problem Formulation:}
For a given database $P$ and a parameter $\varepsilon$,
there may be many different $\varepsilon$-regret sets of $P$.
Naturally, for the compactness of data representation,
we are interested in identifying the smallest possible among them.
We formally define this problem as \emph{Generalized Regret-Minimizing Representative} (\grmr).
\begin{definition}[\grmr]\label{problem:grmr}
  Given a database $P \subset \mathbb{R}^d$ and a real number $\varepsilon \in (0,1)$,
  find the smallest $\varepsilon$-regret set $Q^*_{\varepsilon}$ of $P$. Formally,
  \begin{equation*}
    Q^*_{\varepsilon} = \argmin_{Q \subseteq P \,:\, l(Q) \leq \varepsilon} |Q|   
  \end{equation*}
\end{definition}

Note that there is a dual formulation of \grmr, i.e.,
given a positive integer $r \in \mathbb{Z}^+$, find
a subset $Q$ of size at most $r$ with the smallest $l(Q)$.
One can trivially adapt an algorithm $\mathcal{A}$ for \grmr to solve the dual problem:
By performing a binary search on $\varepsilon$ and computing a solution of \grmr
using $\mathcal{A}$ for each value of $\varepsilon$, one can find the minimum value
of $\varepsilon$ so that the result size is at most $r$.
If $\mathcal{A}$ provides the optimal result for \grmr, the adapted algorithm can also
return the optimal result for the dual problem at the expense of an additional log factor
(for binary search) in running time.

\textbf{Complexity:}
The \grmr problem is trivial in case of $d=1$ since the two points with the minimum and maximum
attribute values are always an optimal result for \grmr, which can be computed in $O(n)$ time.
In case of $d=2$, we propose an exact polynomial-time algorithm for \grmr
as described in Section~\ref{sec:alg:2d}.
However, \grmr is \nphard for databases of higher dimensionality ($d \geq 3$).
\begin{theorem}\label{thm:np:hardness}
  The \emph{Generalized Regret-Minimizing Representative} (\grmr) problem
  is \nphard in $\mathbb{R}^d$ when $d \geq 3$.
\end{theorem}
Please refer to Appendix~\ref{proof:np:hardness} for the proof of Theorem~\ref{thm:np:hardness}.

\textbf{Relationships to Other Problems:}
The proof of Theorem~\ref{thm:np:hardness} essentially shows that \grmr is a generalization of
RMS~\cite{DBLP:journals/pvldb/NanongkaiSLLX10}: while RMS is defined for monotonic linear
functions (resp.~$x\in\mathbb{S}^{d-1}_+$), \grmr is defined for all linear functions
(resp.~$x\in\mathbb{S}^{d-1}$). This generalization makes \grmr more challenging than RMS.
Because most existing RMS algorithms rely on the monotonicity of ranking functions
for computation, they cannot be used for \grmr.

Furthermore, \grmr is a special case of $\varepsilon$-kernel~\cite{DBLP:journals/jacm/AgarwalHV04}.
An $\varepsilon$-kernel is a coreset to approximate the width of a point set within
an error ratio of $\varepsilon$ along all directions, where the width of $P$
for vector $x$ is defined as
$w(x,P) = \max_{p \in P} \langle p,x \rangle - \min_{q \in P} \langle q,x \rangle $.
It can be shown that any $\varepsilon$-regret set of $P$ is also an $\varepsilon$-kernel of $P$.
But the opposite does not always hold and algorithms for $\varepsilon$-kernels may not
provide valid results for GRMR. Nevertheless, with minor modifications, the
\underline{A}pproximate \underline{N}earest \underline{N}eighbor (ANN)
based algorithm of Agarwal et al.~\cite{DBLP:journals/jacm/AgarwalHV04} for $\varepsilon$-kernels
can compute an $\varepsilon$-regret set of size $O(1/\varepsilon^{(d-1)/2})$
for \grmr.
Although this algorithm provides a desirable theoretical property that the result size
is independent of $n$, there may exist much smaller $\varepsilon$-regret sets
in practice. Therefore, the results computed by the ANN-based
algorithm~\cite{DBLP:journals/jacm/AgarwalHV04} are of inferior quality for \grmr in general.

Finally, \grmr is also closely related to $\varepsilon$-nets~\cite{DBLP:conf/wea/AgarwalKSS17}.
An $\varepsilon$-net of $\mathbb{S}^{d-1}$ is a finite subset
$\mathcal{N} \subset \mathbb{S}^{d-1}$
where there exists a vector $\overline{x} \in \mathcal{N}$ with
$\| \overline{x}-x \| \leq \varepsilon$ for any vector $x \in \mathbb{S}^{d-1}$.
It is known that a set of $O(1/\varepsilon^{d-1})$ uniform vectors forms an $\varepsilon$-net
of $\mathbb{S}^{d-1}$. One can find an $O(1)$-approximation result of GRMR
by generating an $O(\varepsilon)$-net $\mathcal{N}$
and computing the minimum subset of $P$ that contains an $\varepsilon$-approximate
tuple in $P$ for any vector $x \in \mathcal{N}$. This is referred to as
the \textsc{HittingSet} algorithm in~\cite{DBLP:conf/wea/AgarwalKSS17}.
However, due to the high complexity of $\varepsilon$-nets,
\textsc{HittingSet} suffers from a low efficiency for \grmr computation,
particularly so in high dimensions.

\subsection{Voronoi Diagram and Delaunay Graph}\label{subsec:background}

In this subsection, we provide the background information on \emph{Voronoi diagram}
and \emph{Delaunay graph}, as they provide the foundations of our algorithms.

Voronoi diagrams and Delaunay graphs~\cite{DBLP:journals/csur/Aurenhammer91} are geometric data 
structures that are widely used in many application domains such as
similarity search~\cite{DBLP:journals/pami/MalkovY20,DBLP:conf/nips/MorozovB18},
mesh generation~\cite{DBLP:journals/comgeo/Shewchuk02},
and clustering~\cite{DBLP:conf/compgeom/InabaKI94}.
They are defined in terms of a similarity measure between data points, which is set to
the \emph{inner product} of vectors in this paper.
The inner-product variants of these structures have been used in the literature
on graph-based approaches to \emph{maximum inner product search}~\cite{DBLP:conf/nips/ZhouTX019,DBLP:conf/nips/MorozovB18,DBLP:conf/emnlp/TanZXL19}.
To the best our knowledge, our work is the first to
introduce them into regret minimization problems.

\textbf{Voronoi Diagram:}
Let the $n$ points in a database $P$ be indexed by $[1,n]$.
The Voronoi diagram of $P$ is a collection of \emph{Voronoi cells},
one defined for each point.
The Voronoi cell $R(p_i)$ of point $p_i \in P$ is defined as
the vector set 
\begin{equation*}
  R(p_i) \coloneqq
  \{x \in \mathbb{R}^{d} \setminus \{0\} \,:\, \langle p_i,x \rangle \geq \omega(x,P)\}
\end{equation*}
i.e., the set of (non-zero) vectors based on which $p_i$ is top-ranked.
A point $p_i\in P$ is an \emph{extreme point} if and only if its Voronoi cell $R(p_i)$
is non-empty, i.e., there is at least one vector
based on which $p_i$ is top-ranked in $P$.
It is not difficult to see that the set \extremes of extreme points consists
exactly of the set of vertices of $\mathcal{CH}(P)$.

To illustrate \grmr geometrically, it is useful to
define the $\varepsilon$-approximate Voronoi cell $R_{\varepsilon}(p_i)$ of a point $p_i$ 
as the vector set
\begin{equation*}
  R_{\varepsilon}(p_i) \coloneqq \{x \in \mathbb{R}^d \setminus \{0\} \,:\,
  \langle p_i,x \rangle \geq (1-\varepsilon)\cdot \omega(x,P)\}
\end{equation*}
i.e., the set of (non-zero) vectors based on which the score of $p_i$ is
at least a $(1-\varepsilon)$-approximation of the score of the top-ranked point.
From a geometric perspective, \grmr can be seen as a set cover problem
of finding the minimum subset of points such that the union of their $\varepsilon$-approximate
Voronoi cells is $\mathbb{R}^d \setminus \{0\}$.

\textbf{Delaunay Graph:}
The \emph{\underline{I}nner-\underline{P}roduct \underline{D}elaunay \underline{G}raph}
(\delaunay) of a database $P$ records the adjacency information of the Voronoi cells of
extreme points in $P$. Formally, it is an undirected graph $\delgraph(P)=(V,E)$
where $V=P$, and there exists an edge $\{p_i,p_j\}\in E$
if and only if $R(p_i) \cap R(p_j) \neq \varnothing$.
Intuitively, an \delaunay has an edge between two extreme points that tie
as top-ranked for some utility vector.
Typically, the number of edges in $\delgraph(P)$ grows exponentially with $d$,
and thus building an exact \delaunay is often not computationally feasible in high
dimensions~\cite{DBLP:conf/emnlp/TanZXL19}.
Nevertheless, for a point set $P \subset \mathbb{R}^2$, since $\mathcal{CH}(P)$ is
a convex polygon and each edge of $\delgraph(P)$ exactly corresponds
to an edge of $\mathcal{CH}(P)$, we can easily build $\delgraph(P)$
based on $\mathcal{CH}(P)$ and the maximum degree of $\delgraph(P)$ is $2$.

\begin{figure}
  \centering
  \subfigure[Voronoi cells and IPDG]{
    \includegraphics[width=0.25\textwidth]{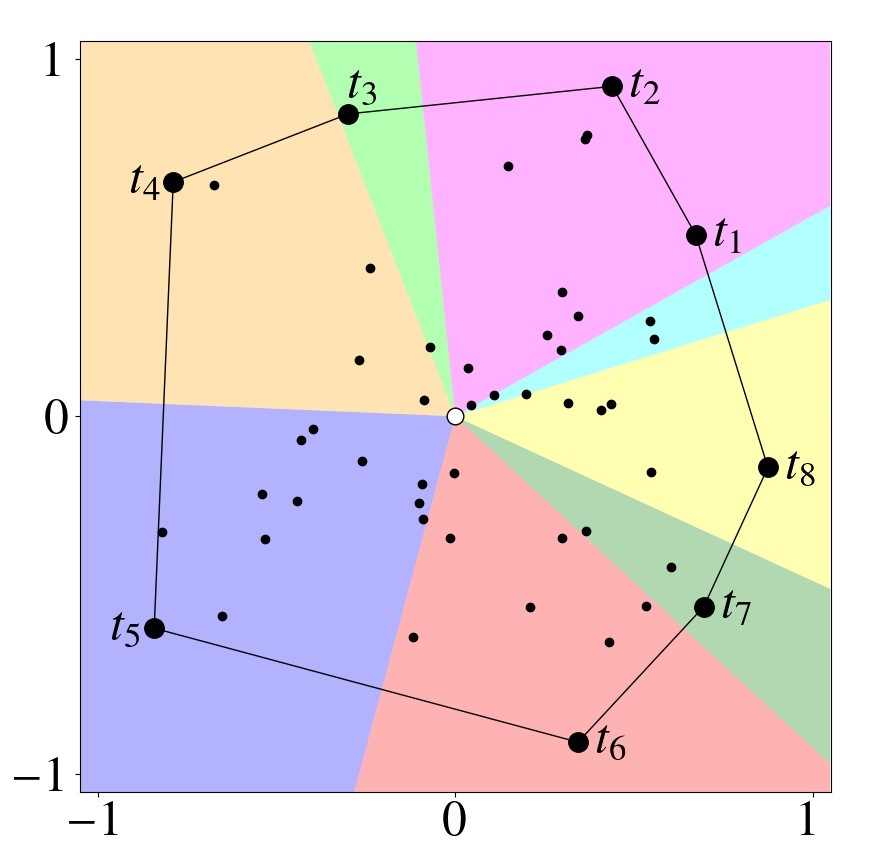}
    \label{fig:ipdg}
  }
  \hspace{1em}
  \subfigure[Attribute values]{
    \includegraphics[width=0.15\textwidth]{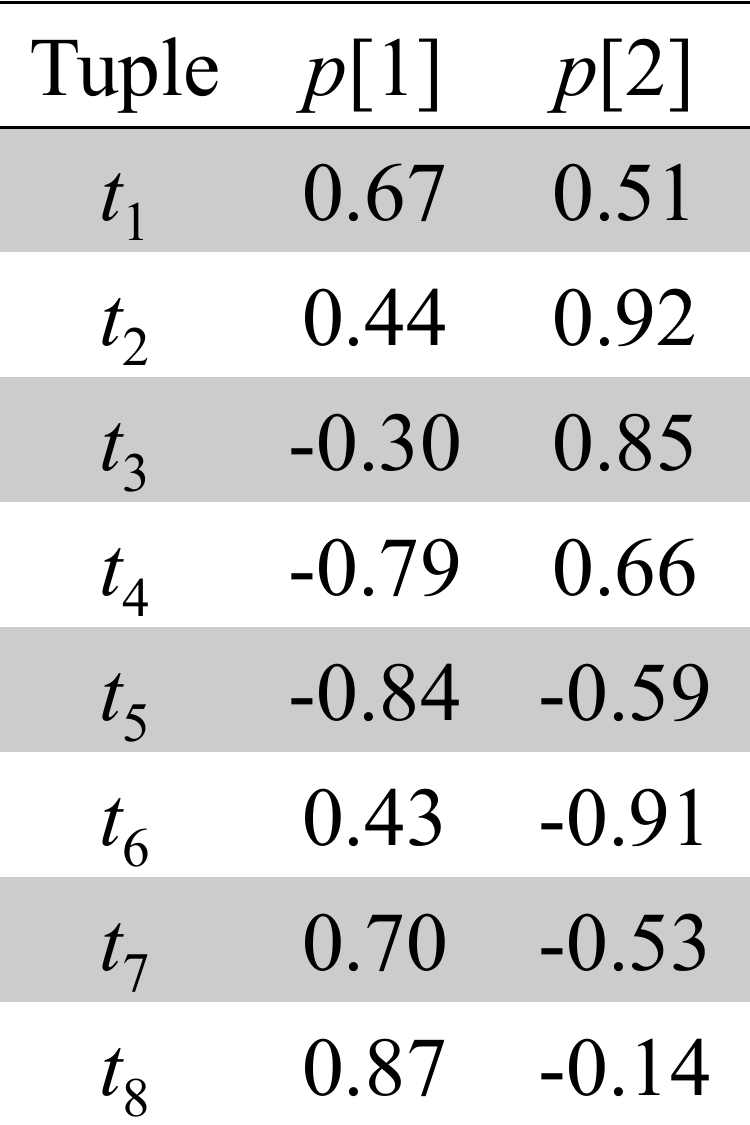}
    \label{fig:tuples}
  }
  \caption{Example of two-dimensional dataset}
  \label{fig:example}
\end{figure}

Figure~\ref{fig:example} illustrates a point set $P$ in $\mathbb{R}^{2}$.
The extreme points $\extremes=\{t_1,\ldots,t_8\}$ are highlighted
and their Voronoi cells are filled with different colors in Figure~\ref{fig:ipdg}.
Note that an extreme point $t$ does not necessarily fall in its Voronoi cell,
as another point $t'$ may be top-ranked for vector $t$,
i.e., $ \langle t,t \rangle < \langle t',t \rangle $.
This explains why $t_1$ is not contained in $R(t_1)$ (the cyan region)
but in $R(t_2)$ (the pink region).
Figure~\ref{fig:ipdg} also illustrates the IPDG of the point set, where any two extreme points
with adjacent Voronoi cells are connected.

\section{Exact Algorithm in 2D}\label{sec:alg:2d}

In this section, we present \exact, our optimal polynomial-time algorithm
for \grmr in $\mathbb{R}^{2}$.
Before delving into the details, we first present some high-level ideas
on the geometric properties of \grmr in $\mathbb{R}^{2}$ and our approach to solving it.

First of all, the set of all linear ranking functions in $\mathbb{R}^{2}$ corresponds to
the one-dimensional unit sphere $\mathbb{S}^1$, i.e., the unit circle,
and thus both exact and $\varepsilon$-approximate
Voronoi cells correspond to arcs of $\mathbb{S}^1$.
Therefore, \grmr in $\mathbb{R}^{2}$ can be equivalently formulated as the problem of
\emph{finding the minimum number of arcs (cells) that fully cover $\mathbb{S}^1$}.
One straightforward approach to this problem is to first compute the $\varepsilon$-approximate
Voronoi cells for an input parameter $\varepsilon$ and then solve the
aforementioned covering problem. Then, we observe that the exact Voronoi cell of each extreme
point can be computed via simple comparisons with its two neighbors in the \delaunay
(recall that the degree of each vertex in $\delgraph(P)$ is $2$ when $d = 2$).
Moreover, computing an $\varepsilon$-approximate Voronoi cell of a point requires
comparisons with merely extreme points. After the Voronoi cells are computed,
the arc-covering problem is polynomial because it is one-dimensional.
In the optimal algorithm we propose below, we further manages to avoid explicit
computations of $\varepsilon$-approximate Voronoi cells and reduce the cost of the
arc-covering problem via a graph-based transformation.

\subsection{The \exact Algorithm}

The main idea behind \exact is to build a directed graph $G$ such that
each (directed) cycle corresponds to an $\varepsilon$-regret set.
Then, the shortest cycle in $G$ provides an optimal solution for \grmr.
The details of \exact is shown in Algorithm~\ref{alg:grmr:2d}.
Note that, besides the database $P$ and parameter $\varepsilon$,
the set \extremes of extreme points is also assumed to be given as input.
For any given $P$, $X$ can be computed using any existing convex hull algorithm,
such as Graham's scan~\cite{DBLP:journals/ipl/Graham72}
and Qhull~\cite{DBLP:journals/toms/BarberDH96}.

For ease of illustration, we arrange all points and vectors in a counterclockwise direction based on their corresponding angles from $0$ to $2\pi$ in the polar coordinate system. 
For a point $p \in \mathbb{R}^2$ or a vector $x \in \mathbb{R}^2$, we use $\vecangle(p)$ or $\vecangle(x) \in [0,2\pi)$
to denote the angle of $p$ or $x$, respectively.
Furthermore, given any two points (vectors) $a,b \in \mathbb{R}^2$, we use $P[a,b] \subseteq P$
to denote a subset of points in $P$ whose angles are in range $[\vecangle(a),\vecangle(b)]$
if $\vecangle(a)<\vecangle(b)$ or ranges $[\vecangle(a),2\pi)$ and $[0,\vecangle(b)]$
if $\vecangle(a)>\vecangle(b)$.

\begin{algorithm}
  \caption{\exact}\label{alg:grmr:2d}
  \Input{Database $P$, Extreme points $X \subseteq P$,
         Parameter $\varepsilon \in (0,1)$}
  \Output{Optimal result $Q^*_{\varepsilon}$ of GRMR}
  \For{$i \gets 1,\ldots,|\extremes|$\label{ln:cand:s}}{
    Let $x^*_i$ be the vector $x \in \mathbb{S}^1$ s.t.~$\langle t_i,x \rangle = \langle t_{i+1},x \rangle \wedge \langle t_i,x \rangle > 0 $\;
  }
  Initialize the candidate set $S \gets X$\;
  \ForEach{point $ p \in P \setminus X $}{
    \If{$ \exists i \in [|\extremes|] : \langle p,x^*_i \rangle \geq (1-\varepsilon) \cdot
    \langle t_i,x^*_i \rangle $\label{ln:cand:condition}}{
      $S \gets S \cup \{p\}$\;
    }
  }
  Arrange all points of $S$ in a counterclockwise direction and index them by
  $[1,\ldots,|S|]$ accordingly\;\label{ln:cand:t}
  Initialize a directed graph $G=(U,A)$ where $U=S$ and $A=\varnothing$\;\label{ln:graph:s}
  \ForEach{$ i,j \in [|S|]$ and $i \neq j $}{
    \textbf{if} $\angle s_i O s_j \geq \pi$ \textbf{then} \textbf{continue}\;
    $l_{ij} \gets$ \textsc{ComputeRegret}$(s_i,s_j)$\;
    \If{$ l_{ij} \leq \varepsilon $}{
      Add a directed edge $(s_i \rightarrow s_j)$ to $A$\;
      \label{ln:graph:t}
    }
  }
  $C^* \gets$ \textsc{ShortestCycle}$(G)$\;\label{ln:cycle}
  \Return{$Q^*_{\varepsilon} \gets \{q : q \in C^*\}$}\;\label{ln:result}
  \BlankLine
  \Fn{\textnormal{\textsc{ComputeRegret}$(s_i,s_j)$}\label{ln:regret:s}}{
    Find the subset of extreme points $X[s_i,s_j] \subseteq X$\;
    Let $x^*$ be the vector $x \in \mathbb{S}^1$ s.t.~$\langle s_i,x \rangle = \langle s_j,x \rangle \wedge \langle s_i,x \rangle \geq 0$\;
    \Return{$ l_{ij} \gets \max_{t \in X[s_i,s_j]} \big(1 - \frac{\langle s_i,x^* \rangle}{\langle t,x^* \rangle}\big) $}\;\label{ln:regret:t}
  }
\end{algorithm}

\exact proceeds in three steps as follows.

\textbf{Step 1 (Candidate Selection, Lines~\ref{ln:cand:s}--\ref{ln:cand:t}):}
The purpose of this step is to identify a subset $S \subseteq P$ of points that may be included
in the solution, while pruning from consideration the points that are certainly not.
To find the candidate points, \exact essentially ignores all points that are never
within an $(1-\varepsilon)$-approximation from a top-ranked point.
In other words, $p \in S$ if and only if the $\varepsilon$-approximate Voronoi cell of $p$
is non-empty, i.e., $R_{\varepsilon}(p) \neq \varnothing$.
Note that this condition holds for any extreme point in $X$.
To determine whether it holds for a point $p \in P \setminus X$ or not,
\exact needs to find at least one vector $x \in R_{\varepsilon}(p)$
by comparing $p$ with each extreme point. In particular, for a point $p \in P \setminus X$
and an extreme point $t \in \extremes$, it should decide if $R(t) \cap R_{\varepsilon}(p)$
is empty. To this end, it computes the minimum of the regret ratio of $p$
w.r.t.~$t$ in $R(t)$. Indeed, this minimum is always reached at a boundary vector
where the score of $t$ is equal to the score of the previous/next extreme point.
The correctness of the above results will be analyzed
in Section~\ref{subsec:analysis:2d}.

Putting everything together, \exact first computes the boundary vector $x^*_i$ of $R(t_i)$
and $R(t_{i+1})$ for each pair of neighboring extreme points $t_i,t_{i+1} \in X$
(when $i = |\extremes|$, $t_{|\extremes|}$ and $t_1$ are used to compute $x^*_{|X|}$).
Then, it checks the regret ratio of $p$ for each $x^*_i$ ($i \in [1,|\extremes|]$)
and adds $p$ to $S$ if its regret ratio is at most $\varepsilon$ for some $x^*_i$.
Finally, all points in $S$ are arranged in a counterclockwise direction from $0$ to $2\pi$
and indexed by $[1,\ldots,|S|]$ accordingly.

\textbf{Step 2 (Graph Construction, Lines~\ref{ln:graph:s}--\ref{ln:graph:t}):}
The purpose of this step is to identify each pair of points in $S$ that,
if included in the solution, would make the points between them redundant
(in the sense that the solution without these points is still an $\varepsilon$-regret set).
This is achieved by building a directed graph $G=(U,A)$ as follows.
First, $G$ includes all points in $S$ as vertices.
For each two vertices $s_i,s_j \in U$ ($i \neq j$), there will be
a directed edge from $s_i$ to $s_j$ if the maximum regret ratio $l_{ij}$
caused by removing all points in $S$ between them (excluding themselves),
i.e., $S[s_{i+1},s_{j-1}]$, is at most $\varepsilon$.
Note that if the vector angle between $s_i$ and $s_j$ is greater than $\pi$,
the regret led by removing points between them will be unbounded.
In this case, the regret computation will be skipped directly.

The procedure to compute $l_{ij}$ is given in Lines~\ref{ln:regret:s}--\ref{ln:regret:t}.
Firstly, it retrieves all extreme points $X[s_i,s_j]$ between $s_i$ and $s_j$.
If there is no extreme point in $X[s_i,s_j]$ (other than $s_i$ or $s_j$ possibly),
then it sets $l_{ij}=0$, leading to the inclusion of an edge from $s_i$ to $s_j$
for any $\varepsilon \geq 0$.
Otherwise, to compute $l_{ij}$, it first finds the boundary vector $x^*$ where $s_i$ and $s_j$
have the same score. Similar to the case of \emph{candidate selection},
the maximum of the regret ratio
of $\{s_i,s_j\}$ w.r.t.~any extreme point $t$ between them is also always reached
at the boundary vector $x^*$. Subsequently, it computes the regret ratio of $\{s_i,s_j\}$
w.r.t.~each extreme point $t \in X[s_i,s_j]$ for vector $x^*$
and finally returns the maximum one as $l_{ij}$.

\textbf{Step 3 (Result Computation, Lines~\ref{ln:cycle}--\ref{ln:result}):}
Given a directed graph $G=(U,A)$ constructed according to the above procedure,
the final step is to find the shortest cycle $C^*$ of $G$ and
return the vertices of $C^*$ as the optimal result $Q^*_{\varepsilon}$ of GRMR on
database $P$.
In our implementation, we run Dijkstra's algorithm~\cite{DBLP:journals/nm/Dijkstra59}
from each vertex of $G$ to find the shortest cycle $C^*$.

\begin{figure}
  \centering
  \subfigure[Candidate set $S$]{
    \label{subfig:2d:cand}
    \includegraphics[width=0.25\textwidth]{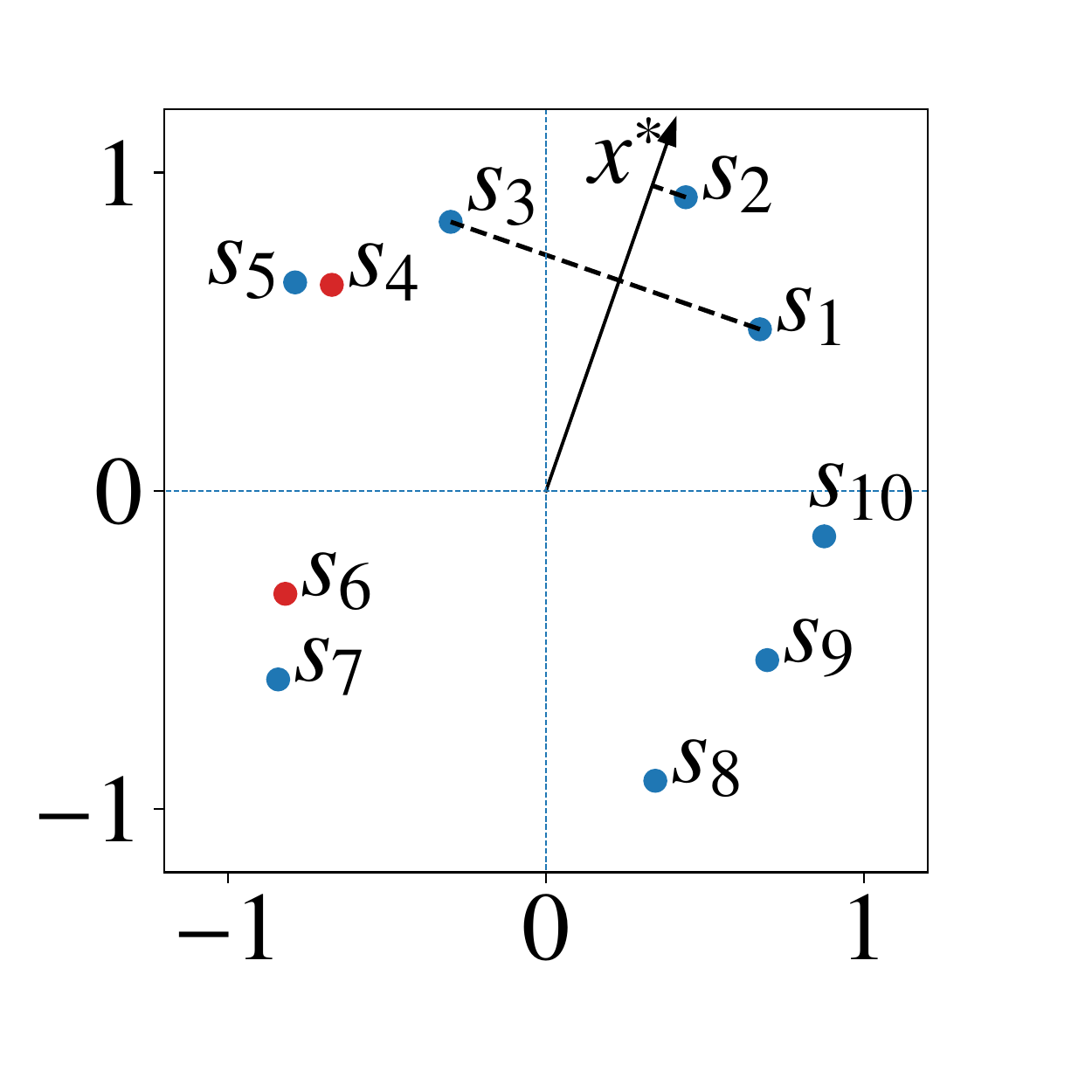}
  }
  \hspace{1em}
  \subfigure[Graph $G$ for $\varepsilon=0.1$]{
    \label{subfig:2d:cycle}
    \includegraphics[width=0.25\textwidth]{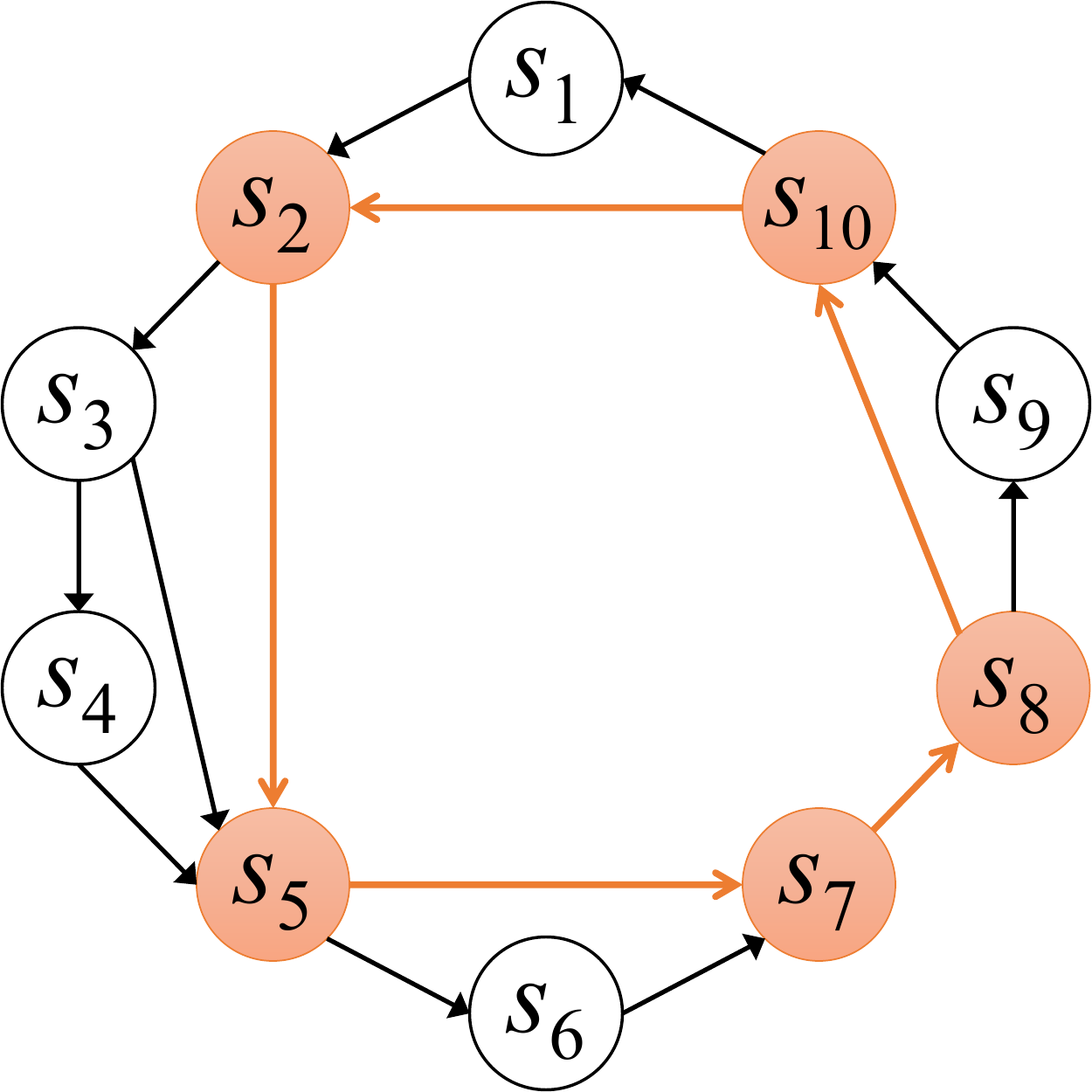}
  }
  \caption{Example for \exact in 2D}
  \label{fig:2d:example}
\end{figure}

\begin{example}
  We give an example of Algorithm~\ref{alg:grmr:2d} in Figure~\ref{fig:2d:example}.
  In Figure~\ref{subfig:2d:cand}, we illustrate the candidate set $S$ for $\varepsilon=0.1$
  extracted from the point set in Figure~\ref{fig:ipdg}.
  The extreme points (resp.~$t_1$--$t_8$ in Figure~\ref{fig:tuples}) are blue
  and the remaining two candidates ($s_4=(-0.67,0.65)$ and $s_6=(-0.82,-0.32)$) are red.
  We also show how to compute $l_{13}$ for $s_1$ and $s_3$ in Figure~\ref{subfig:2d:cand}.
  First of all, the boundary vector $x^*=(0.34,0.97)$
  with $\langle s_1,x^* \rangle = \langle s_3,x^* \rangle$
  is found. Then, the regret ratio of $\{s_1,s_3\}$ over $s_2$ for $x^*$
  is computed as $l_{13}$. Since $l_{13} \approx 0.3 > 0.1$,
  the edge $(s_1 \rightarrow s_3)$ is not added to $G$.
  The graph $G$ constructed for $\varepsilon=0.1$ is shown in Figure~\ref{subfig:2d:cycle}.
  The shortest cycle $C^*$ of $G$ is highlighted in orange and
  $Q^*_{0.1}=\{s_2,s_5,s_7,s_8,s_{10}\}$ is the optimal result of GRMR for $\varepsilon=0.1$.
\end{example}

\subsection{Theoretical Analysis}\label{subsec:analysis:2d}

Next, we will analyze the correctness and time complexity of
Algorithm~\ref{alg:grmr:2d} theoretically. The road map of our analysis is as follows.
Firstly, Lemma~\ref{lm:cand} shows the validity of \emph{candidate selection}.
Secondly, Lemma~\ref{lm:regret} proves the correctness of \emph{regret computation}
and \emph{graph construction}. Thirdly, Lemma~\ref{lm:cycle} states the equivalence between
computing the optimal result of GRMR on $P$ and finding the shortest cycle in $G$.
Combining Lemmas~\ref{lm:cand}--\ref{lm:cycle}, we prove the optimality of
Algorithm~\ref{alg:grmr:2d} in Theorem~\ref{thm:2d:correctness}.
Finally, we analyze the time complexity of Algorithm~\ref{alg:grmr:2d}
in Theorem~\ref{thm:2d:complexity}.

\begin{lemma}\label{lm:cand}
  A point $p \in S$ if and only if $R_{\varepsilon}\neq\varnothing$.
\end{lemma}
\begin{lemma}\label{lm:regret}
  For any $s_i,s_j \in S$ ($i < j$), the maximum regret ratio of $\{s_i,s_j\}$
  after removing all points in $S$ between them is at most $l_{ij}$.
\end{lemma}
\begin{lemma}\label{lm:cycle}
  If $C$ is a cycle of $G$, then $Q = \{q : q \in C\}$ is an $\varepsilon$-regret set of $P$;
  If $Q$ is a locally minimal $\varepsilon$-regret set of $P$, then there exists
  a cycle $C$ of $G$ corresponding to $Q$.
\end{lemma}
Please refer to Appendix~\ref{proof:cand}--\ref{proof:cycle} for the proofs of
Lemma~\ref{lm:cand}-\ref{lm:cycle}.

\begin{theorem}\label{thm:2d:correctness}
  The result $Q^*_{\varepsilon}$ returned by Algorithm~\ref{alg:grmr:2d}
  is optimal for GRMR with parameter $\varepsilon$ on database $P$.
\end{theorem}
\begin{proof}
  Based on Lemma~\ref{lm:cand}, Algorithm~\ref{alg:grmr:2d} excludes all
  redundant points from computation. According to Lemmas~\ref{lm:regret}
  and~\ref{lm:cycle}, it is guaranteed that any locally minimal $\varepsilon$-regret
  set of $P$ forms a cycle in $G$. Therefore, the optimal result $Q^*_{\varepsilon}$ of
  GRMR on $P$, i.e., the smallest $\varepsilon$-regret set of $P$, must correspond to
  the shortest cycle of $G$.
  Hence, we conclude that Algorithm~\ref{alg:grmr:2d} is optimal for GRMR in $\mathbb{R}^2$.
\end{proof}

\begin{theorem}\label{thm:2d:complexity}
  The time complexity of Algorithm~\ref{alg:grmr:2d} is $O(n^3)$.
\end{theorem}
\begin{proof}
  Firstly, the time complexity of candidate selection is $O(|\extremes|+n\cdot\log{|\extremes|})$.
  Here, a binary search can be used to find the index $i$
  and thus it takes $O(\log{|\extremes|})$ time to decide whether $p$ is added to $S$ or not.
  Secondly, the time complexity of graph construction is $O(|S|^2\cdot|\extremes|)$.
  Thirdly, the time complexity of finding the shortest cycle in a directed graph
  is $O(|U|\cdot|A|+|U|^2\cdot\log{|U|})$
  when the variant of Dijkstra's algorithm in~\cite{DBLP:journals/jacm/FredmanT87}
  is used for computing the shortest path from each vertex.
  Recently, an $O(|U|\cdot|A|)$ time algorithm~\cite{DBLP:conf/soda/OrlinS17}
  has also been proposed for finding the shortest directed cycle.
  In the worst case (i.e., $\varepsilon$ is close to $1$),
  since $|S|=|U|=O(n)$ and $|A|=O(|S|^2)$, the time complexity
  of Algorithm~\ref{alg:grmr:2d} is $O(n^3)$.
  Nevertheless, if $\varepsilon$ is far away from $1$,
  it typically holds that $|S|=|U| \ll O(n)$ and $ |A| = O(|U|) $.
  In this case, the time complexity of Algorithm~\ref{alg:grmr:2d}
  is $ O(n\cdot\log{|\extremes|}+|S|^2\cdot|\extremes|) $.
\end{proof}

\section{Heuristic Algorithm in HD}\label{sec:alg:hd}

In the case of $d \geq 3$, \grmr becomes much more challenging
(see Theorem~\ref{thm:np:hardness}).
Let us first highlight some challenges of \grmr in HD.
Firstly, as discussed in Section~\ref{subsec:background},
\grmr can be seen as a geometric set-cover problem.
Solving it directly as such will require the invocation of set operations between Voronoi cells.
However, the geometric shapes of Voronoi cells in high dimensions
are convex cones defined by intersections of half-spaces, and
set operations on Voronoi cells become very inefficient.
Secondly, constructing the \delaunay exactly is also computationally intensive
when $d \geq 3$, as the number of edges in the \delaunay grows exponentially with $d$.
Thirdly, even when the Voronoi cells and \delaunay of a database are given as inputs,
finding the optimal solution of \grmr is still infeasible unless P=NP
due to the combinatorial complexity of geometric covering problems
in two or higher dimensions~\cite{DBLP:journals/ipl/FowlerPT81}.

To devise a practical heuristic that addresses the above issues, we adopt two major
simplifications in relation to \grmr. First, we only consider solutions that are subsets
of the extreme points \extremes instead of the entire database.
Empirically, this does not significantly degrade the solution quality,
as the optimal result of \grmr is often a subset of \extremes.
Secondly, when considering an extreme point $t \in \extremes$ for the solution set,
we restrict ourselves to merely two possibilities: either $t$ is in the solution;
or there exists an extreme point $t'$ in the solution whose $\varepsilon$-approximate Voronoi
cell can fully cover the Voronoi cell of $t$, in which case we say
that $t'$ \emph{dominates} $t$. In other words, we do not consider the case that $R(t)$
is covered by a \emph{union} of the $\varepsilon$-approximate Voronoi cells of two or more points,
because computing the union is very inefficient.
The above simplifications allow us to develop an approach that can be expressed in terms of a graph structure,
to which we refer as the \emph{dominance graph}, because it encapsulates
information about the dominance relationships among vertices.
The resulting approach can be seen as targeting a simplified formulation of
the original \grmr problem.

In what follows, we first describe the dominance graph  and how we obtain a simplified problem
formulation from it in Section~\ref{sec:graph-formulation}.
Subsequently, in Section~\ref{sec:heurhd}, we describe \heuristic, an efficient heuristic
algorithm for \grmr on databases of arbitrary dimensionality, and analyze its properties.
Finally, we discuss some issues on the practical implementation of \heuristic
in Section~\ref{sec:implementation}.

\subsection{Dominance Graph}\label{sec:graph-formulation}

Let the points in \extremes be indexed by $[1,\ldots,|\extremes|]$
as $\{t_1,t_2,\ldots,t_{|\extremes|}\}$. The dominance graph $\domgraph=(V,E)$
is a directed weighted graph where $V$ is the set $\extremes$ of extreme points.
A directed edge $(t_i \rightarrow t_j)$ with an associated
weight $\varepsilon_{ij} \in (0,1)$ exists if and only if the $\varepsilon_{ij}$-approximate
Voronoi cell of $t_i$ fully covers the Voronoi cell of $t_j$.
Therefore, the presence of edge $(t_i \rightarrow t_j)$ signifies that, in any solution
of \grmr with $\varepsilon \geq \varepsilon_{ij}$, $t_i$ can replace $t_j$ without
a violation of the regret constraint.

The edge weight $\varepsilon_{ij}$ for each pair of points $t_i,t_j$ can be computed
from the linear program (LP) in Eq.~\ref{eq:lp:regret}.
\begin{gather}\label{eq:lp:regret}
  \begin{aligned}
  \text{Maximize:}   \quad & 1 - t_i^{\top}x \\
  \text{Subject to:} \quad & (t_j - t)^{\top}x \geq 0, \forall t \in N(t_j) \\
                           & t_j^{\top}x = 1
  \end{aligned}
\end{gather}
In Eq.~\ref{eq:lp:regret}, $N(t)$ is the set of neighbors of $t$ in $\delgraph(P)$,
i.e., the set of extreme points whose Voronoi cells are adjacent to $R(t)$.
The first set of inequality constraints means that the feasible region of the linear program is
the Voronoi cell $R(t_j)$ of $t_j$, which is defined by the intersections of $|N(t_j)|$
closed half-spaces. Each half-space corresponds to the region where the score of $t_j$ is
greater than or equal to that of $t \in N(t_j)$. Hence,
$\langle t_j,x \rangle \geq \langle t,x \rangle \; \Leftrightarrow \; (t_j - t)^{\top} x \geq 0$.
The second constraint normalizes the score of $t_j$ to be $1$.
Under this constraint, for a given vector $x$, the regret ratio of $t_i$ w.r.t.~$t_j$
is equal to $1-t_i^{\top}x$.
The LP in Eq.~\ref{eq:lp:regret} finds the maximum regret ratio $\varepsilon_{ij}$ of $t_i$
w.r.t.~$t_j$ over the feasible region, i.e., $R(t_j)$.

After the dominance graph \domgraph is built, one can use it to compute a
(possibly suboptimal) solution of \grmr. In particular, for a parameter
$\varepsilon \in (0,1)$, an $\varepsilon$-regret set can be obtained as any
subset $S \subseteq \extremes$ of points that satisfies the following condition:
for each $t_j \in X$, either $t_j \in S$ or there exists an
edge $(t_i \rightarrow t_j)$ with $\varepsilon_{ij} \leq \varepsilon$ for some $t_i \in S$.
Let $\domgraph_{\varepsilon}=(V,E_{\varepsilon})$ be the subgraph of $\domgraph$
where $V=X$ and $E_{\varepsilon} \subseteq E$ is a subset of edges
with weights at most $\varepsilon$.
Then, a heuristic solution of \grmr can be obtained by
finding the minimum dominating set of $\domgraph_{\varepsilon}$.

\subsection{The \heuristic Algorithm}\label{sec:heurhd}

The \heuristic algorithm proceeds in two steps as presented in
Algorithms~\ref{alg:builddom} and~\ref{alg:grmr:hd}.

\begin{algorithm}
  \caption{\builddom}\label{alg:builddom}
  \Input{Extreme points $\extremes \subseteq P$, IPDG $\delgraph(P)$, Parameter $\varepsilon$}
  \Output{Dominance graph $\domgraph_{\varepsilon}$}
  Initialize a directed graph $\domgraph_{\varepsilon}=(V,E_{\varepsilon})$
  where $V = \extremes$ and $E_{\varepsilon} = \varnothing$\;\label{ln:hd:graph:s}
  \For{$i \gets 1,\ldots,|\extremes|$}{
    Initialize an empty queue $\mathcal{Q}$\;
    \textbf{foreach} $t \in N(t_i)$ \textbf{do} $\mathcal{Q}$.enqueue$(t)$\;
    Set $t_i$ as \emph{visited}\;
    \While{$\mathcal{Q}$ is not empty}{
      $t_j \gets \mathcal{Q}$.dequeue()\;
      \If{$t_j$ is not visited}{
        Set $t_j$ as \emph{visited}\;
        Solve LP in Eq.~\ref{eq:lp:regret} for $t_i,t_j$ to compute $\varepsilon_{ij}$\;
        \If{$\varepsilon_{ij} \leq \varepsilon$}{
          Add a directed edge $(t_i \rightarrow t_j)$ to $E_{\varepsilon}$\;
          \textbf{foreach} $t \in N(t_j)$ \textbf{do} $\mathcal{Q}$.enqueue$(t)$\;
        }
      }
    }
    \label{ln:hd:graph:t}
  }
  \Return{$\domgraph_{\varepsilon}$}\;
\end{algorithm}

\textbf{Step 1 (Dominance Graph Construction, Algorithm~\ref{alg:builddom}):}
As discussed above, for a parameter $\varepsilon \in (0,1)$,
we only need to build a subgraph $\domgraph_{\varepsilon}$ of the dominance graph $\domgraph$
to compute a solution of \grmr.
The procedure to build $\domgraph_{\varepsilon}$ is described in Algorithm~\ref{alg:builddom}.
It assumes that the \delaunay $\delgraph(P)$ is provided as an input,
and in Section~\ref{sec:implementation} we will discuss a practical alternative
to this assumption. Generally, it performs a breadth-first search (BFS) on $\delgraph(P)$
starting from each vertex $t_i \in \extremes$.
For a starting vertex $t_i$, when the BFS encounters another vertex $t_j$,
the LP in Eq.~\ref{eq:lp:regret} is solved
to compute the weight $\varepsilon_{ij}$ from $t_i$ to $t_j$.
If $\varepsilon_{ij} \leq \varepsilon$, it will add an
edge $(t_i \rightarrow t_j)$ to $E_{\varepsilon}$ and continue to traverse
the neighbors of $t_j$ in $\delgraph(P)$.
Otherwise, the BFS does not expand to the neighbors of $t_j$ anymore.
This is because vertices with higher inner-product similarities (resp.~lower weights)
tend to be closer to each other in the \delaunay and stopping the BFS expansion early
allows us to reduce the number of LP computations with little loss of solution quality.

\begin{algorithm}
  \caption{\heuristic}\label{alg:grmr:hd}
  \Input{Database $P$, Extreme points $X$, IPDG $\delgraph(P)$, Parameter $\varepsilon$}
  \Output{Result $Q_{\varepsilon}$ of GRMR}
  $\domgraph_{\varepsilon}\gets$ \builddom($\extremes,\delgraph(P),\varepsilon$)\;\label{ln:hd:domgraph}
  Initialize $Q_{\varepsilon} \gets \varnothing$ and $\mathcal{U} \gets X$\;
  \label{ln:hd:result:s}
  \For{$i \gets 1,\ldots,|\extremes|$}{
    $Dom(t_i) \gets \{t_i\} \cup \{t_j \in X : (t_i \rightarrow t_j) \in E_{\varepsilon}\}$\;
  }
  \While{$\mathcal{U} \neq \varnothing$}{
    $t^* \gets \arg\max_{t_i \in X \setminus Q_{\varepsilon}} |Dom(t_i) \cap \mathcal{U}|$\;
    $Q_{\varepsilon} \gets Q_{\varepsilon} \cup \{t^*\}$
    and $\mathcal{U} \gets \mathcal{U} \setminus Dom(t^*)$\;
  }
  \Return{$Q_{\varepsilon}$}\;\label{ln:hd:result:t}
\end{algorithm}

\textbf{Step 2 (Result Computation, Algorithm~\ref{alg:grmr:hd}):}
After building $\domgraph_{\varepsilon}$ using Algorithm~\ref{alg:builddom},
\heuristic approaches the dominating set problem on $\domgraph_{\varepsilon}$
as an equivalent set-cover problem on a set system
$\Sigma=(\mathcal{U},\mathcal{S})$ where $\mathcal{U}$ is equal to $X$ and
$\mathcal{S}$ is a collection of sets, with the $i$\textsuperscript{th} set
equal to all points $Dom(t_i)$ dominated by $t_i$, i.e.,
$Dom(t_i) \coloneqq \{t_i\}\cup\{t_j: (t_i\rightarrow t_j) \in \domgraph_{\varepsilon}\}$.
Then, it runs the greedy algorithm to compute an approximate set-cover solution on $\Sigma$. 
Specifically, starting from $Q_{\varepsilon}=\varnothing$, it adds a vertex whose dominating
set contains the most number of uncovered vertices at each iteration until all vertices
in $\mathcal{U}$ are covered. 
Finally, $Q_{\varepsilon}$ is returned as the solution of \grmr on database $P$.

\begin{example}
Figure~\ref{subfig:hd:graph} illustrates the dominance graph $\domgraph$ of the dataset
in Figure~\ref{fig:example}. Then, in Figure~\ref{subfig:hd:result}, we show how $\domgraph$
can be used to compute the solution of \grmr for $\varepsilon=0.2$.
Specifically, $\domgraph_{0.2}$ is a subgraph of $\domgraph$ where only the edges
with weights at most $0.2$ are preserved (deleted edges are in gray).
Then, \heuristic runs the greedy algorithm on $\domgraph_{0.2}$ to compute
the smallest dominating set.
At the first iteration, $t_8$ is added to $Q_{0.2}$ because $|Dom(t_8)|=3$ is
the maximum among all vertices. Then, $t_4$ is added at the second iteration.
Next, $t_2,t_5,t_6$ are added accordingly. 
Finally, the dominating set $Q=\{t_2,t_4,t_5,t_6,t_8\}$ of $\domgraph_{0.2}$
provides a heuristic solution of \grmr for $\varepsilon=0.2$.
\end{example}

\textbf{Theoretical Analysis:}
The result $Q_{\varepsilon}$ returned by Algorithm~\ref{alg:grmr:hd} is guaranteed not to break the regret constraint of \grmr, even if it may not be the smallest possible solution.
\begin{theorem}\label{thm:grmr:hd}
  It holds that $Q_{\varepsilon}$ returned by Algorithm~\ref{alg:grmr:hd} is
  an $\varepsilon$-regret set of $P$, i.e., $l(Q_{\varepsilon}) \leq \varepsilon$.
\end{theorem}
\begin{proof}
  For any vector $x \in \mathbb{S}^{d-1}$, there exists a point $t_j$ such that $x \in R(t_j)$.
  Since $Q_{\varepsilon}$ is a dominating set of $G_{\varepsilon}$, we have either
  $t_j \in Q_{\varepsilon}$ or there exists an edge $(t_i \rightarrow t_j) \in E_{\varepsilon}$
  for some $t_i \in Q_{\varepsilon}$. In the previous case, we have
  $l_x(Q_{\varepsilon})=0$; In the latter case, we have
  $l_x(Q_{\varepsilon}) \leq \varepsilon_{ij} \leq \varepsilon$.
  In both cases, we have $l(Q_{\varepsilon}) \leq \varepsilon$.
\end{proof}

\begin{figure}
  \centering
  \subfigure[Dominance graph $\domgraph$]{
    \label{subfig:hd:graph}
    \includegraphics[width=0.25\textwidth]{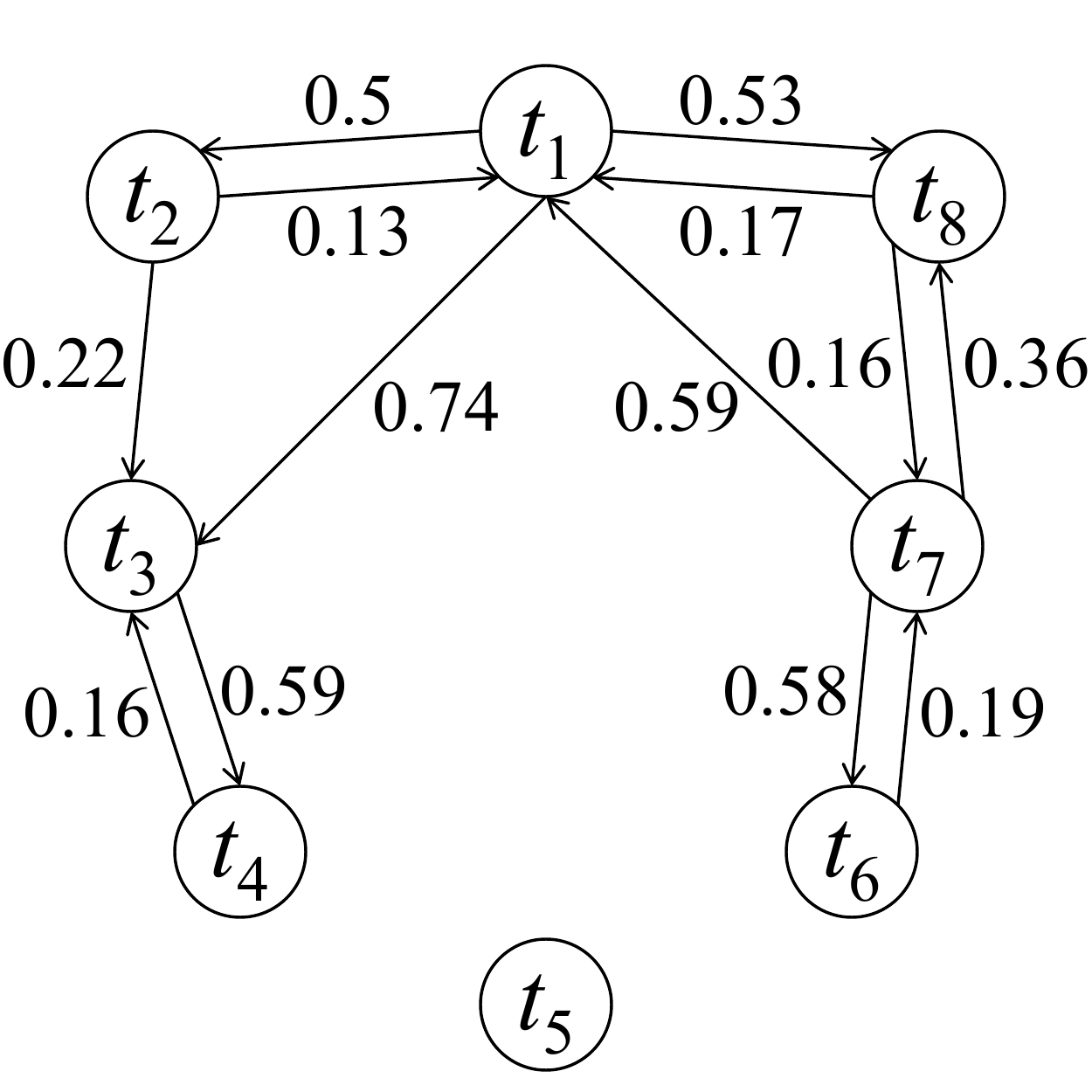}
  }
  \subfigure[Subgraph $\domgraph_{0.2}$ for $\varepsilon=0.2$]{
    \label{subfig:hd:result}
    \includegraphics[width=0.25\textwidth]{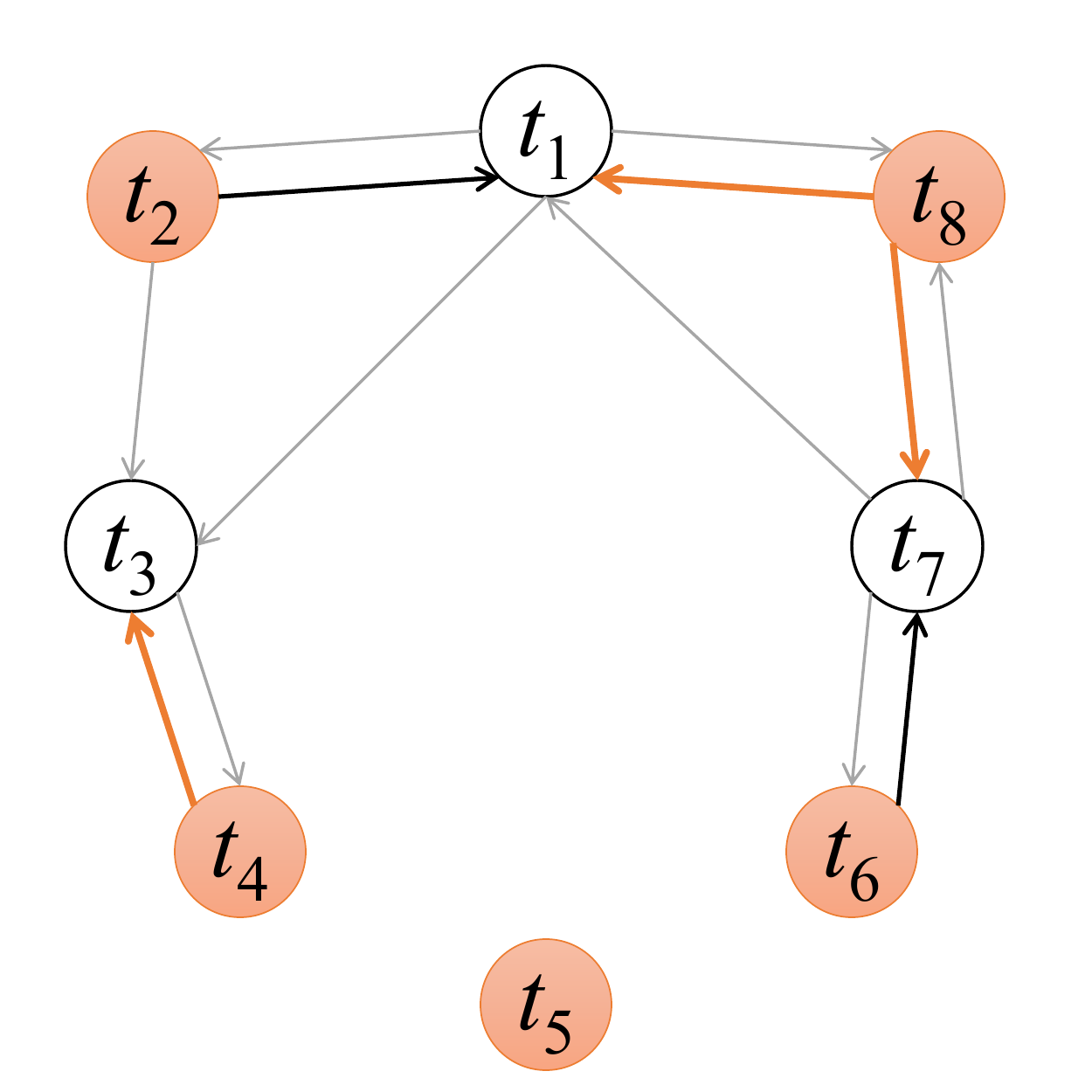}
  }
  \caption{Example of dominance graph and \heuristic}
  \label{fig:hd:example}
\end{figure}

How large is the size of the solution of \heuristic compared to the size of the optimal solution
for \grmr? While the approximation factor of \heuristic on the solution size is still
an open problem, we do have an upper bound for it.
\begin{theorem}
  If $Q_{\varepsilon}^*$ is the optimal solution of \grmr with parameter $\varepsilon$
  and $Q_{\varepsilon}$ is the solution of \grmr returned by \heuristic, then it holds
  that $|Q_{\varepsilon}| \leq \frac{|\extremes|}{d+1} \cdot |Q_{\varepsilon}^*|$.
\end{theorem}
\begin{proof}
First, the size of $Q_{\varepsilon}$ is at most $|\extremes|$ since $Q_{\varepsilon} \subseteq X$.
Second, $Q^*_{\varepsilon}$ must contain at least $d+1$ points to guarantee
$l(Q^*_{\varepsilon}) < 1$. This is because, for any point set $Q$ of size $d$
in $\mathbb{R}^d$, we can find a vector $x$ perpendicular to the hyperplane
containing all points in $Q$ such that $l_x(Q) = 1$.
Thus, $|Q_{\varepsilon}| \leq \frac{|\extremes|}{d+1}\cdot|Q^*_{\varepsilon}|$.
\end{proof}

Next, we prove that \heuristic is a polynomial-time algorithm.
\begin{theorem}\label{thm:grmr:hd:time}
  The time complexity of \heuristic is
  $O(|\extremes|^2 \cdot (\Delta \cdot d^{3.5} + \Gamma))$ where
  $\Delta = \max_{t \in \extremes}|N(t)|$ and $\Gamma = \max_{t \in \extremes}Dom(t)$.
\end{theorem}
\begin{proof}
  First of all, the number of LPs solved by \builddom is $O(|\extremes|^2)$, as the worst-case corresponds to computing the weight for each pair of vertices.
  Each LP in Eq.~\ref{eq:lp:regret} has $|N(t_j)|+1$ constraints and $d$ variables. 
  When the interior point method is used as the LP solver, a worst-case time complexity
  of $O(\Delta \cdot d^{3.5})$ can always be guaranteed.
  Therefore, the time complexity of \builddom is
  $O(|\extremes|^2 \cdot \Delta \cdot d^{3.5})$.
  Then, the time to build the set system $\Sigma$ is $O(|\extremes| \cdot \Gamma)$
  where $\Gamma = \max_{t \in \extremes} Dom(t)$.
  The greedy algorithm should evaluate the union of $Dom(t)$ and $\mathcal{U}$
  for each $t \in X$ at each iteration. Thus, the running time of each iteration is $O(|\extremes| \cdot \Gamma)$. Then, it runs $|Q_{\varepsilon}|=O(|\extremes|)$ iterations.
  Hence, the time complexity of result computation is $O(|\extremes|^2 \cdot \Gamma)$.
  We conclude the proof by summing up both results.
\end{proof}

\subsection{Practical Implementation}\label{sec:implementation}

Next, we discuss practical aspects of the implementation of \heuristic.
The first one is how to build an approximate \delaunay $\widehat{\delgraph}(P)$.
The second one is how to reuse the dominance graph for
achieving better solution quality.

\textbf{\delaunay Construction:}
A na\"ive method to construct an exact \delaunay $\delgraph(P)$ for a point
set $P \subset \mathbb{R}^d$ is to compute $\mathcal{CH}(P)$
and enumerate all edges of $\mathcal{CH}(P)$. 
Then, an edge of $\mathcal{CH}(P)$ corresponds to an edge
of $\delgraph(P)$~\cite{DBLP:conf/nips/ZhouTX019}.
Although this method is practically efficient in $\mathbb{R}^2$ or even $\mathbb{R}^3$,
it quickly becomes computationally intensive for higher dimensionality.
One alternative approach we considered was based on existing works on graph-based
\emph{maximum inner product search}~\cite{DBLP:conf/nips/MorozovB18,DBLP:conf/emnlp/TanZXL19,DBLP:conf/nips/ZhouTX019},
which propose to build a graph analogous to \delaunay.
However, since graphs built by these methods are significantly
different from the original IPDG, e.g., they are directed graphs and
contain non-extreme points, it is not reasonable to use them for our problems directly.

\begin{algorithm}
  \caption{IPDG Construction}\label{alg:ipdg}
  \Input{Extreme points $X$, Sample size $m$, Parameter $k$}
  \Output{Approximate IPDG $\widehat{\delgraph}(P)$}
  Initialize an undirected graph $G=(V,E)$ where $V=X$ and $E=\varnothing$\;
  \For{$i \gets 1,\ldots,m$}{
    Draw a random vector $x_i$ uniformly from $\mathbb{S}^{d-1}$\;
    Compute the top-$k$ results $\Phi_{k}(x_i,X)$ for $x_i$ in $X$\;
    Let $t^*$ be the top-ranked tuple for $x_i$ in $X$\;
    \ForEach{$t \in \Phi_{k}(x_i,X)$ and $t \neq t^*$}{
      Add an undirected edge $\{t^*,t\}$ to $E$\;
    }
  }
  \Return{$\widehat{\delgraph}(P) \gets G$}\;
\end{algorithm}

Instead, we propose a new method to build an approximate
\delaunay as described in Algorithm~\ref{alg:ipdg}.
Our method is based on an intuitive observation:
if the Voronoi cells of two points $t,t^\prime$ are adjacent,
then there is some vector in $R(t)$ for which $t^\prime$ is high-ranked,
i.e., $t^\prime$ is among the top-$k$ results for a small $k$ -- and vice versa.
We use this observation in Algorithm~\ref{alg:ipdg} as follows.
First, we draw $m$ random vectors from $\mathbb{S}^{d-1}$.
Then, we compute the top-$k$ results $\Phi_{k}(x_i,X)$ for each sampled vector $x_i$.
Subsequently, we identify the top-ranked point $t^*$ from $\Phi_{k}(x_i,X)$
and add the edges between $t^*$ and the remaining points in $\Phi_{k}(x_i,X)$.
Finally, the resultant graph after processing
all sampled vectors are returned as $\widehat{\delgraph}(P)$.

We note that using an approximate \delaunay $\widehat{\delgraph}(P)$ instead of the exact one
$\delgraph(P)$ does not affect the correctness of \heuristic, i.e.,
the solution returned by \heuristic is still an $\varepsilon$-regret set.
Compared with $\delgraph(P)$, $\widehat{\delgraph}(P)$ may both
contain some additional edges and miss some existing edges.
On the one hand, an additional edge has no effect on the result of
the LP in Eq.~\ref{eq:lp:regret}. This is because its feasible region is exactly
the Voronoi cell $R(t_j)$ of $t_j$.
Hence, an additional edge leads to a redundant constraint that does not reduce
the feasible region at all. On the other hand, a missing edge may cause that the result
the LP in Eq.~\ref{eq:lp:regret} is larger than the optimal maximum regret ratio,
because the feasible region is larger than $R(t_j)$.
As a result, $\domgraph_{\varepsilon}$ may contain fewer edges
and the solution $Q_{\varepsilon}$ may have more points.
Nevertheless, $Q_{\varepsilon}$ is still guaranteed to be an $\varepsilon$-regret set.
Therefore, the values of $k$ and $m$ can affect the performance of \heuristic
by controlling the number of edges in $\widehat{\delgraph}(P)$.
When larger values of $k$ and $m$ are used, $\widehat{\delgraph}(P)$ will contain more edges
in $\delgraph(P)$ and $Q_{\varepsilon}$ will be smaller.
At the same time, larger values of $k$ and $m$ will also lead to more additional edges
in $\widehat{\delgraph}(P)$ and thus a lower efficiency of dominance graph construction.

\textbf{Graph Reuse:}
In \heuristic, we build a graph $\domgraph_{\varepsilon}$
that only contains edges with weights at most $\varepsilon$ for computation.
In fact, it is possible to use $\domgraph_{\varepsilon}$ for GRMR
with any $\varepsilon' \in (0,\varepsilon)$ since one can extract
a subgraph $\domgraph_{\varepsilon'}$ from $\domgraph_{\varepsilon}$ by deleting the edges
with weights greater than $\varepsilon'$.
Then, the dominating set of $\domgraph_{\varepsilon'}$ is also
the result $Q_{\varepsilon'}$ for GRMR with parameter $\varepsilon'$.
In light of this observation, we devise an approach to improving the solution quality
of \heuristic.

As discussed already, \heuristic returns an $\varepsilon$-regret set
that may not have the minimum size.
To explore solutions of smaller size than the one returned by \heuristic, we invoke it with a
larger parameter, and by doing so, we reuse previously materialized instances of \domgraph
as much as possible. Specifically, given a parameter $\varepsilon \in (0,1)$, we first build
a graph $\domgraph_{\eta}$ for some $\eta > \varepsilon$.
Then, we perform a binary search on $\delta$ in the range $[\varepsilon,\eta]$,
and for each value of $\delta$, obtain a solution $Q_{\delta}$ by invoking \heuristic
with parameter $\delta$, while reusing $\domgraph_{\eta}$ as described above.
The goal of the binary search is to find the maximum value of $\delta$ that satisfies the regret constraint of \grmr, i.e., $l(Q_{\delta}) \leq \varepsilon$,
and return $Q_{\delta}$ for GRMR with parameter $\varepsilon$.
In practice, we observe that the size of $Q_{\delta}$
is smaller than the size of $Q_{\varepsilon}$ we would obtain by a single invocation
of \heuristic for the input parameter.
In practice, for small values of $\varepsilon$, setting $\eta=3\cdot\varepsilon$
is good enough in almost all cases.
With the incorporated binary search, the time complexity of \heuristic increases to
$O(|\extremes|^2 \cdot (\Delta \cdot d^{3.5} + \Gamma\cdot\log{\frac{1}{\varepsilon}}))$,
since the result computation procedure is repeated $O(\log{\frac{1}{\varepsilon}})$ times
for the binary search on the value of $\delta$.

Similar approach can also be used for the dual formulation of \grmr in
Section~\ref{subsec:def}.
By building $\domgraph$ and performing a binary search on $\delta \in (0,1)$,
we can find the minimum value of $\delta$ that guarantees $|Q_{\delta}| \leq r$
and return $Q_{\delta}$ as the result for the dual problem.

\section{Experiments}\label{sec:exp}

In this section, we evaluate the performance of our algorithms on real and synthetic datasets.
We first introduce our experimental setup in Section~\ref{subsec:setup}.
Then, we present the experimental results on two-dimensional datasets
in Section~\ref{subsec:result:2d}. Finally, the experimental results on high-dimensional datasets
are reported in Section~\ref{subsec:result:hd}.

\subsection{Experimental Setup}\label{subsec:setup}

\textbf{Implementation:}
We conduct all experiments on a server running Ubuntu 18.04.1
with a 2.3GHz processor and 256GB memory. All algorithms are
implemented in C++11. We use GLPK as the LP solver.
The implementation is published
on GitHub\footnote{\url{https://github.com/yhwang1990/Generalized-RMS}}.

\begin{table}
  \centering
  \caption{Statistics of Real Datasets}
  \label{tbl:datasets}
  \begin{tabular}{cccc}
  \toprule
    Dataset & Size & Dimension & Source \\
  \midrule
    \textsc{Airline} & 1,700,782 & 2 & DOT \\
    \textsc{NBA}     & 24,585    & 2 & Kaggle \\
  \midrule
    \textsc{Climate}   & 566,262   & 6 & CRU \\
    \textsc{El Nino}   & 178,080   & 5 & UCI \\
    \textsc{Household} & 2,049,280 & 7 & UCI \\
    \textsc{SUSY}      & 5,000,000 & 8 & UCI \\
  \bottomrule
  \end{tabular}
\end{table} 

\textbf{Real Datasets:}
We use six publicly available real datasets for evaluation.
Basic statistics of these datasets are listed in Table~\ref{tbl:datasets}.
\begin{itemize}
  \item \textsc{Airline}\footnote{\url{https://www.transtats.bts.gov/DL_SelectFields.asp?Table_ID=236&DB_Short_Name=On-Time}}:
  It records the information of all flights conducted by US carriers
  from January 2019 to March 2019. We used two attributes, i.e.,
  \emph{arrival delay} and \emph{air-time}, for evaluation.
  \item \textsc{NBA}\footnote{\url{https://www.kaggle.com/drgilermo/nba-players-stats}}:
  It contains the statistics of NBA players aggregated by season/team.
  We used two attributes, i.e., \emph{offensive win shares} (\emph{ows}) and \emph{defensive win shares} (\emph{dws}), for evaluation.
  \item \textsc{Climate}\footnote{\url{http://www.cru.uea.ac.uk/data}}:
  It contains the climate information of different locations.
  We used six attributes in our experiments, including (annual average) \emph{temperature} and \emph{humidity}.
  \item \textsc{El Nino}\footnote{\url{https://archive.ics.uci.edu/ml/datasets/El+Nino}}:
  It contains oceanographic data in the Pacific ocean. We used all five numerical attributes for evaluation.
  \item \textsc{Household}\footnote{\url{https://archive.ics.uci.edu/ml/datasets/Individual+household+electric+power+consumption}}:
  It contains the electric power consumption measurements gathered in a household.
  We used all seven numerical attributes for evaluation.
  \item \textsc{SUSY}\footnote{\url{https://archive.ics.uci.edu/ml/datasets/SUSY}}: A high-energy physics dataset that contains eight numerical features, representing kinematic properties measured by particle detectors.
\end{itemize}
We normalized all attribute values of each dataset to the range $[-1,1]$
in the preprocessing step.

\textbf{Synthetic Datasets:}
To evaluate the performance of different algorithms in controlled settings,
we use two synthetic datasets, namely \textsc{Normal} and \textsc{Uniform},
in our experiments. In the \textsc{Normal} dataset, each attribute is independently drawn
from the standard normal distribution $\mathcal{N}(0,1)$.
In the \textsc{Uniform} dataset, each attribute is independently drawn
from a uniform distribution $\mathcal{U}(-1,1)$.
For both datasets, we vary the number of tuples $n$ from $10^4$ ($10$k) to $10^7$ ($10$m)
and the dimensionality $d$ from $2$ to $10$
for testing the performance of different algorithms with varying $n$ and $d$. 
By default, we use the dataset with $n=10^6$ ($1$m) and $d=6$.

\textbf{Algorithms:}
We compare the following eight algorithms for \grmr in our experiments.
\begin{itemize}
  \item \textsc{$\varepsilon$-Kernel}~\cite{DBLP:journals/jacm/AgarwalHV04}:
  We use the ANN-based algorithm to compute
  an $\varepsilon$-kernel as the result of \grmr. 
  The approach is a straightforward adaptation of the analysis
  in~\cite{DBLP:journals/jacm/AgarwalHV04} and
  returns an $\varepsilon$-regret set of size $O(\frac{1}{\varepsilon^{(d-1)/2}})$ for \grmr.
  \item \textsc{HittingSet}~\cite{DBLP:conf/wea/AgarwalKSS17}
  for RMS can be used for \grmr by extending the sample space of ranking functions from
  monotonic linear to general linear functions
  (i.e., sampling vectors from $\mathbb{S}^{d-1}$ instead of $\mathbb{S}^{d-1}_{+}$)
  based on the $\varepsilon$-net property.
  It returns an $\varepsilon$-regret set of
  size $O(|Q^{*}_{\varepsilon}| \cdot \log{|Q^{*}_{\varepsilon}|})$
  where $Q^{*}_{\varepsilon}$ is the optimal solution of \grmr.
  \item \textsc{2D-RRMS}~\cite{DBLP:conf/sigmod/AsudehN0D17} is
  an exact algorithm for RMS in $\mathbb{R}^2$.
  \item \textsc{Greedy}~\cite{DBLP:journals/pvldb/NanongkaiSLLX10},
  \textsc{HD-RRMS}~\cite{DBLP:conf/sigmod/AsudehN0D17},
  and \textsc{Sphere}~\cite{DBLP:conf/sigmod/XieW0LL18} are typical
  (approximate/heuristic) algorithms for RMS in $\mathbb{R}^d$ when $d \geq 3$.
  \item \exact is our exact algorithm for \grmr
  in $\mathbb{R}^2$ (Section~\ref{sec:alg:2d}).
  \item \heuristic is our heuristic algorithm for \grmr
  in $\mathbb{R}^d$ (Section~\ref{sec:alg:hd}).
\end{itemize}

Theoretically, the sampling complexities of \textsc{$\varepsilon$-Kernel} and \textsc{HittingSet}
are $O(\frac{1}{\varepsilon^{(d-1)/2}})$ and $O(\frac{1}{\varepsilon^{d-1}})$, respectively,
which makes them impractical in high dimensions.
In our implementation, rather than performing the sampling all at once, we sample them in stages and maintain a result $Q$ based on the current samples until $l(Q) \leq \varepsilon$.
Moreover, the RMS algorithms, i.e., \textsc{2D-RRMS}, \textsc{Greedy}, \textsc{HD-RRMS},
and \textsc{Sphere}, cannot be directly used for \grmr because they only consider monotonic
linear functions, but not ones with negative weights. 
In practice, we cast a \grmr problem into $2^d$ RMS problems
by dividing the points in a dataset into the standard $2^d$ orthants.
We then run an RMS algorithm on each partition (resp.~each orthant)
and return the union of results for GRMR.
Note that this method has no theoretical guarantees, and empirically it often fails to provide a valid $\varepsilon$-regret set when the points are not evenly distributed among orthants. 
Nevertheless, we could not find in the literature any other approach to adapting them to GRMR.
Finally, the set of extreme points (or the skylines of all partitions for RMS algorithms)
on each dataset is precomputed and provided as an input to all algorithms.

\textbf{Performance Measures:}
The efficiency of each algorithm is measured by \emph{running time}, i.e.,
the CPU time to compute the result of \grmr for a given $\varepsilon$.
The quality of a result $Q$ is measured by the maximum regret ratio $l(Q)$,
as well as \emph{its size} $|Q|$ for a given $\varepsilon$.
Among all solutions satisfying $l(Q) \leq \varepsilon$,
the one with a smaller size is considered better.
To estimate the maximum regret ratio, we draw a set of $10^6$ random vectors from $\mathbb{S}^{d-1}$,
compute the regret ratio of $Q$ for each vector, and use the maximum one
as an estimate for $l(Q)$. For a given dataset, we run each algorithm $10$ times and
take the averages of these measures for evaluation.

\subsection{Results on Two-Dimensional Datasets}\label{subsec:result:2d}

In this subsection, we focus on the two-dimensional datasets and compare the performance of
our proposed algorithms \exact and \heuristic with algorithms tailored to two-dimensional
settings, i.e., \textsc{$\varepsilon$-Kernel}, \textsc{HittingSet}, and \textsc{2D-RRMS}.
We present the results for two real datasets (\textsc{Airline} and \textsc{NBA})
and one synthetic dataset (\textsc{Normal}). Since the results for \textsc{Uniform}
are similar to those for \textsc{Normal} in 2D, we omit them due to space limitations.

\begin{figure}[t]
  \centering
  \includegraphics[width=0.6\textwidth]{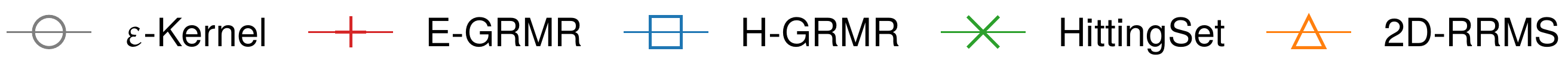}
  \\
  \subfigure[Airline]{
    \label{subfig:al:loss}
    \includegraphics[width=0.25\textwidth]{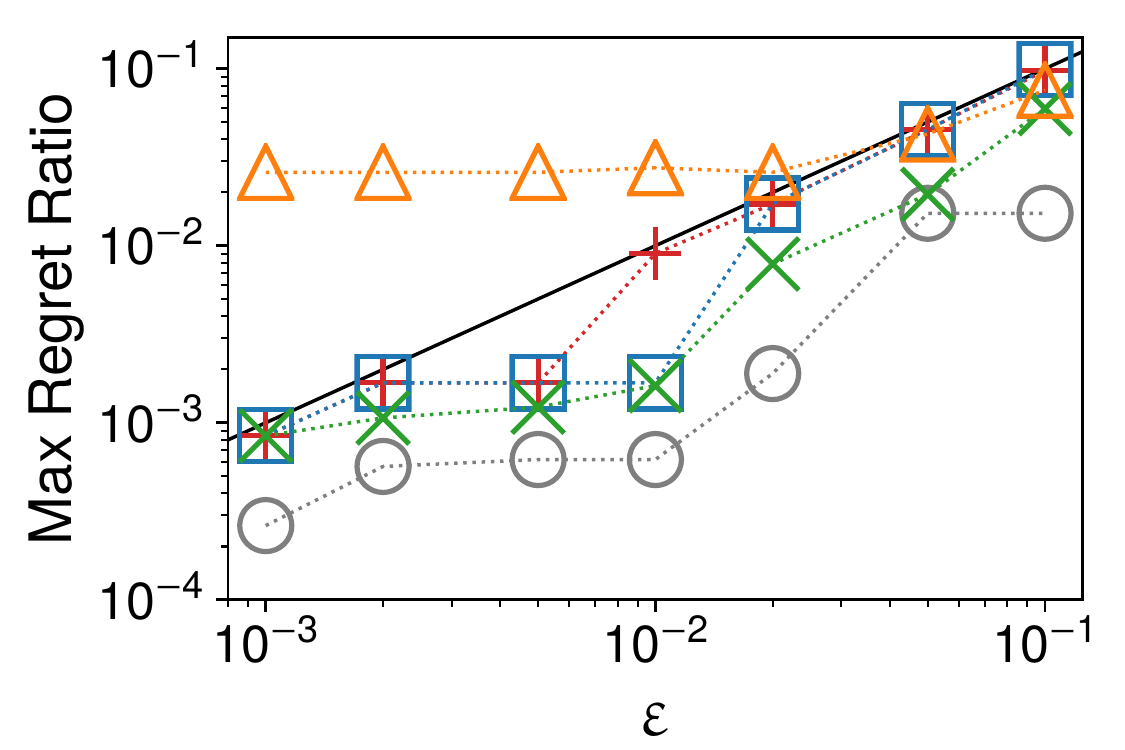}
  }
  \subfigure[NBA]{
    \label{subfig:nba:loss}
    \includegraphics[width=0.25\textwidth]{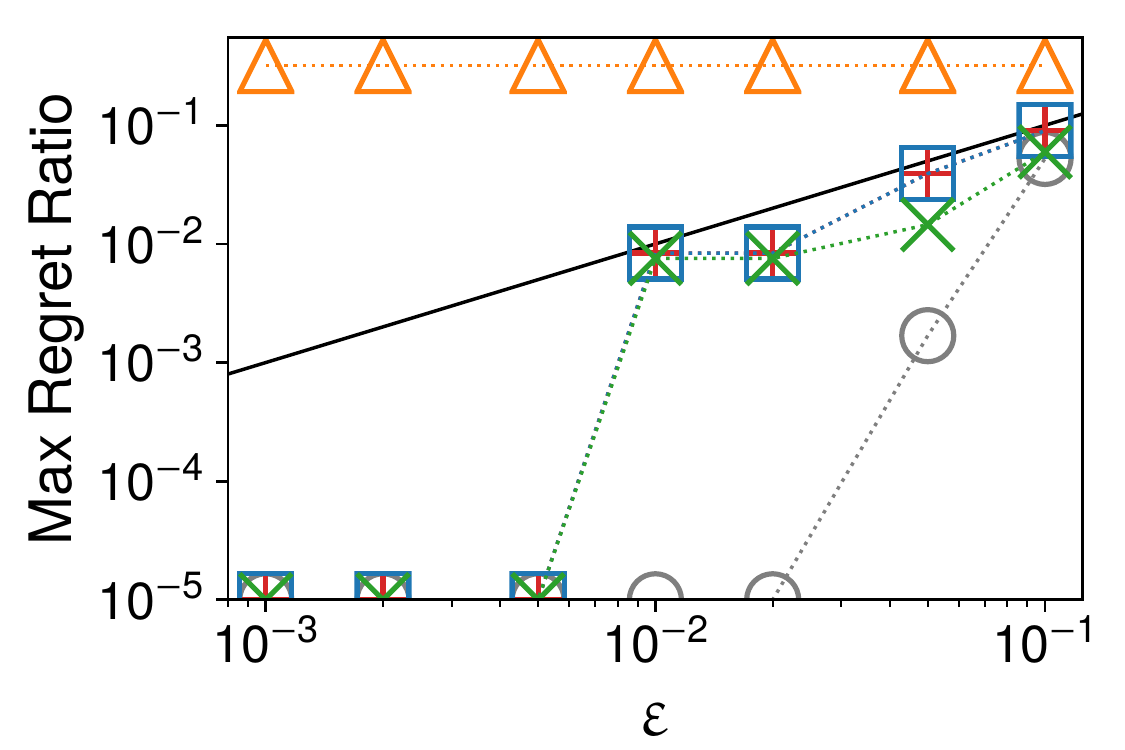}
  }
  \caption{Parameter $\varepsilon$ vs.~max regret ratio}\label{fig:loss:2d}
  \vspace{1em}
  \subfigure[Airline]{
    \label{subfig:al:time}
    \includegraphics[width=0.25\textwidth]{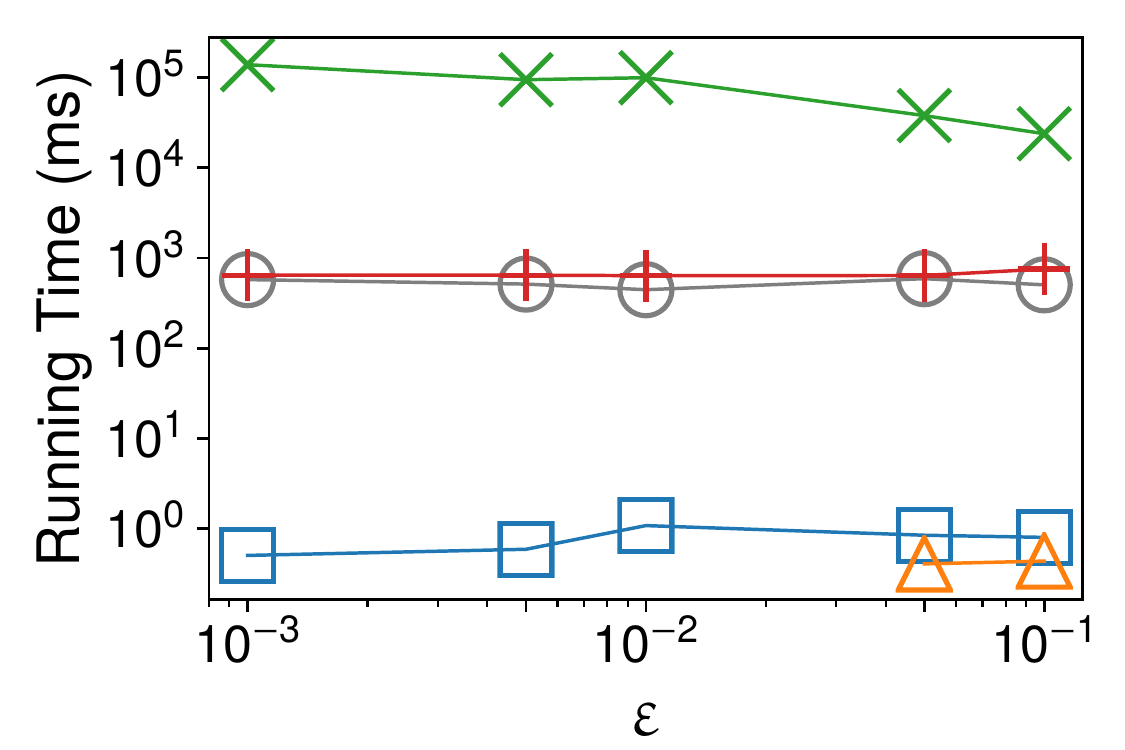}
  }
  \subfigure[NBA]{
    \label{subfig:nba:time}
    \includegraphics[width=0.25\textwidth]{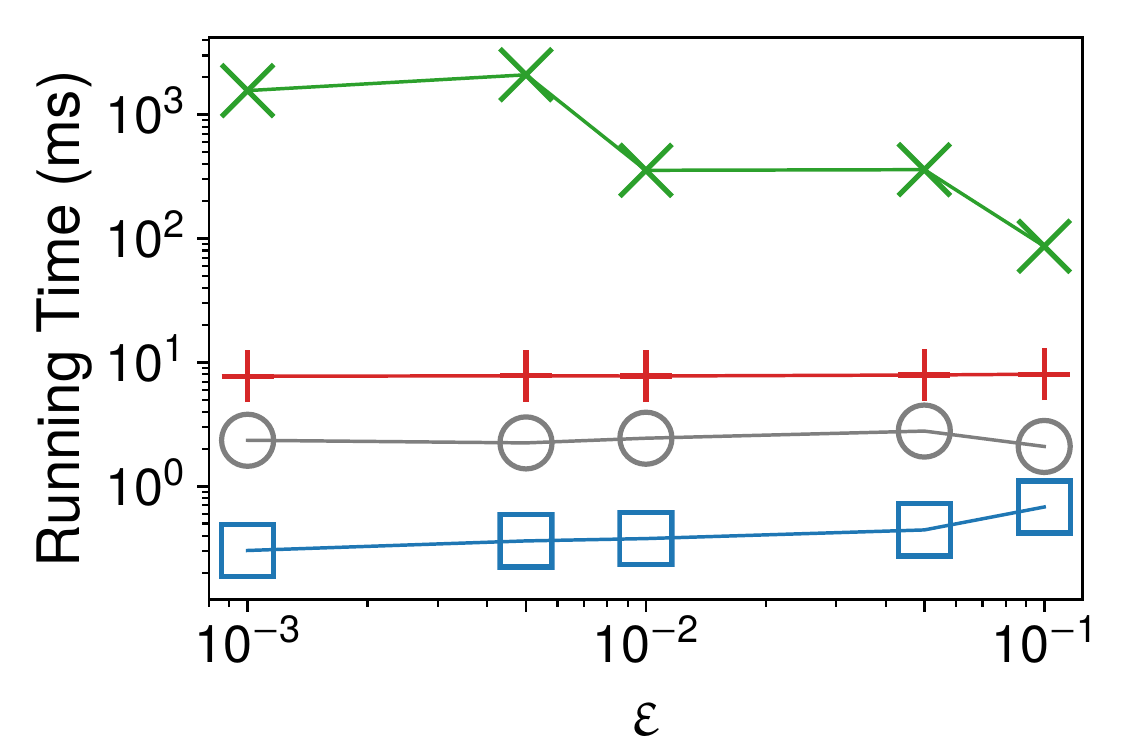}
  }
  \caption{Running time with varying $\varepsilon$}\label{fig:time:2d}
  \vspace{1em}
  \subfigure[Airline]{
    \label{subfig:al:size}
    \includegraphics[width=0.25\textwidth]{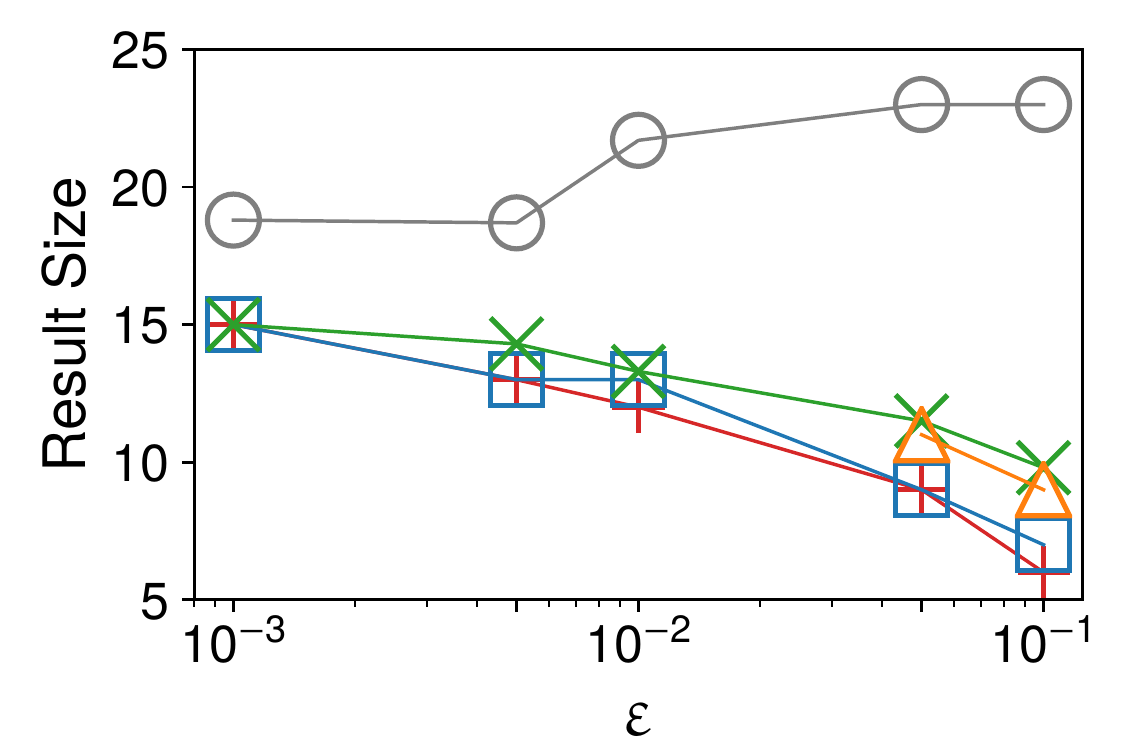}
  }
  \subfigure[NBA]{
    \label{subfig:nba:size}
    \includegraphics[width=0.25\textwidth]{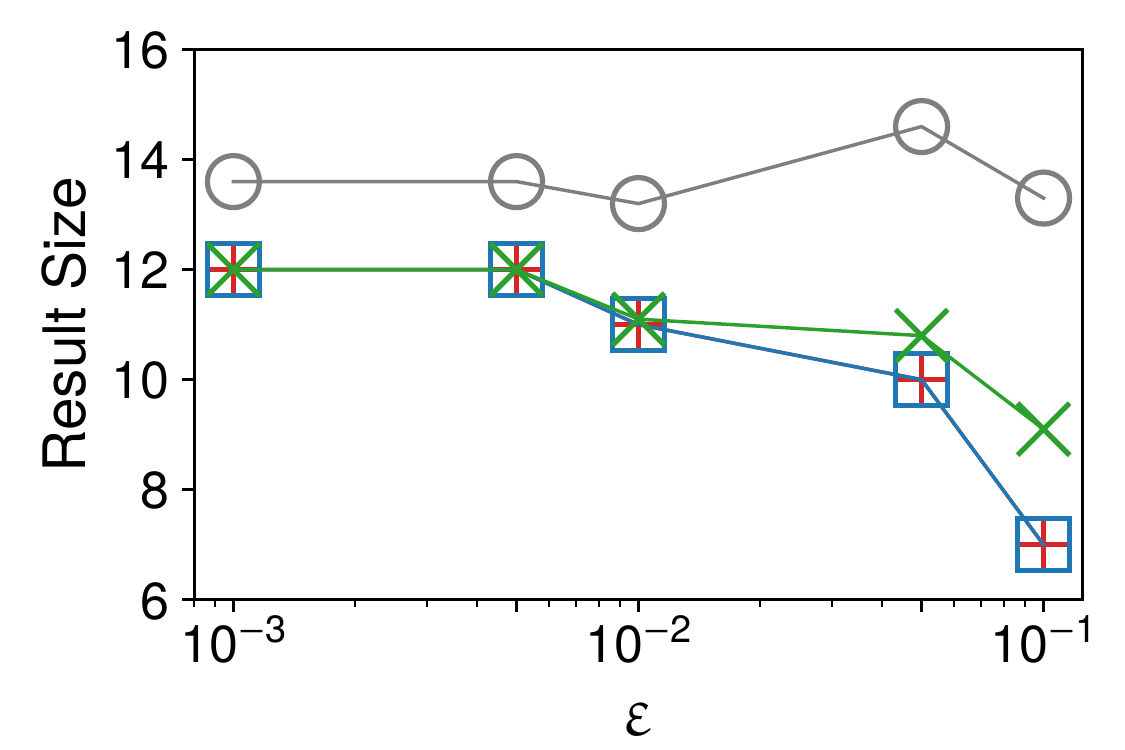}
  }
  \caption{Result sizes with varying $\varepsilon$}\label{fig:size:2d}
\end{figure}

\begin{figure}[t]
  \centering
  \includegraphics[width=0.6\textwidth]{legend-2d.pdf}
  \\
  \subfigure[Normal (2D)]{
    \label{subfig:n:2d:time}
    \includegraphics[width=0.25\textwidth]{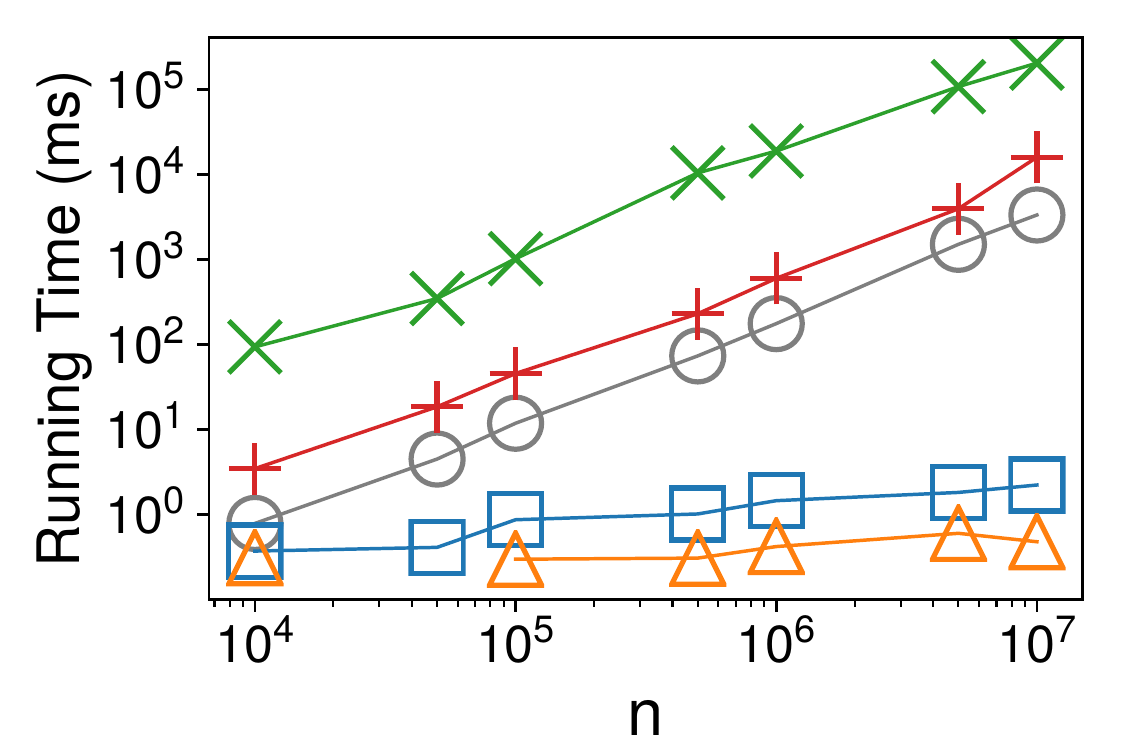}
  }
  \subfigure[Normal (2D)]{
    \label{subfig:n:2d:size}
    \includegraphics[width=0.25\textwidth]{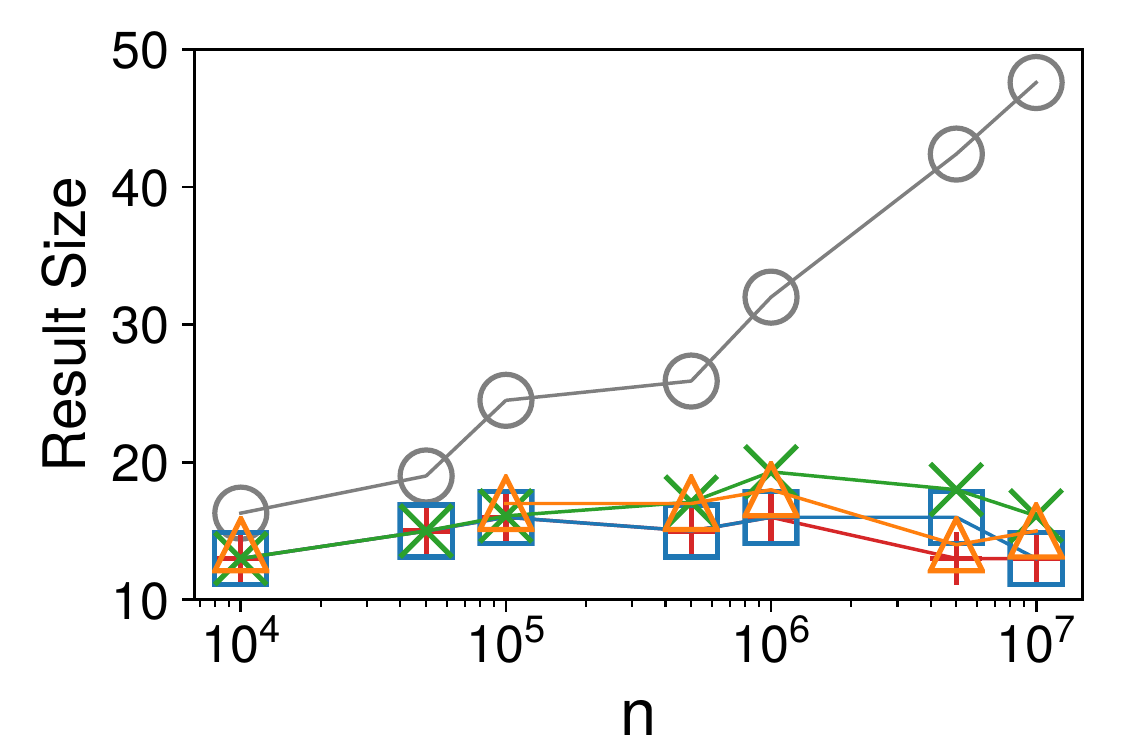}
  }
  \caption{Performance with varying $n$ ($\varepsilon=0.01$)}\label{fig:n:2d}
\end{figure}

\textbf{Impact of Parameter $\varepsilon$:}
We vary $\varepsilon$ from $0.001$ to $0.1$ to evaluate the effect of $\varepsilon$
on the performance of each algorithm.
Figure~\ref{fig:loss:2d} illustrates the maximum regret ratio $l(Q)$ of the result $Q$
of each algorithm for different $\varepsilon$ on \textsc{Airline} and \textsc{NBA}.
Note that \grmr requires that $l(Q) \leq \varepsilon$, and so valid results should
map below the line $l(Q) = \varepsilon$ in Figure~\ref{fig:loss:2d}.
We observe that all algorithms except 2D-RRMS can always guarantee to provide
an $\varepsilon$-regret set.
2D-RRMS fails to provide valid results in most cases because casting GRMR in 2D
into $4$ RMS subproblems does not work well in the case that
the data points are not evenly distributed among orthants.

The running time and result sizes of each algorithm with varying $\varepsilon$
are presented in Figures~\ref{fig:time:2d} and~\ref{fig:size:2d}, respectively.
First of all, the running time of \exact is stable with increasing $\varepsilon$.
This is because \exact spends over $95\%$ of CPU time on
\emph{candidate selection}, whose time complexity is independent of $\varepsilon$.
Since the size of the candidate set is much smaller than the size of the dataset,
the time for \emph{graph construction} and \emph{result computation}
is nearly negligible compared with that for \emph{candidate selection}.
Then, as expected, the result size of \exact decreases with $\varepsilon$
and is always the smallest (optimal) one among all algorithms.
However, \exact only outperforms \textsc{HittingSet} in terms of efficiency.
On the other hand, \heuristic runs one to four orders of magnitude faster than
all other algorithms (except 2D-RRMS that cannot provide valid results in most cases).
At the same time, we observe empirically that \heuristic often provides the optimal or near-optimal results for \grmr, because of the optimality of Delaunay graphs in $\mathbb{R}^2$
and the effectiveness of the dominance graph for \grmr.

\textbf{Impact of Dataset Size $n$:}
Subsequently, we fix $\varepsilon = 0.01$ and vary the size $n$ of \textsc{Normal}
from $10^4$ to $10^7$. The performance of each algorithm is shown in Figure~\ref{fig:n:2d}.
First of all, the running time of \textsc{$\varepsilon$-Kernel}, \exact, and \textsc{HittingSet}
increases linearly with $n$ since their time complexities are all linear with $n$.
2D-RRMS and \heuristic show much better scalability w.r.t.~$n$
because they only use the skyline and extreme points for computation,
whose sizes are much smaller than $n$ and increase sub-linearly with $n$.
Furthermore, the solution quality of \textsc{$\varepsilon$-Kernel} is
significantly inferior to all other algorithms, especially when
$n$ is larger, because the ANN-based method
does not consider whether the $\varepsilon$-kernel is the smallest or not.
Conversely, our proposed algorithms as well as \textsc{HittingSet} and 2D-RRMS
prefer smaller results to larger ones.
2D-RRMS can provide valid results on \textsc{Normal} because
the data points are almost evenly distributed among four orthants.

In general, \exact always returns the optimal results of \grmr
within reasonable time while \heuristic provides near-optimal results
of \grmr with superior efficiency on all 2D datasets
for different values of $\varepsilon$ and $n$.

\subsection{Results on High-Dimensional Datasets}\label{subsec:result:hd}

In this subsection, we compare the performance of \heuristic with
\textsc{$\varepsilon$-Kernel}, \textsc{HittingSet},
and typical RMS algorithms (\textsc{Greedy}, \textsc{HD-RRMS}, and \textsc{Sphere})
in high dimensions.
Note that \exact and \textsc{2D-RRMS} can only work in 2D
and thus are not evaluated in this subsection.
We test these algorithms on four real datasets of dimensionality $d>2$
(\textsc{Climate}, \textsc{El Nino}, \textsc{Household}, and \textsc{SUSY}),
as well as two synthetic datasets (\textsc{Normal} and \textsc{Uniform}).

\begin{figure}[t]
  \centering
  \subfigure[Climate]{
    \label{subfig:climate:k}
    \includegraphics[width=0.235\textwidth]{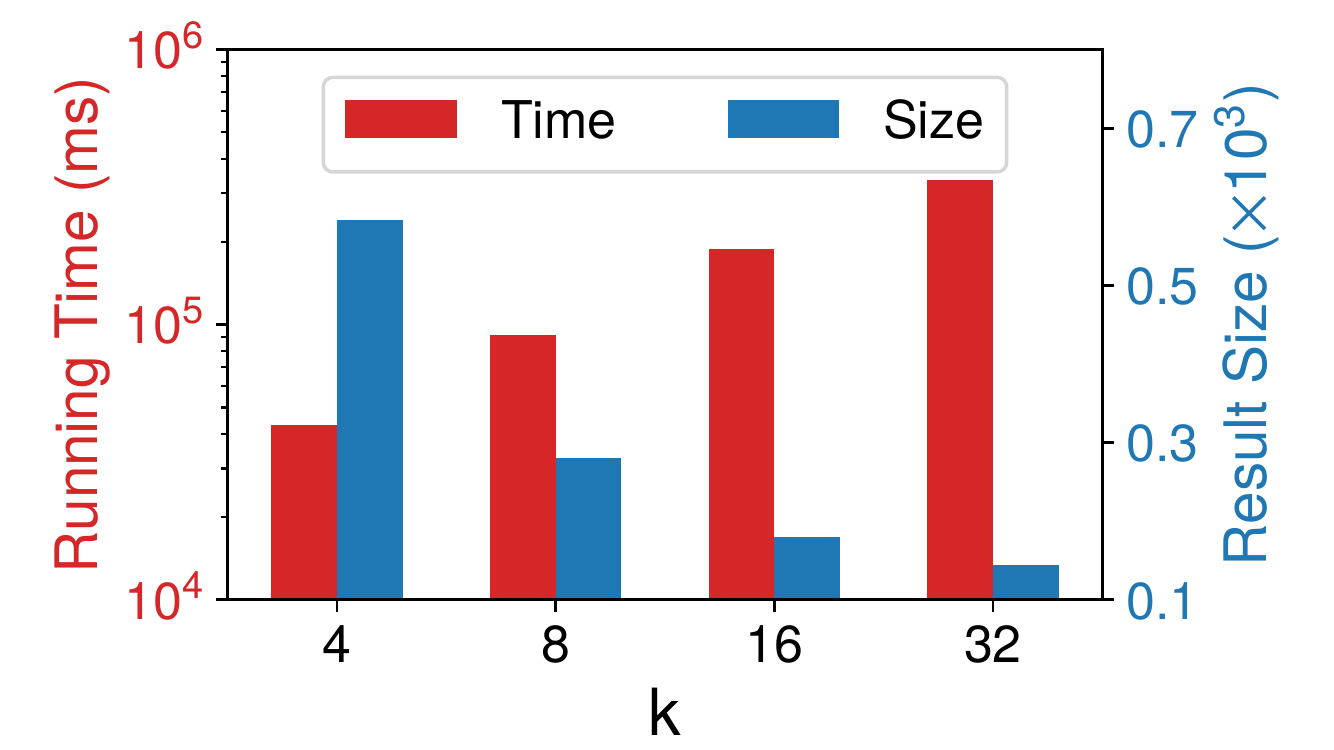}
  }
  \subfigure[El Nino]{
    \label{subfig:elnino:k}
    \includegraphics[width=0.235\textwidth]{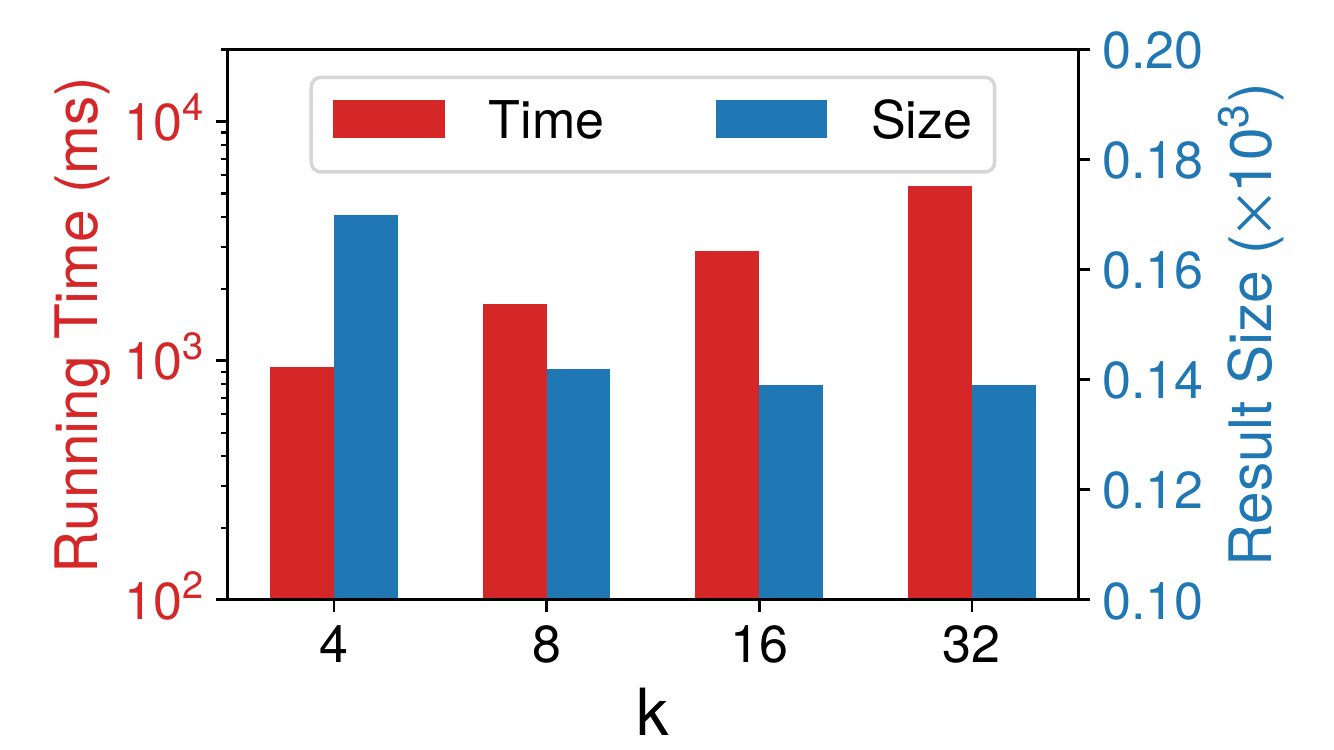}
  }
  \subfigure[Household]{
    \label{subfig:household:k}
    \includegraphics[width=0.235\textwidth]{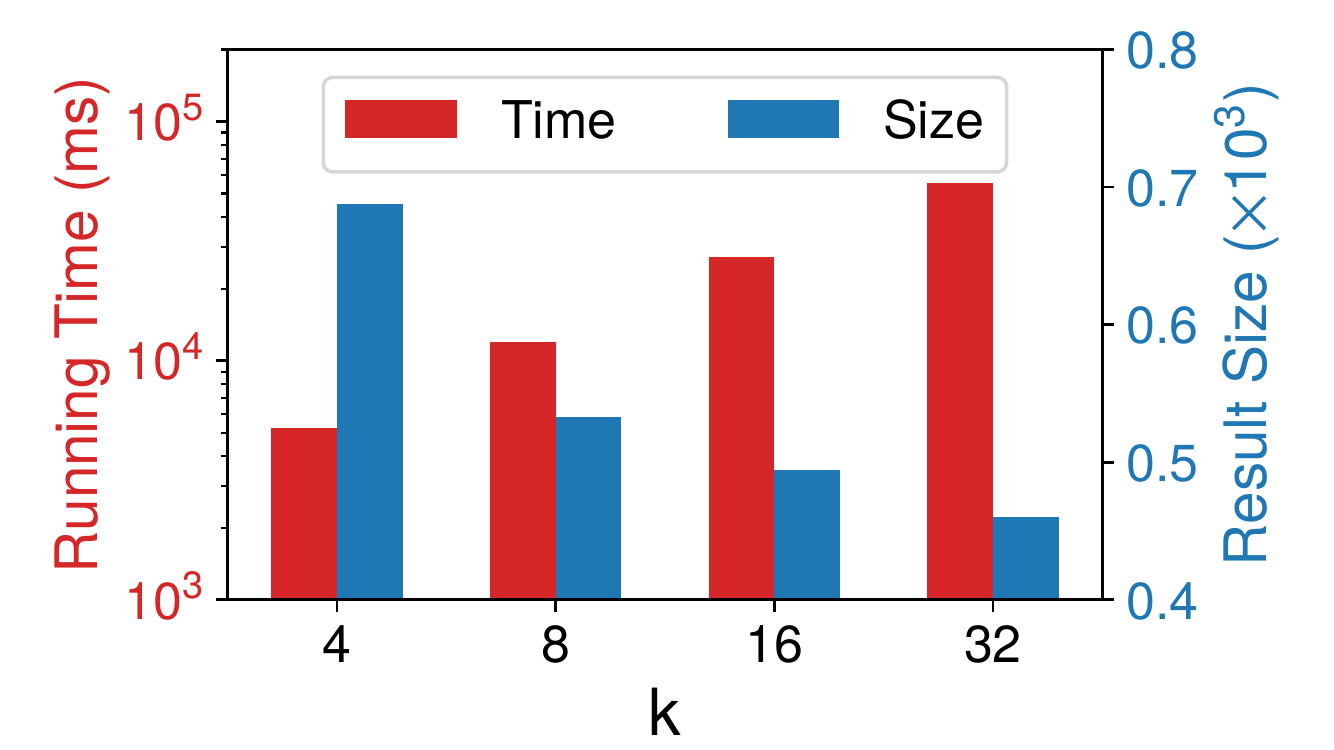}
  }
  \subfigure[SUSY]{
    \label{subfig:susy:k}
    \includegraphics[width=0.235\textwidth]{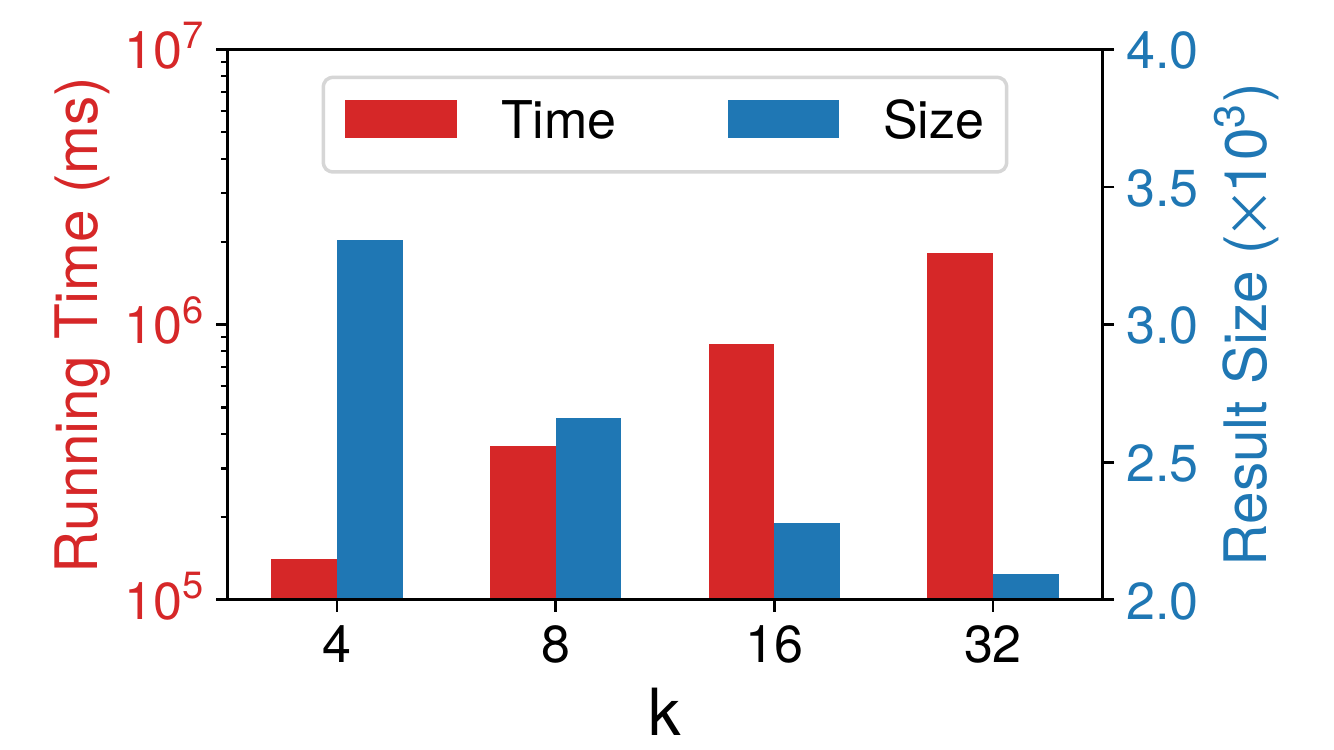}
  }
  \caption{Performance of \heuristic with varying $k$ in the IPDG construction ($\varepsilon=0.1$)}
  \label{fig:k:hd}
\end{figure}

\textbf{Impact of IPDG Construction:}
First of all, we test the effect of $k$ in the IPDG construction
on the performance of \heuristic. Figure~\ref{fig:k:hd} shows
the running time and result sizes of \heuristic for $k=4,8,16,32$
on four real-world datasets when $\varepsilon$ and $m$ are fixed to $0.1$ and $10^6$,
respectively. As discussed in Section~\ref{sec:alg:hd}, when $k$ is larger,
the running time of \heuristic increases significantly because
the approximate IPDG has more edges, which leads to
the increases in both the number of LPs and the number of constraints in each LP for \emph{dominance graph construction}.
Meanwhile, the result sizes of \heuristic decrease with increasing $k$
because more edges in the exact IPDG are contained
in the approximate one and thus the edge weights computed from LPs
are tighter and closer to the optimal ones. Nevertheless,
the solution quality on \textsc{El Nino}
does not improve anymore when $k=16,32$. This is because
the approximate IPDG built for $k=8$
has covered almost all edges of the exact IPDG.
Using a larger $k$ only leads to more redundant edges in this case.
In the remaining experiments, we will use the values of $k$
selected from $[4,8,16,32]$ that can strike the best balance between
efficiency and quality of results for \heuristic.
Note that if we fix $k$ and vary $m$
in the IPDG construction, we can observe the same trend as varying $k$:
The running time increases while the result sizes decrease for a larger $m$.
The results for varying $m$ are omitted due to space limitations.

\begin{figure}[t]
  \centering
  \includegraphics[width=0.96\textwidth]{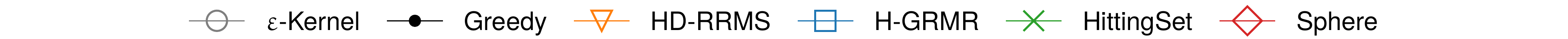}
  \subfigure[Climate]{
    \label{subfig:climate:loss}
    \includegraphics[width=0.235\textwidth]{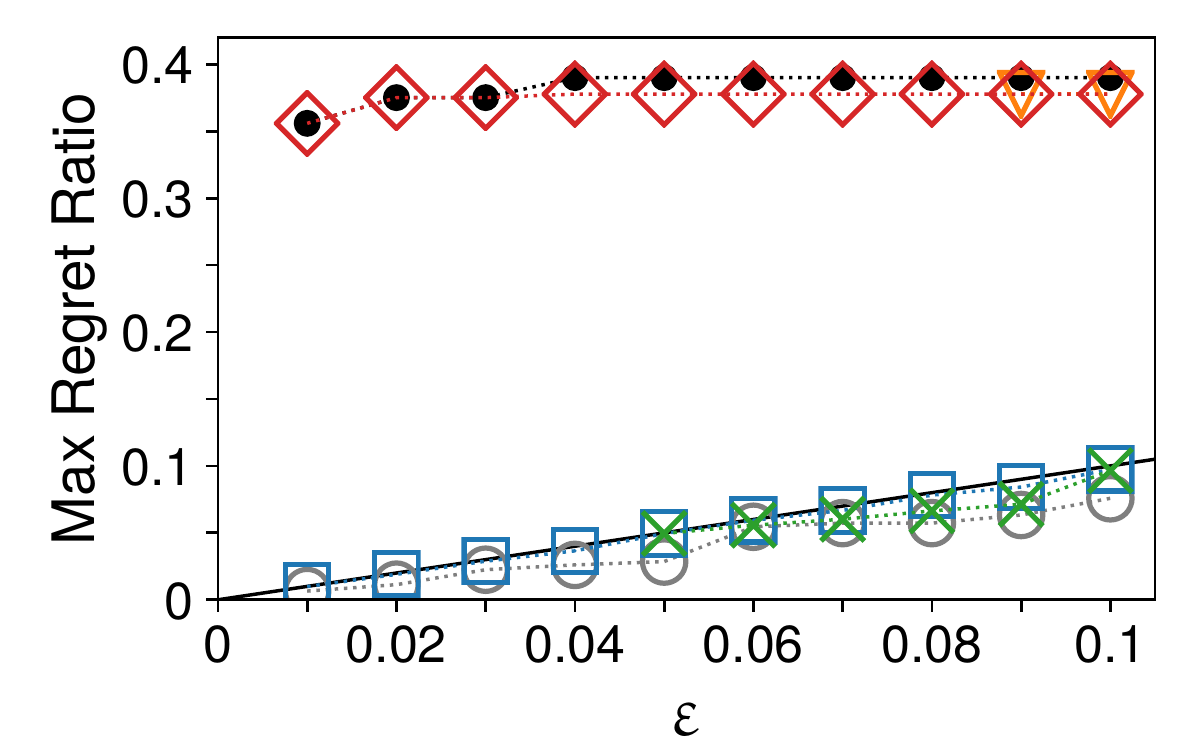}
  }
  \subfigure[El Nino]{
    \label{subfig:elnino:loss}
    \includegraphics[width=0.235\textwidth]{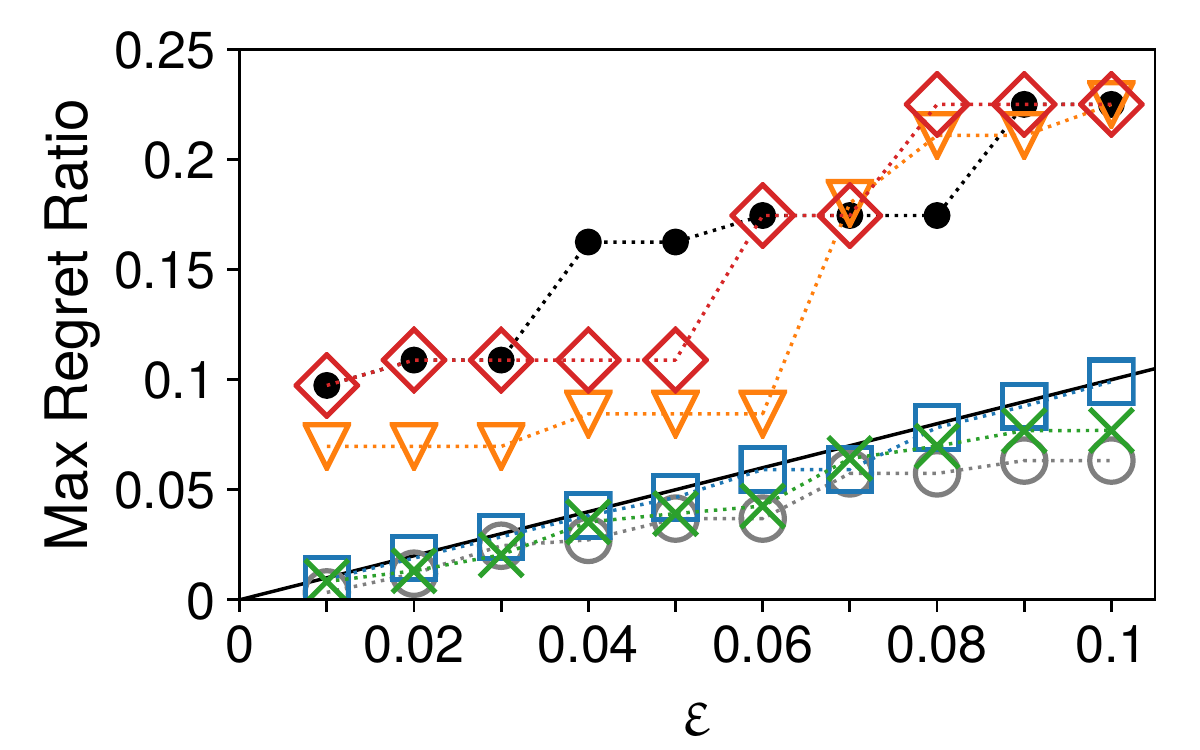}
  }
  \subfigure[Household]{
    \label{subfig:household:loss}
    \includegraphics[width=0.235\textwidth]{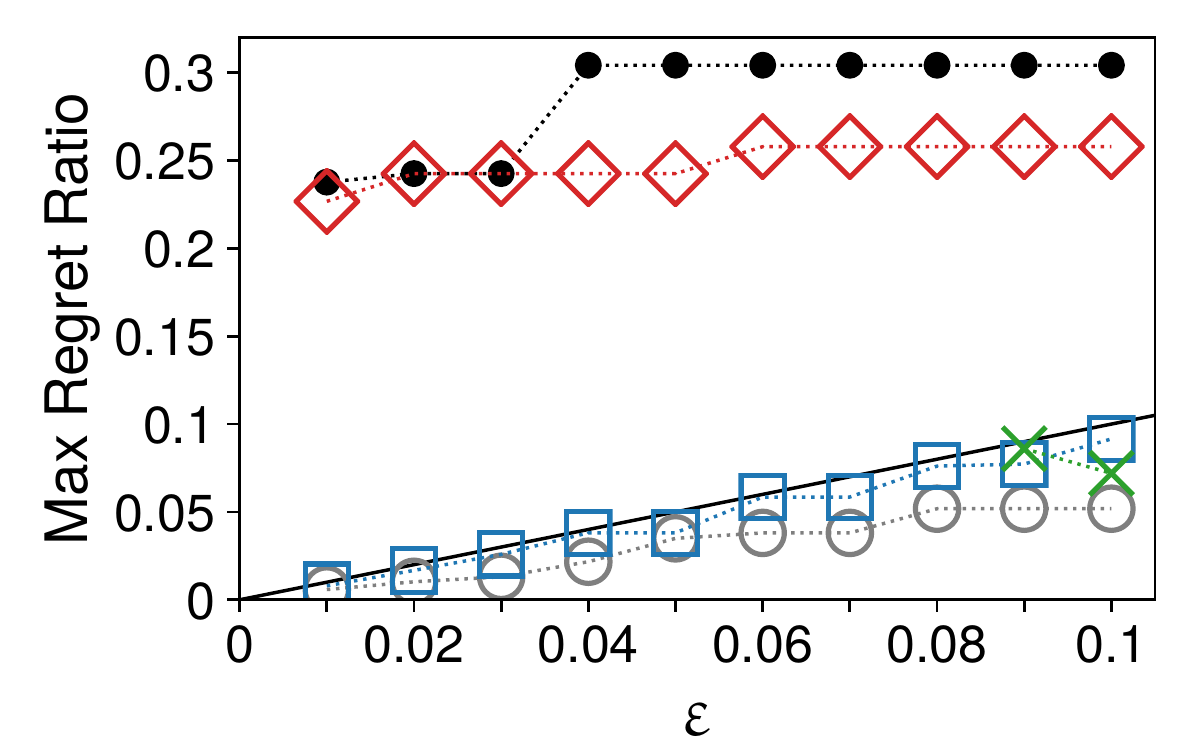}
  }
  \subfigure[SUSY]{
    \label{subfig:susy:loss}
    \includegraphics[width=0.235\textwidth]{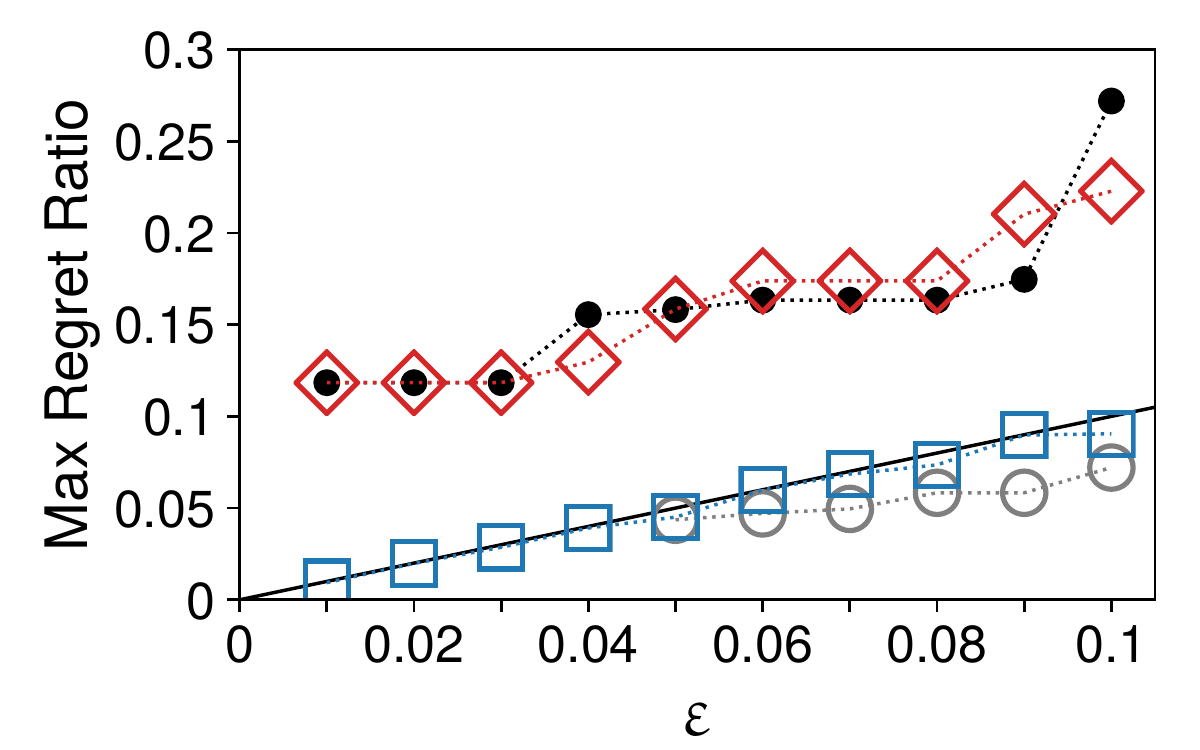}
  }
  \caption{Parameter $\varepsilon$ vs.~max regret ratio}\label{fig:loss:hd}
  \vspace{1em}
  \subfigure[Climate]{
    \label{subfig:climate:time}
    \includegraphics[width=0.235\textwidth]{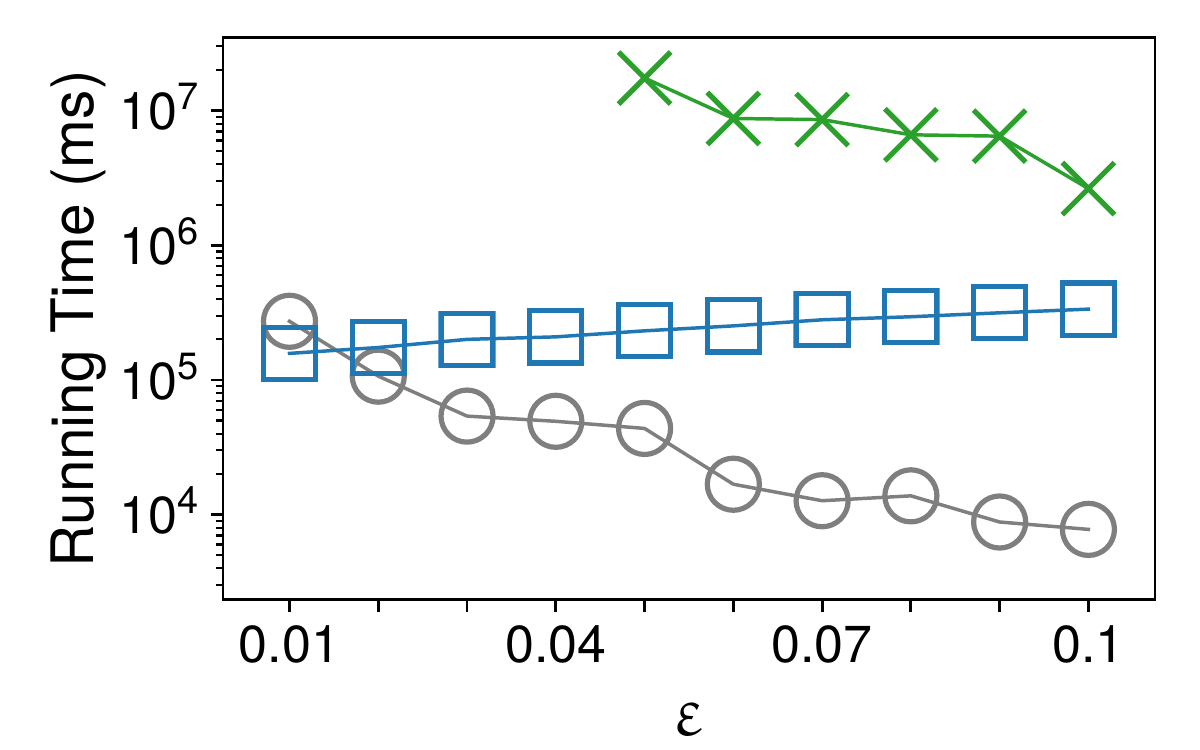}
  }
  \subfigure[El Nino]{
    \label{subfig:elnino:time}
    \includegraphics[width=0.235\textwidth]{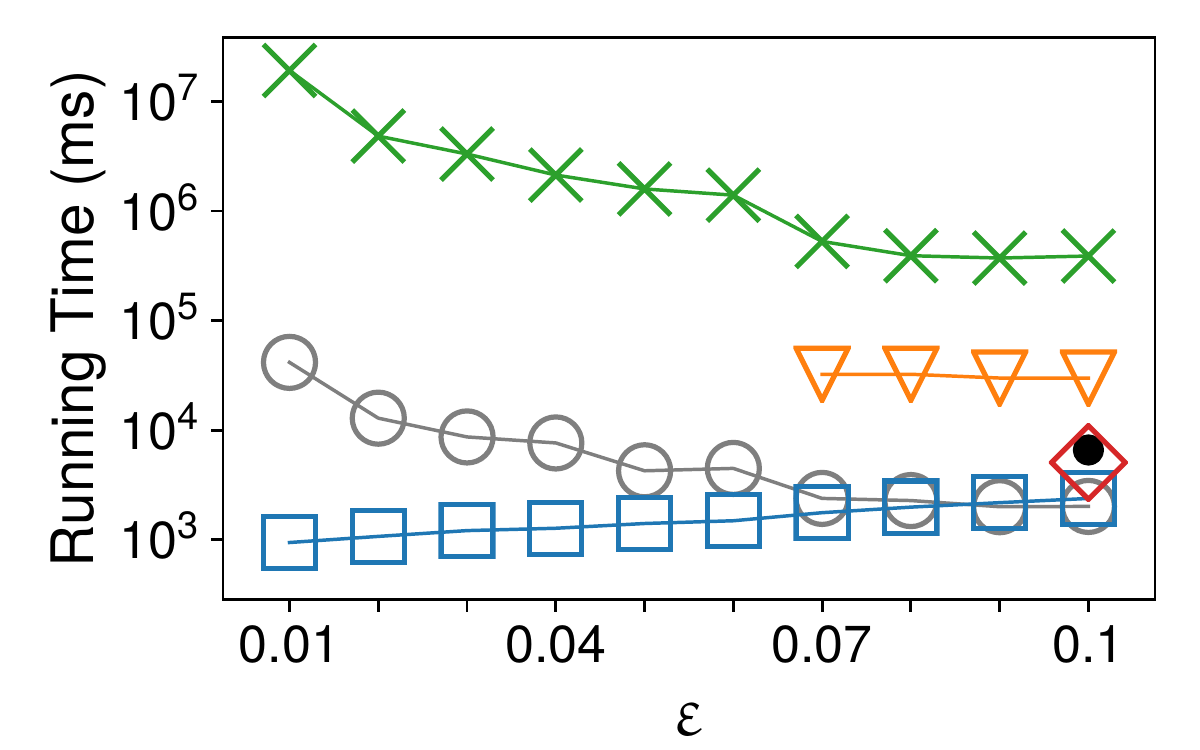}
  }
  \subfigure[Household]{
    \label{subfig:household:time}
    \includegraphics[width=0.235\textwidth]{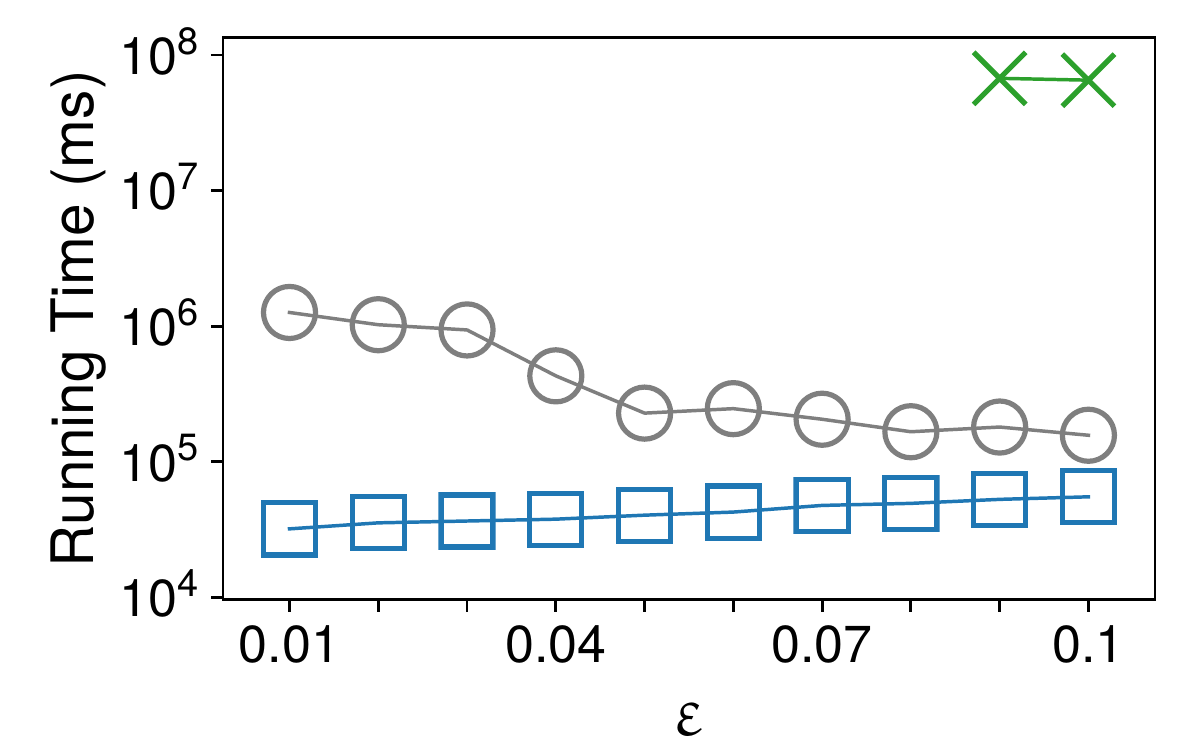}
  }
  \subfigure[SUSY]{
    \label{subfig:susy:time}
    \includegraphics[width=0.235\textwidth]{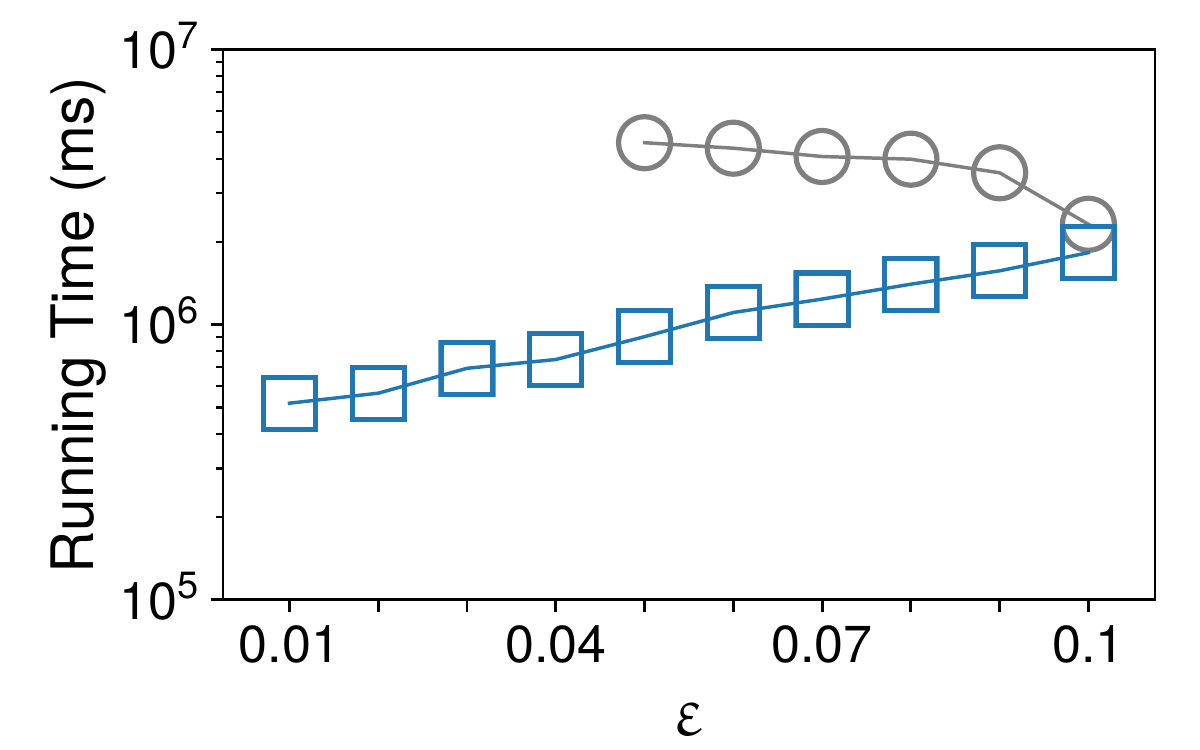}
  }
  \caption{Running time with varying $\varepsilon$}\label{fig:time:hd}
  \vspace{1em}
  \subfigure[Climate]{
    \label{subfig:climate:size}
    \includegraphics[width=0.235\textwidth]{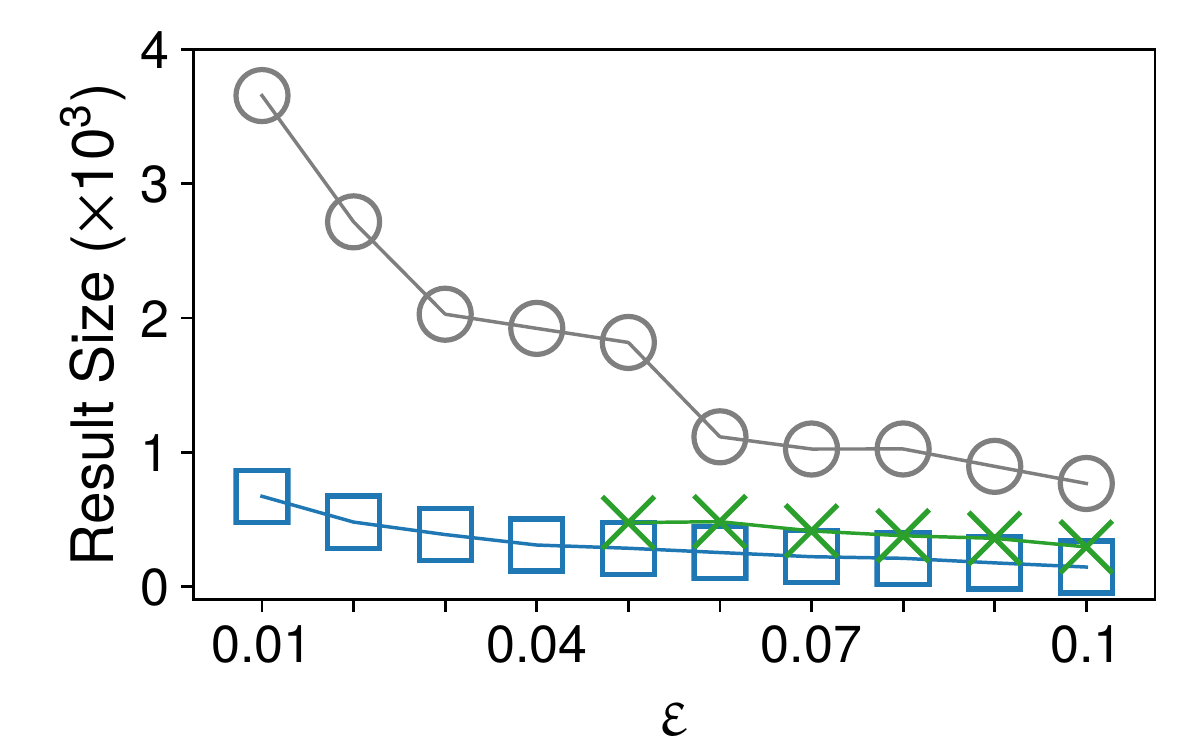}
  }
  \subfigure[El Nino]{
    \label{subfig:elnino:size}
    \includegraphics[width=0.235\textwidth]{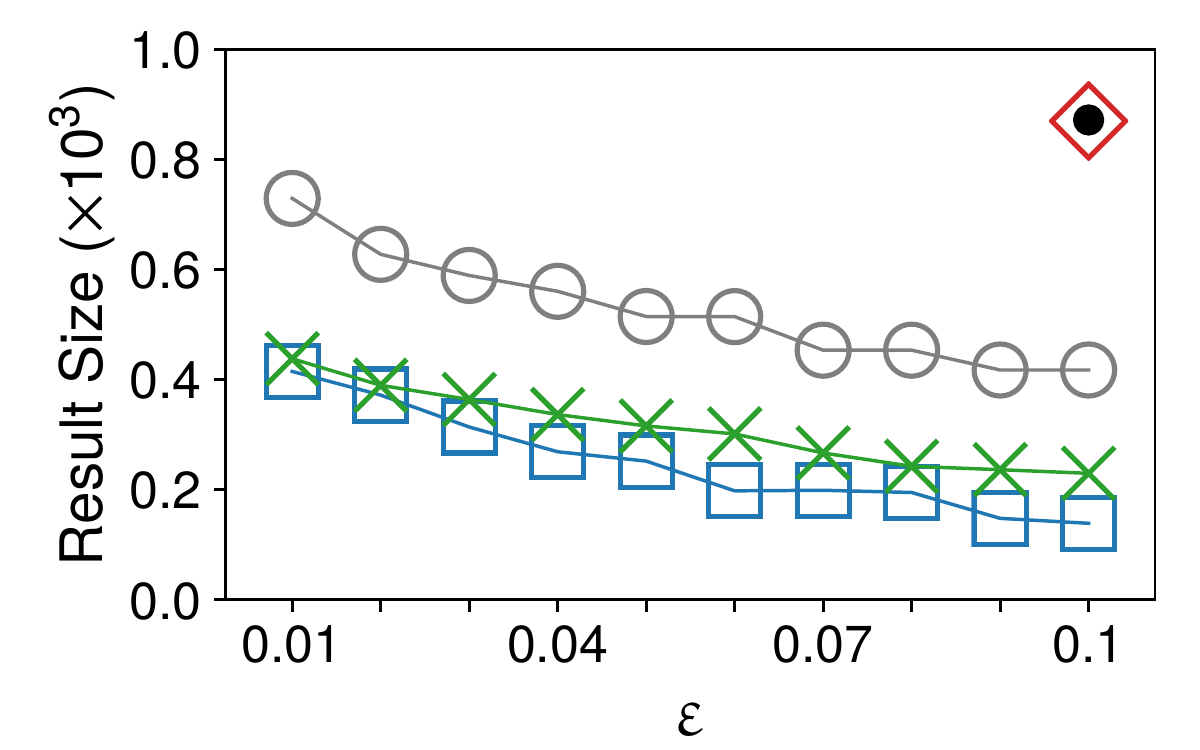}
  }
  \subfigure[Household]{
    \label{subfig:household:size}
    \includegraphics[width=0.235\textwidth]{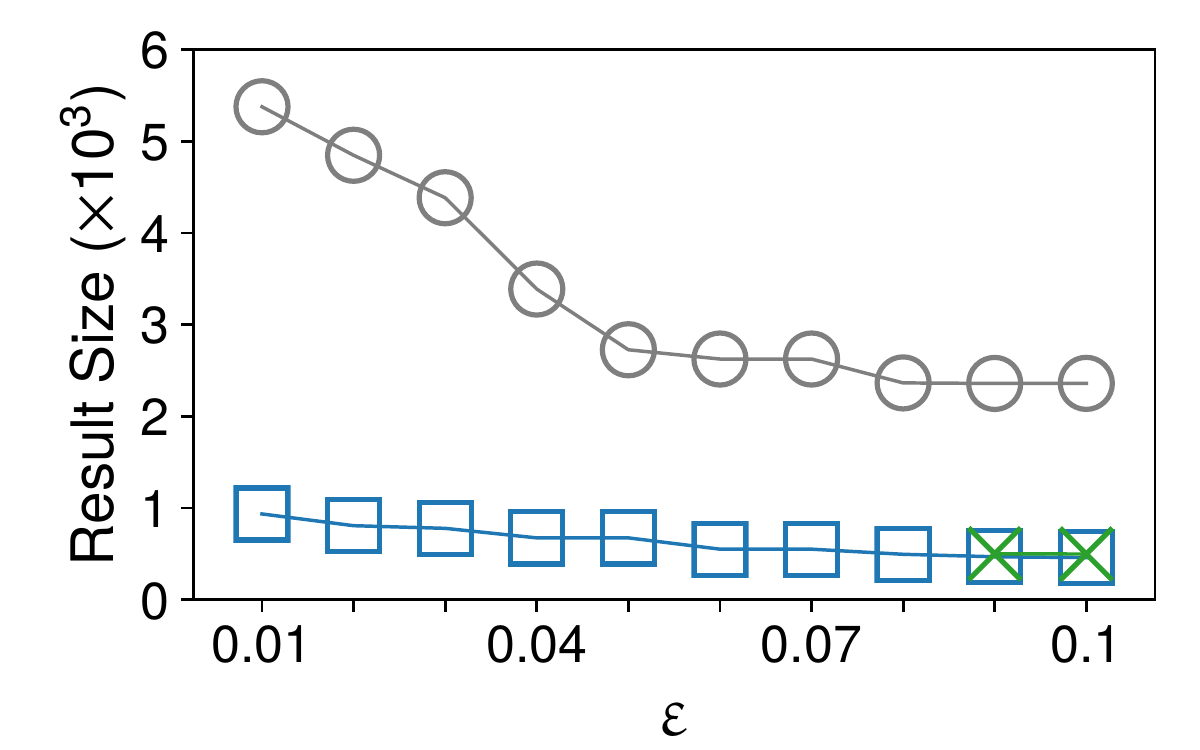}
  }
  \subfigure[SUSY]{
    \label{subfig:susy:size}
    \includegraphics[width=0.235\textwidth]{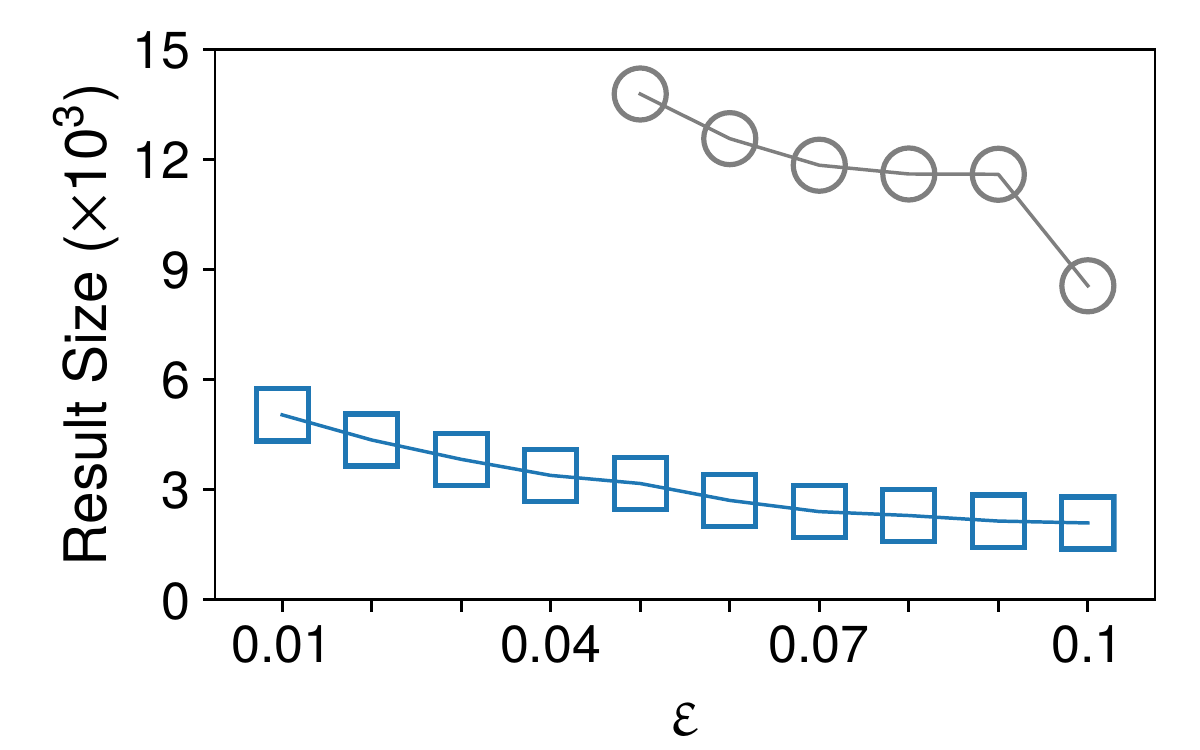}
  }
  \caption{Result size with varying $\varepsilon$}\label{fig:size:hd}
\end{figure}

\textbf{Impact of Parameter $\varepsilon$:}
Next, we vary $\varepsilon$ from $0.01$ to $0.1$ to evaluate the performance of each algorithm.
Note that we terminate the execution of an algorithm after running it
on a dataset for one day. \textsc{HD-RRMS} and \textsc{HittingSet} suffer from a low efficiency
and cannot return any result on a large dataset when $\varepsilon$ is small. 
In Figure~\ref{fig:loss:hd}, we show the maximum regret ratios of the results
of each algorithm for different $\varepsilon$. 
Similar to the 2D setting, we find empirically that RMS algorithms (\textsc{Greedy}, \textsc{HD-RRMS}, and \textsc{Sphere}) do not provide any valid result for GRMR on real datasets due to the skewness of data distributions.
The running time and result sizes of each algorithm with varying $\varepsilon$ are presented in Figures~\ref{fig:time:hd} and~\ref{fig:size:hd}, respectively.
First of all, we notice that \heuristic runs slower when $\varepsilon$ is larger. 
This is expected, as the dominance graph contains all edges with weights at most $\varepsilon$
and therefore has more edges for a larger $\varepsilon$.
It thus takes more time for both \emph{graph construction} and \emph{result computation}
when $\varepsilon$ is larger.
On the other hand, both \textsc{$\varepsilon$-Kernel} and \textsc{HittingSet} run faster
when $\varepsilon$ is larger because of smaller sample sizes for computation.
In terms of result sizes, all three algorithms identify smaller results
with increasing $\varepsilon$, as expected.
Finally, compared with \textsc{$\varepsilon$-Kernel}, \heuristic achieves higher efficiencies
in all datasets except \textsc{Climate}. 
At the same time, \heuristic produces results of significantly better quality than
\textsc{$\varepsilon$-Kernel}: the result size of \heuristic is up to $5.7$ times smaller than
that of \textsc{$\varepsilon$-Kernel}.
Compared with \textsc{HittingSet}, \heuristic runs up to four orders of magnitude faster while providing results with $1.07$--$2.1$ times smaller sizes.
In particular, \heuristic is the only algorithm that
returns a valid result for GRMR on \textsc{SUSY} in reasonable time when $\varepsilon < 0.05$.

\begin{figure}[t]
  \centering
  \includegraphics[width=0.96\textwidth]{legend-hd.pdf}
  \subfigure[Normal, time]{
    \label{subfig:normal:d:time}
    \includegraphics[width=0.235\textwidth]{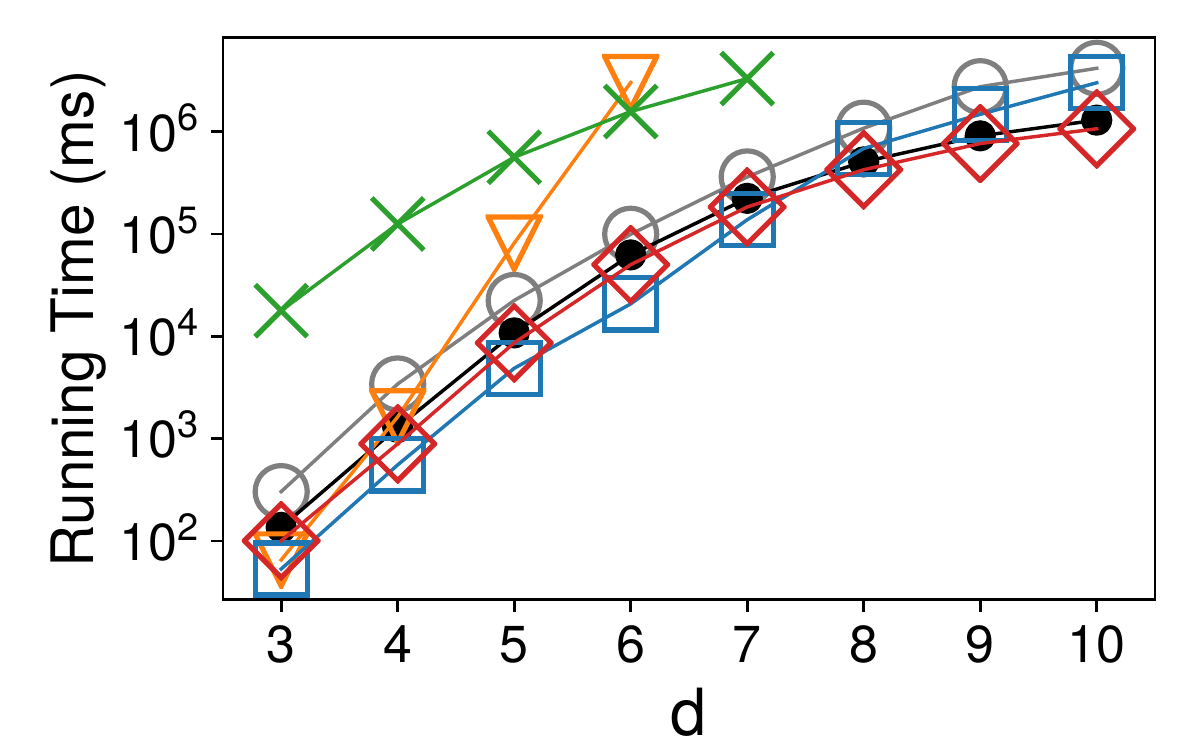}
  }
  \subfigure[Normal, size]{
    \label{subfig:normal:d:size}
    \includegraphics[width=0.235\textwidth]{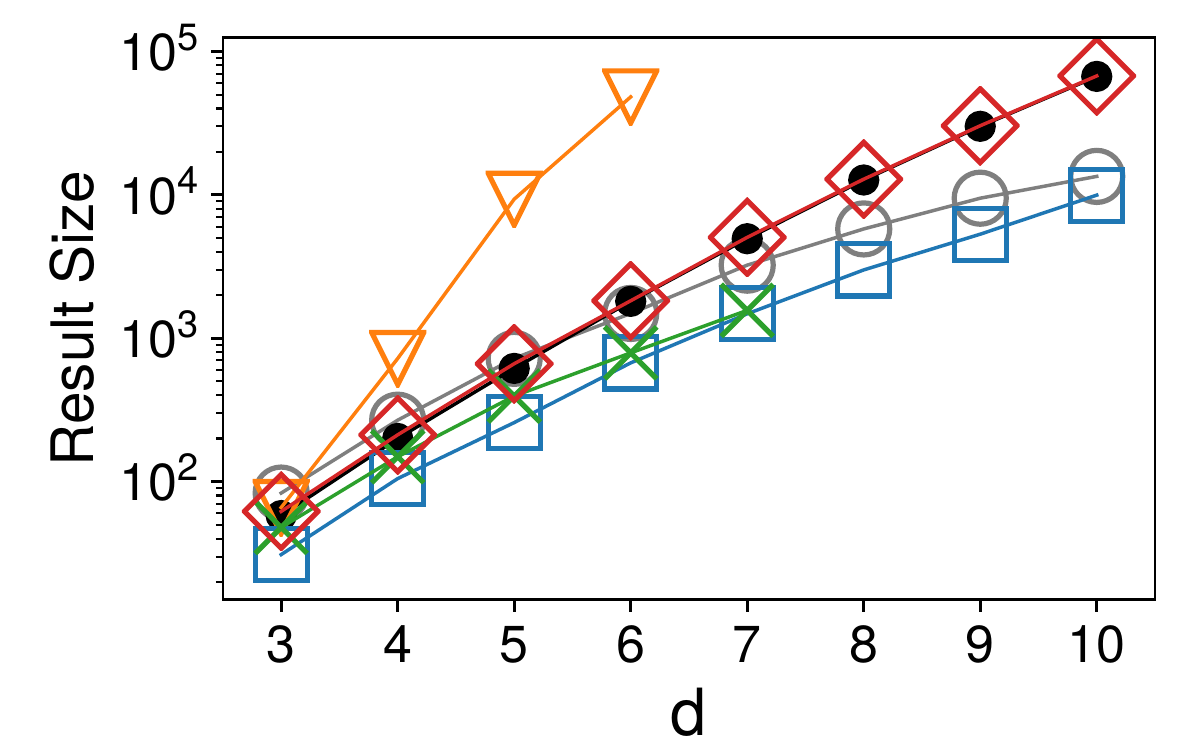}
  }
  \subfigure[Uniform, time]{
    \label{subfig:uniform:d:time}
    \includegraphics[width=0.235\textwidth]{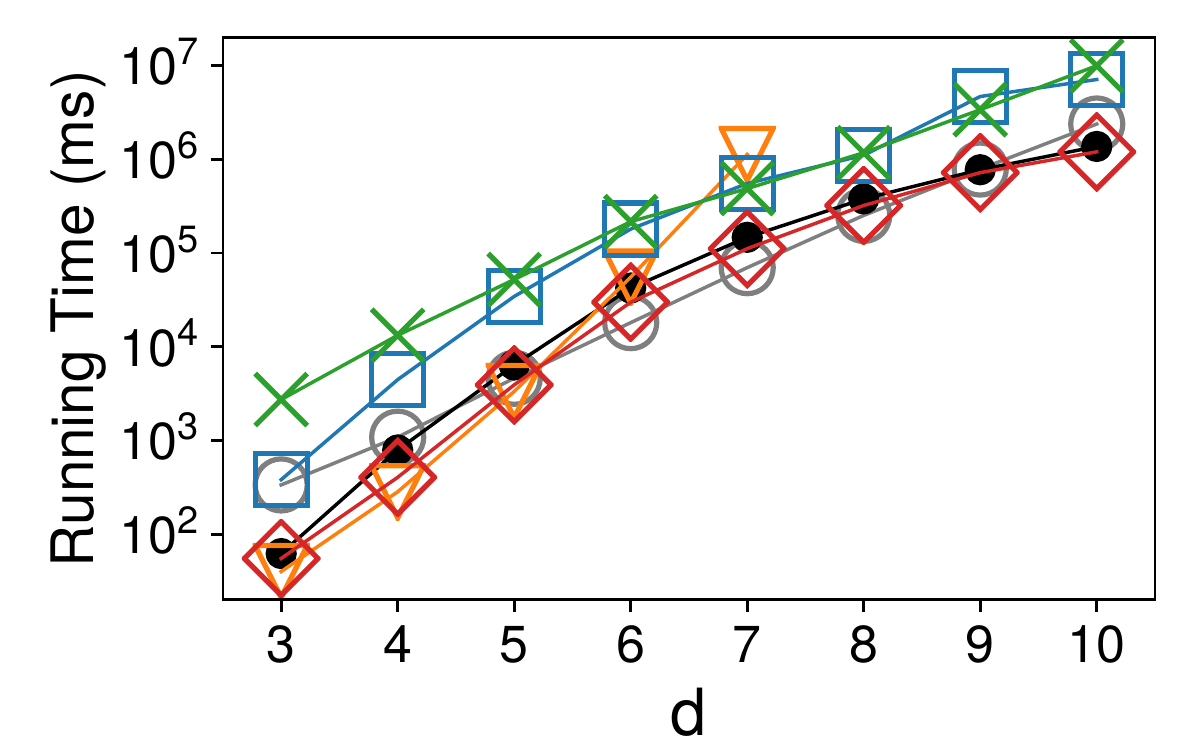}
  }
  \subfigure[Uniform, size]{
    \label{subfig:uniform:d:size}
    \includegraphics[width=0.235\textwidth]{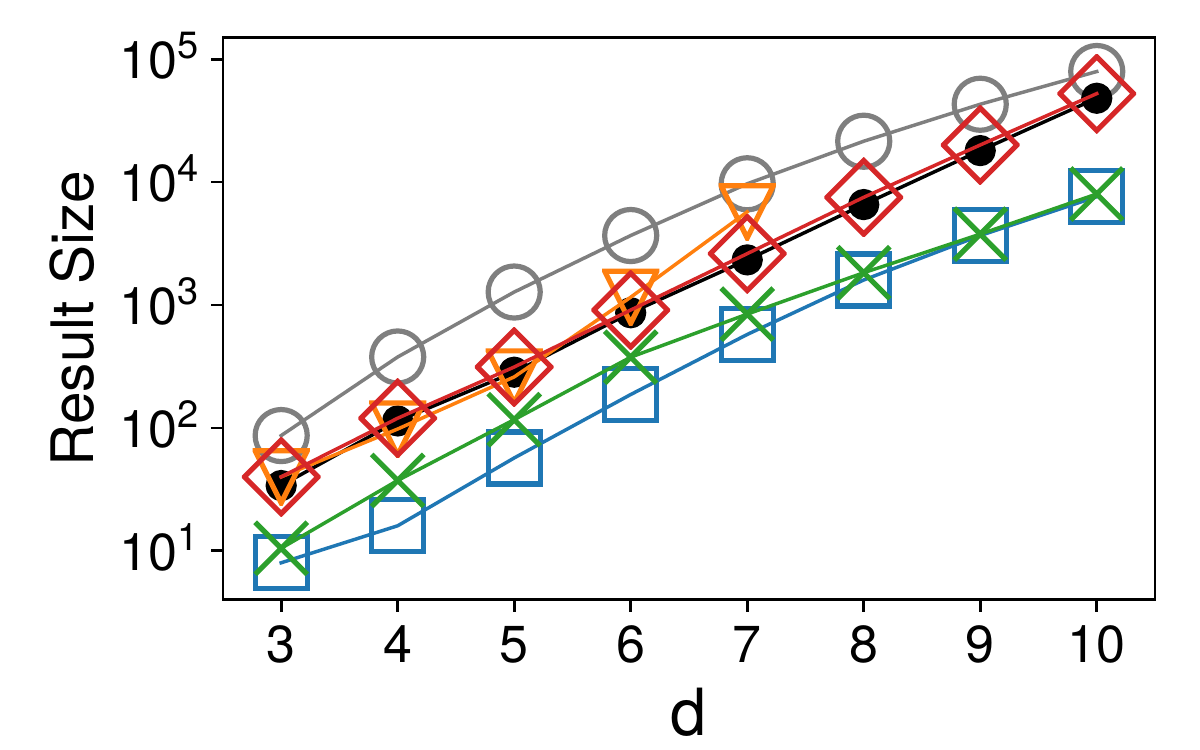}
  }
  \caption{Performance with varying the dimensionality $d$ ($\varepsilon=0.1$)}
  \label{fig:results:d}
  \vspace{1em}
  \subfigure[Normal (6D), time]{
    \label{subfig:normal:N:time}
    \includegraphics[width=0.235\textwidth]{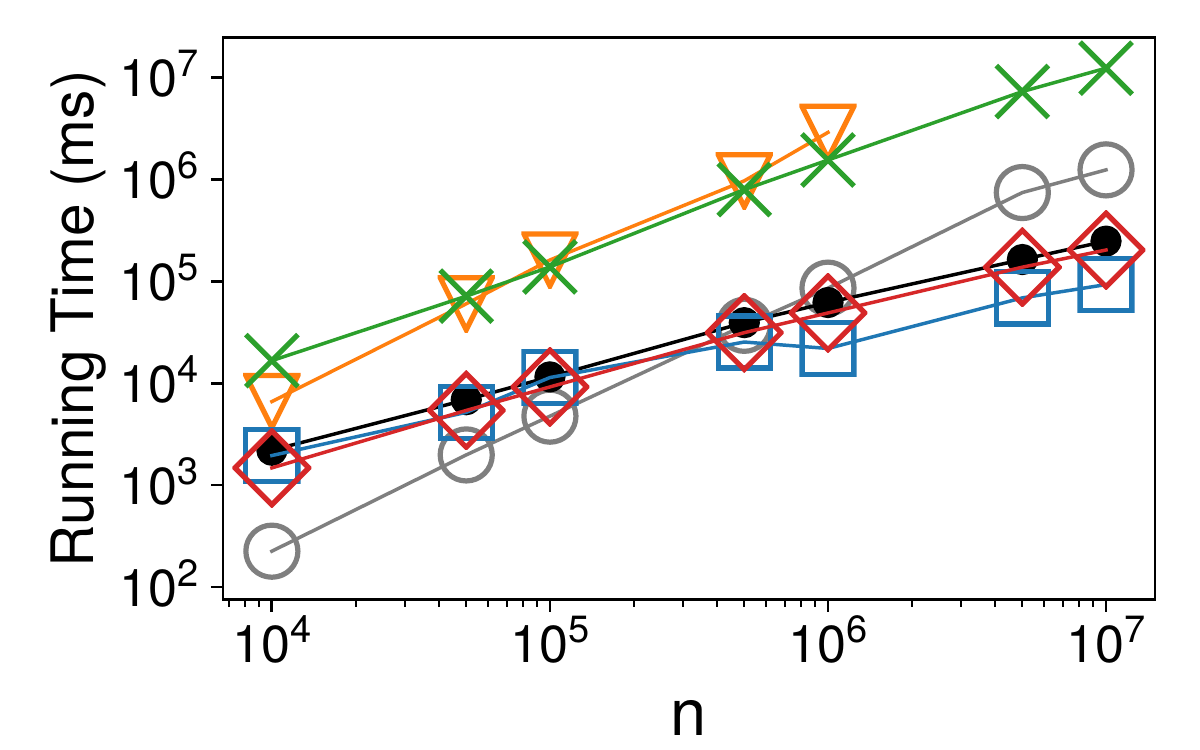}
  }
  \subfigure[Normal (6D), size]{
    \label{subfig:normal:N:size}
    \includegraphics[width=0.235\textwidth]{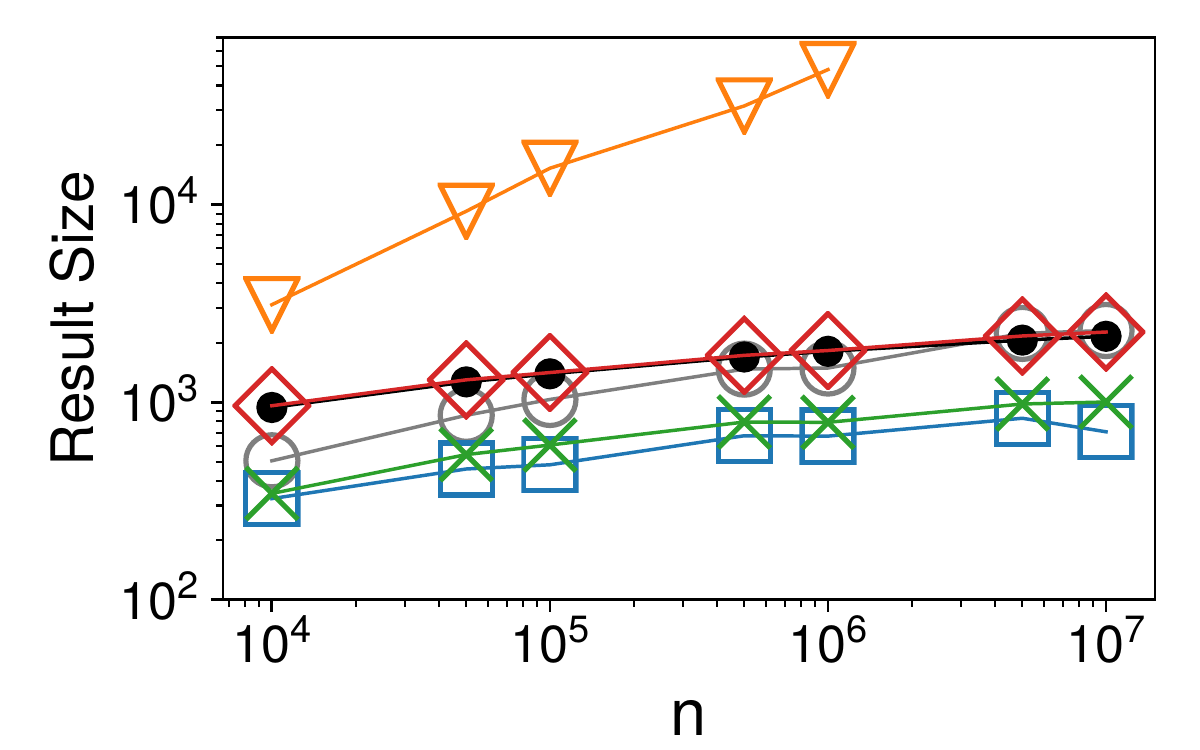}
  }
  \subfigure[Uniform (6D), time]{
    \label{subfig:uniform:N:time}
    \includegraphics[width=0.235\textwidth]{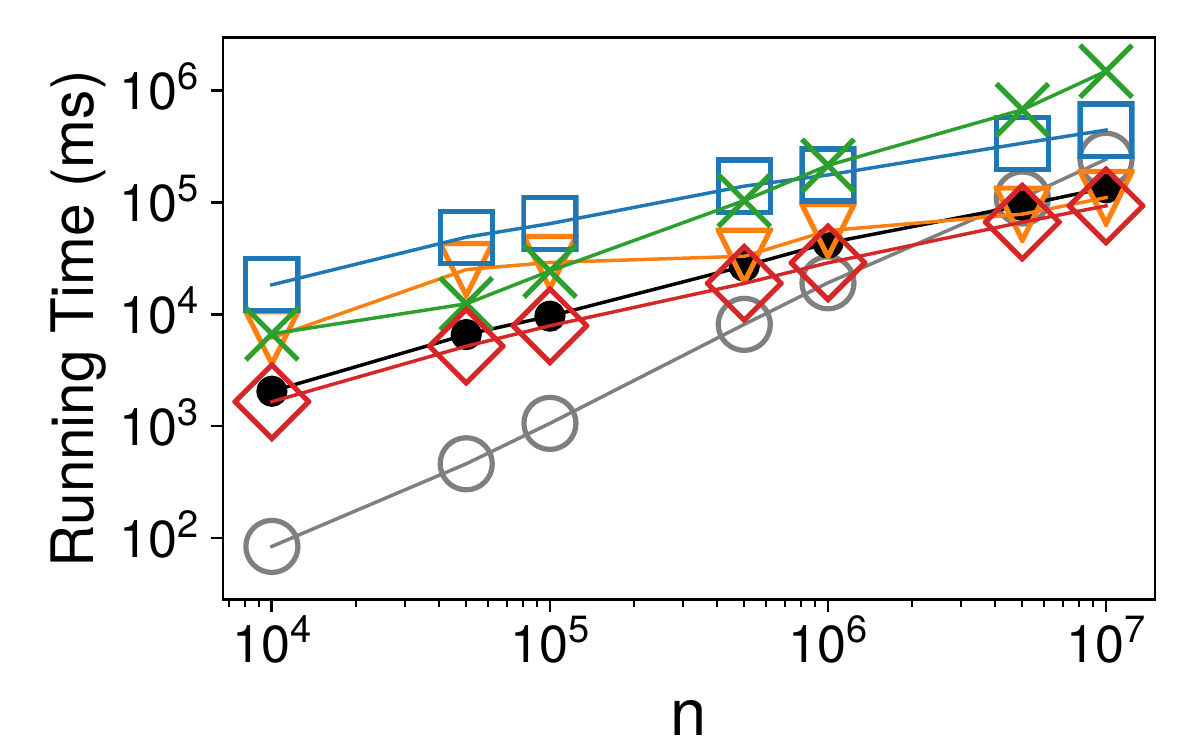}
  }
  \subfigure[Uniform (6D), size]{
    \label{subfig:uniform:N:size}
    \includegraphics[width=0.235\textwidth]{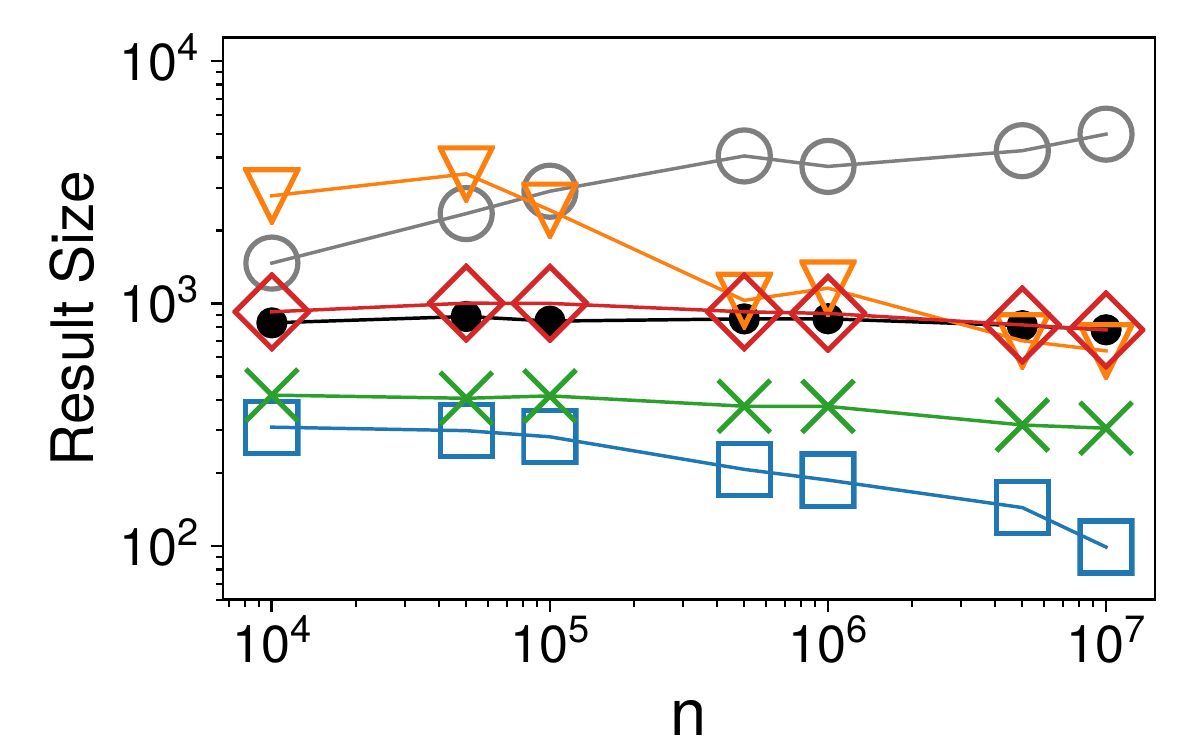}
  }
  \caption{Performance with varying the dataset size $n$ ($\varepsilon=0.1$)}
  \label{fig:results:n}
\end{figure}

\textbf{Impact of Dimensionality $d$ and Dataset Size $n$:}
Finally, we evaluate the impact of the dimensionality $d$ and the dataset size $n$ for different algorithms on the two synthetic datasets.
In these experiments, we fix $\varepsilon$ to $0.1$.
The results for varying dimensionality $d$ ($d = 3$ to $10$) are shown in
Figure~\ref{fig:results:d}. 
Both the running time and result sizes
of all algorithms grow rapidly with $d$. 
Such results are not surprising, because the number of extreme points in a dataset increases super-linearly with $d$. 
\textsc{HD-RRMS} and \textsc{HittingSet} are inefficient in high dimensions and
cannot return any result on \textsc{Normal} when $d>7$.
\textsc{Greedy} and \textsc{Sphere} return valid results for \grmr
by solving $2^d$ RMS problems. However, their solution quality is clearly inferior to \heuristic,
especially in higher dimensions.
Generally, \heuristic finds results of better quality than any other algorithm in reasonable time
when the dimensionality $d$ ranges from $3$ to $10$.
The results with varying the dataset size $n$ from $10^4$ to $10^7$ are illustrated in Figure~\ref{fig:results:n}.
The trends are generally similar to those for varying $d$.
Both the running time and result sizes grow with $n$ because
of the increasing number of extreme points.
The only exception is that the result size of most algorithms
(except \textsc{$\varepsilon$-Kernel}) decreases with $n$ on \textsc{Uniform}.
\heuristic outperforms all other algorithms in terms of quality of results.
In particular, the result of \heuristic is nearly $50$ times smaller than
that of \textsc{$\varepsilon$-Kernel} on \textsc{Uniform} when $n=10^7$.

\section{Related Work}\label{sec:literature}

A great variety of approaches to representing a large dataset by a small subset of data points
have been proposed recently~\cite{DBLP:journals/jacm/AgarwalHV04,DBLP:journals/pvldb/NanongkaiSLLX10,DBLP:conf/kdd/BadanidiyuruMKK14,DBLP:conf/cikm/ZhuangRHGHA16,DBLP:conf/kdd/BachemL018,DBLP:conf/icml/CelisKS0KV18,DBLP:journals/tkde/WangLT19,DBLP:conf/sigmod/AsudehN00J19,DBLP:conf/kdd/WangLT19}.
One important method that has been extensively investigated is \emph{maxima representation}~\cite{DBLP:conf/icde/BorzsonyiKS01,DBLP:conf/sigmod/ChangBCLLS00}, which finds a compact subset that contains the maxima (i.e., points with the highest scores) of a dataset for any possible ranking function. According to the class of ranking functions considered, there are many different definitions of maxima representations. 
In particular, the \emph{convex hull}~\cite{DBLP:books/sp/PreparataS85,DBLP:conf/sigmod/ChangBCLLS00} and \emph{skyline}~\cite{DBLP:conf/icde/BorzsonyiKS01} are two examples of maxima representations when the classes of all \emph{linear functions} and \emph{nonnegative monotonic functions} are considered, respectively.

In practice, since maxima representations can still be overwhelmingly large~\cite{DBLP:conf/sigmod/AsudehN00J19,DBLP:conf/sigmod/AsudehN0D17}, recent efforts have been directed toward reducing their sizes using approximation techniques. 
Nanongkai et al.~\cite{DBLP:journals/pvldb/NanongkaiSLLX10} were the first to propose the \emph{regret-minimizing set} (RMS) problem for approximate maxima representations.
They introduced the \emph{regret ratio}, which is the relative difference in scores between the top-ranked point of the dataset and the top-ranked point of a subset,
as the measure of \emph{regret} for a ranking function.
They used the \emph{maximum regret ratio}, i.e., the maximum of the regret ratios over all nonnegative (monotonic) linear functions, to measure how well a subset approximates the maxima representation of a dataset.
An RMS is defined as the smallest subset whose maximum regret ratio was at most $\varepsilon$.
The RMS problem was proven to be NP-hard~\cite{DBLP:conf/icdt/CaoLWWWWZ17,DBLP:conf/wea/AgarwalKSS17} for any dataset in three or higher dimensions.
Since the seminal work of Nanongkai et al., different approximation and heuristic algorithms~\cite{DBLP:journals/pvldb/NanongkaiSLLX10,DBLP:conf/icde/PengW14,DBLP:conf/sigmod/AsudehN0D17,DBLP:conf/wea/AgarwalKSS17,DBLP:conf/alenex/KumarS18,DBLP:conf/sigmod/XieW0LL18,DBLP:journals/pvldb/ShetiyaAAD19,DBLP:conf/icdt/CaoLWWWWZ17} were proposed for RMS. 
Please refer to~\cite{DBLP:journals/vldb/XieWL20} for a survey of algorithmic techniques for RMS.

Furthermore, several works~\cite{DBLP:journals/pvldb/ChesterTVW14,DBLP:conf/icdt/CaoLWWWWZ17,DBLP:journals/pvldb/FaulknerBL15,DBLP:journals/tods/QiZSY18,DBLP:conf/aaai/SomaY17a,DBLP:conf/icde/ZeighamiW19,DBLP:conf/aaai/StorandtF19,DBLP:conf/sigmod/AsudehN00J19,DBLP:conf/sigmod/NanongkaiLSM12,DBLP:conf/sigmod/XieWL19,DBLP:conf/icde/XieW0T20} studied different generalizations and variations of RMS.
Chester et al.~\cite{DBLP:journals/pvldb/ChesterTVW14} generalized the regret ratio to $k$-regret ratio that expressed the score difference between the top-ranked point in the subset and the $k$\textsuperscript{th}-ranked point in the dataset.
Accordingly, they extended RMS to $k$-RMS for approximating maxima representations of top-$k$ results (instead of top-$1$ results) w.r.t.~all ranking functions.
Different RMS problems with nonlinear utility functions were studied in~\cite{DBLP:journals/pvldb/FaulknerBL15,DBLP:journals/tods/QiZSY18,DBLP:conf/aaai/SomaY17a}.
Specifically, they considered \emph{convex/concave functions}~\cite{DBLP:journals/pvldb/FaulknerBL15}, \emph{multiplicative functions}~\cite{DBLP:journals/tods/QiZSY18}, and \emph{submodular functions}~\cite{DBLP:conf/aaai/SomaY17a}, respectively, but all of which were monotonic.
The \emph{average regret minimization}~\cite{DBLP:conf/icde/ZeighamiW19,DBLP:conf/aaai/StorandtF19,DBLP:journals/pvldb/ShetiyaAAD19} problem has also been investigated recently. 
Instead of targeting the maximum regret ratio, it uses the average of regret ratios over all ranking functions as the measure of representativeness.
In another variant, Nanongkai el al.~\cite{DBLP:conf/sigmod/NanongkaiLSM12} and Xie et al.~\cite{DBLP:conf/sigmod/XieWL19} proposed the \emph{interactive regret minimization} problem by introducing user interactions to enhance RMS. 
Moreover, Asudeh et al.~\cite{DBLP:conf/sigmod/AsudehN00J19} proposed the \emph{ranking-regret representatives} (RRR) in which the regret was defined by rankings instead of scores.
Specifically, an RRR is the smallest subset that contains at least one of
the top-$k$ points in the dataset for all ranking functions.
Xie et al.~\cite{DBLP:conf/icde/XieW0T20} studied a dual problem of RMS called \emph{happiness maximization}, where the goal is to maximize the happiness (i.e., one minus the regret) instead of minimizing the regret.

Note that all above approaches to approximate maxima representations
are designed for monotonic linear (or nonlinear, in some cases) ranking functions. 
Accordingly, most existing algorithms for RMS and related problems rely on
the monotonicity of ranking functions and naturally cannot be directly used for \grmr.
To the best of our knowledge, \grmr is the only method that considers the class of all linear
functions including non-monotonic ones with negative weights as the ranking functions.

\section{Conclusion}\label{sec:conclusion}

In this paper, we proposed the \emph{generalized regret-minimizing representative} (\grmr) problem
to identify the smallest subset that can approximate the maximum score of the dataset for any linear function
within a regret ratio of at most $\varepsilon$.
We proved the NP-hardness of \grmr in three or higher
dimensions. Following a geometric interpretation of \grmr, we designed
an exact algorithm for \grmr in two dimension and a heuristic algorithm for \grmr in arbitrary
dimensions. Finally, we conducted extensive experiments
on real and synthetic datasets to verify the performance of our proposed
algorithms. The experimental results confirmed the efficiency, effectiveness,
and scalability of our algorithms for \grmr.

\bibliographystyle{abbrvnat}
\bibliography{references}

\begin{appendix}
\section{Proof of Theorem~\ref{thm:np:hardness}}
\label{proof:np:hardness}

\begin{proof}
  For a set $P_0 $ of $n$ points in $\mathbb{R}^3_+$
  and a positive integer $ r \in \mathbb{Z}^+ $,
  a 3D RMS instance $\mathsf{RMS}(P_0,r)$ asks,
  given a real number $\varepsilon \in (0,1)$,
  whether there exists a subset $Q_0 \subseteq P_0$ of size $r$
  such that $\omega(x,Q_0) \geq (1-\varepsilon) \cdot \omega(x,P_0)$
  for any $x \in \mathbb{S}^2_{+}$. Intuitively, RMS is a restricted
  version of GRMR where both data points and utility vectors are in the nonnegative orthant.
  Here, we can restrict $P_0 \in [0,1]^3$
  because of the scale-invariance of RMS~\cite{DBLP:journals/pvldb/NanongkaiSLLX10}. We use
  $ l^{+}(Q_0)=\max_{x \in \mathbb{S}^2_+} 1-\frac{\omega(x,Q_0)}{\omega(x,P_0)}$
  to denote the \emph{maximum regret ratio} of $Q_0$ over $P_0$ for RMS.
  Given any $\mathsf{RMS}(P_0,r)$, we should construct an instance $\mathsf{GRMR}(P',r')$
  satisfying that there exists a subset $Q_0 \subseteq P_0$ of size $r$
  such that $l^{+}(Q_0) \leq \varepsilon$ if and only if there exists a subset $Q' \subseteq P'$
  of size $r'$ such that $l(Q') \leq \varepsilon$ for an arbitrary $\varepsilon \in (0,1)$.

  For $\mathsf{RMS}(P_0,r)$ and $\varepsilon \in (0,1)$, we add three new points
  $B=\{b_1,b_2,b_3\}$ to $P_0$.
  Let $b_1=(1-\eta,1,1)$, $b_2=(1,1-\eta,1)$, and $b_3=(1,1,1-\eta)$ where $\eta > 3$.
  The value of $\eta$ should be determined by $P_0$, $Q_0$, and $\varepsilon$ as discussed later.
  We will prove that $P_0$ has an $\varepsilon$-regret set of size $r$
  if and only if $ P' = P_0 \cup B $ has an $\varepsilon$-regret set of size $r'=r+3$.
  To prove this, we need to show (1) If $l^+(Q_0) \leq \varepsilon$ and
  $Q' = Q_0 \cup B$, then $l(Q') \leq \varepsilon$; and (2) If $l(Q') \leq \varepsilon$,
  then $B \subset Q'$ and, for $Q_0 = Q' \setminus B$,
  $l^+(Q_0) \leq \varepsilon$ over $P_0 = P' \setminus B$.

  We first prove (1) by showing that $l_x(Q') \leq \varepsilon$ for all $x \in \mathbb{S}^2$.
  First of all, we consider the case when $x \in \mathbb{S}^2_{+}$.
  Let $p^*=\arg\max_{p \in P'} \langle p,x \rangle $.
  If $p^* \in B$, then $l_x(Q') = 0$ because $B \subset Q'$;
  Otherwise, we have $p^* \in P_0$ and there always exists some $p \in Q_0$ such that
  $ \langle p,x \rangle \geq (1-\varepsilon) \cdot \langle p^*,x \rangle $
  and thus $l_x(Q') \leq \varepsilon$ because $Q_0$ is an $\varepsilon$-regret set of $P_0$.
  Next, we consider the case when $x \in \mathbb{S}^2 \setminus \mathbb{S}^2_{+}$.
  We want to show that the point with the highest score
  for any $x \in \mathbb{S}^2 \setminus \mathbb{S}^2_{+}$
  is always in $B$ and thus $l_x(Q') = 0$.
  Furthermore, we consider three cases for $x=(x[1],x[2],x[3])$ as follows:
  \begin{itemize}
  \item \textbf{Case 1.1 ($x[1] \geq 0, x[2] \geq 0, x[3] \leq 0$):}
  For any $p \in P_0$, we have
  \begin{displaymath}
    \langle p,x \rangle \leq p[1] \cdot x[1] + p[2] \cdot x[2] < x[1]+x[2] \leq \sqrt{2}
  \end{displaymath}
  In addition, we have
  \begin{displaymath}
    \langle b_3,x \rangle = x[1] + x[2] + (1-\eta) \cdot x[3] \geq \sqrt{2}
  \end{displaymath}
  Thus, $b_3$ always has a larger score than all points in $P_0$.
  This result also holds for $b_1$ or $b_2$ when $x[1] \leq 0$ or $x[2] \leq 0$
  and other dimensions are positive.
  \item \textbf{Case 1.2 ($x[1] \geq 0, x[2] \leq 0, x[3] \leq 0$):}
  For any $p \in P_0$, $ \langle p,x \rangle < x[1] \leq 1 $.
  Moreover, we have
  \begin{displaymath}
    \langle b_2,x \rangle = x[1] + x[2] + x[3] - \eta \cdot x[2]
  \end{displaymath}
  and
  \begin{displaymath}
    \langle b_3,x \rangle = x[1] + x[2] + x[3] - \eta \cdot x[3]
  \end{displaymath}
  If $ x[2] \leq x[3] $, then $ \langle b_2,x \rangle \geq \langle b_3,x \rangle $,
  and vice versa. So the minimum of $\max(\langle b_2,x \rangle,\langle b_3,x \rangle)$
  is always reached when $x[2] = x[3]$. In this case, we have
  \begin{displaymath}
    \langle b_2,x \rangle = \langle b_3,x \rangle = x[1] + (2-\eta) \cdot x[2]
  \end{displaymath}
  Let $x[2] = \beta$ and thus $ x[1] = \sqrt{1 - 2\beta^2} $.
  We consider the score $\langle b_2,x \rangle$ as a function $f(\beta)$, i.e.,
  \begin{displaymath}
    f(\beta) = \langle b_2,x \rangle = \sqrt{1 - 2\beta^2} + (2-\eta) \cdot \beta
  \end{displaymath}
  where $ \beta \in [-1,0] $. As $f(\beta)$ first increases and then decreases
  in the range $[-1,0]$, we have $\langle b_2,x \rangle \geq f(0) = 1$
  and $\langle b_2,x \rangle \geq f(-1) = \eta - 2 > 1$.
  Thus, the larger one between $\langle b_2,x \rangle$ and $\langle b_3,x \rangle$
  is always greater than the scores of all points in $P_0$ w.r.t.~$x$.
  Similar results can be implied when $x[2] \geq 0$ or $x[3] \geq 0$
  and other dimensions are negative.
  \item \textbf{Case 1.3 ($x[1] \leq 0, x[2] \leq 0, x[3] \leq 0$):}
  For any $p \in P_0$, $ \langle p,x \rangle \leq 0 $.
  In addition, there always exists $i\in[1,3]$ such that $x[i] \leq -\frac{\sqrt{3}}{3}$.
  Taking $x[1] \leq -\frac{\sqrt{3}}{3}$ as an example, we have
  \begin{displaymath}
    \langle b_1,x \rangle \geq \frac{\sqrt{3}}{3}\cdot\eta - \sqrt{3} > 0
  \end{displaymath}
  because $\eta > 3$.
  \end{itemize}
  We can prove $l(Q') \leq \varepsilon$ from the above three cases.

  To verify (2), we need to show: (2.1) If $B \not\subset Q'$, then $l(Q') > \varepsilon$;
  (2.2) If $l^+(Q_0) > \varepsilon$, then $Q_0 \cup B$ is not
  an $\varepsilon$-regret set of $P'$.
  The correctness of (2.1) is easy to prove: Taking $x=(-1,0,0) \in \mathbb{S}^2 $,
  it is obvious that $ \langle b_1,x \rangle > 0 $ and $ \langle p,x \rangle < 0 $
  for any $ p \in P'\setminus\{b_1\} $. So if $B \not\subset Q'$, then $l(Q') > 1$.
  The proof of (2.2) involves determining the value of $\eta$ according to $P_0$, $Q_0$,
  and $\varepsilon$. If $l^+(Q_0) > \varepsilon$,
  then there exists a point $p \in P_0 \setminus Q_0$ and a vector
  $ x \in \mathbb{S}^2_{+} $ such that
  $ (1-\varepsilon) \cdot \langle p,x \rangle > \omega(x,P_0) $.
  Since $l^+(Q_0)$ and the vector $x$ with $l^+(Q_0)=l_x(Q_0)$
  can be computed by a linear program~\cite{DBLP:journals/pvldb/NanongkaiSLLX10},
  it is guaranteed that such a point $p$ and a vector $x$ can be found in polynomial time.
  When $\eta > \frac{3 - (1-\varepsilon) \cdot \langle p,x \rangle}
  {x[i]}$, we have $ \langle b_i,x \rangle < (1-\varepsilon) \cdot \langle p,x \rangle $
  ($i=1,2,3$).
  Therefore, if $l^+(Q_0) > \varepsilon$,
  $ l_x(Q') > \varepsilon $ once $\eta$ is large enough
  and thus $Q'$ is not an $\varepsilon$-regret set of $P'$.
  We prove (2) from (2.1) and (2.2).

  We complete the reduction from RMS in $\mathbb{R}^3_{+}$ to GRMR in $\mathbb{R}^3$
  in polynomial time and hence prove the NP-hardness of GRMR in $\mathbb{R}^3$.
  Since the reduction can be generalized 
  to higher dimensions, GRMR is also NP-hard when $d>3$.
\end{proof}

\section{Proof of Lemma~\ref{lm:cand}}
\label{proof:cand}

\begin{proof}
  First of all, it is obvious that $R_{\varepsilon}(p) \neq \varnothing$ if $p \in S$
  since it must hold that $ x^*_i \in R_{\varepsilon}(p) $.
  Next, we will prove $R_{\varepsilon}(p) = \varnothing$ if $p \notin S$ by showing
  the minimum of the regret ratio $l_x(p) = 1 -\frac{\langle p,x \rangle}{\omega(x,P)}$
  of $p$ for any $x \in \mathbb{S}^1$ is greater than $\varepsilon$.
  Then, we consider each extreme point $t_i \in X$ separately.
  Let $f(x) = 1 - \frac{\langle p,x \rangle}{\langle t_i,x \rangle}$ where $x \in R(t_i)$
  and $t_i \in X$. According to Line~\ref{ln:cand:condition} of Algorithm~\ref{alg:grmr:2d},
  we have got if $p \notin S$, then $1-\frac{\langle p,x^*_i \rangle}{\langle t_i,x^*_i \rangle}
  > \varepsilon$ for all $i \in [|X|]$.
  Thus, we have
  \begin{displaymath}
    f(x) = 1 - \frac{\|p\|}{\|t_i\|}\cdot
    \frac{\cos(\theta(p)-\theta(x))}{\cos(\theta(t_i)-\theta(x))}
  \end{displaymath}
  where $x \in R(t_i)$, i.e., $x \in [x^*_{i-1},x^*_i]$.
  If $\theta(p)<\theta(t)$, $f(x)$ will monotonically increase with increasing $\theta(x)$;
  If $\theta(p)>\theta(t)$, $f(x)$ will monotonically decrease with increasing $\theta(x)$;
  And if $\theta(p)=\theta(t)$, $f(x)$ will be a constant $1 - \frac{\|p\|}{\|t_i\|}$.
  Thus, the minimum of $f(x)$ can only be reached when $ x=x^*_{i-1}$ or $ x^*_i $.
  In addition, we have
  \begin{displaymath}
    f(x^*_{i-1})=1-\frac{\langle p,x^*_{i-1} \rangle}{\langle t_{i-1},x^*_{i-1} \rangle} > \varepsilon
    \quad \textnormal{and} \quad
    f(x^*_i)=1-\frac{\langle p,x^*_i \rangle}{\langle t_i,x^*_i \rangle} > \varepsilon
  \end{displaymath}
  Combining the above results, we have $l_x(p)>\varepsilon$ for any $x \in \mathbb{S}^1$
  if $\frac{\langle p,x \rangle}{\langle t_i,x \rangle} < 1-\varepsilon$ for any $t_i \in X$
  and conclude the proof.
\end{proof}

\section{Proof of Lemma~\ref{lm:regret}}
\label{proof:regret}

\begin{proof}
  Firstly, the removal of any non-extreme point does not lead to any regret and
  it is safe to only consider extreme points for regret computation.
  For $s_i,s_j \in S (i<j)$ and $t \in X$ where
  $ \theta(s_i) \leq \theta(t) \leq \theta(s_j) $, we need to compute an upper bound
  $l_{ij}(t)$ of the maximum regret ratio of $\{s_i,s_j\}$ over $t$
  for any vector $x \in R(t)$. By sweeping all utility vectors in a counterclockwise
  direction, we observe that, similar to the proof of Lemma~\ref{lm:cand},
  \begin{displaymath}
    f_i(x)=1-\frac{\langle s_i,x \rangle}{\langle t,x \rangle}=1-
    \frac{\|s_i\|}{\|t\|}\cdot\frac{\cos(\theta(s_i)-\theta(x))}{\cos(\theta(t)-\theta(x))}
  \end{displaymath}
  monotonically increases with increasing $x$ since $\theta(s_i) \leq \theta(t)$ while
  \begin{displaymath}
    f_j(x)=1-\frac{\langle s_j,x \rangle}{\langle t,x \rangle}=1-
    \frac{\|s_j\|}{\|t\|}\cdot\frac{\cos(\theta(s_j)-\theta(x))}{\cos(\theta(t)-\theta(x))}
  \end{displaymath}
  monotonically decreases with increasing $x$ since $\theta(s_j) \geq \theta(t)$.
  Therefore, the regret ratio of $\{s_i,s_j\}$, i.e., $ \min\big(f_i(x),f_j(x)\big) $ is
  maximized when $ f_i(x)=f_j(x) $ (i.e., $\langle s_i,x \rangle=\langle s_j,x \rangle$).
  Finally, we have
  \begin{displaymath}
    l_{ij}=\max_{t \in X[s_i,s_j]}1-\frac{\langle s_i,x \rangle}{\langle t,x \rangle}
  \end{displaymath}
  is the upper bound of the maximum regret ratio of $\{s_i,s_j\}$ after removing the candidates between them.
\end{proof}

\section{Proof of Lemma~\ref{lm:cycle}}
\label{proof:cycle}

\begin{proof}
  For each extreme point $t \in X \setminus Q$, there always exists two points
  $s_i,s_j \in Q$ where $ \theta(s_i) \leq \theta(t) \leq \theta(s_j) $
  and $l_{ij} \leq \varepsilon$ because $C$ is a cycle of $G$.
  According to Lemma~\ref{lm:regret}, we have the regret ratio of $\{s_i,s_j\}$
  over $t$ is at most $\varepsilon$ for any $x \in R(t)$.
  Moreover, if $t \in X \cap Q$, it is obvious that
  the maximum regret ratio of $Q$ over $t$ is $0$.
  Therefore, the maximum regret ratio of $Q$ over $t$ is at most $\varepsilon$
  for every $t \in X$ and $Q$ is an $\varepsilon$-regret set of $P$.

  A subset $Q \subseteq P$ is called a locally minimal $\varepsilon$-regret set
  if (1) $Q$ is an $\varepsilon$-regret set and (2) $Q \setminus \{q\}$ is not
  an $\varepsilon$-regret set for any $q \in Q$.
  Obviously, the optimal result $Q^*_{\varepsilon}$
  (i.e., the globally minimum $\varepsilon$-regret set)
  is guaranteed to be locally minimal.
  First, if $Q$ is a locally minimal $\varepsilon$-regret set, we can get $Q \subseteq S$.
  For any $p \notin S$, if a subset $Q^\prime$ where $p \in Q^\prime$
  is an $\varepsilon$-regret set,
  it will hold that $Q^\prime \setminus \{p\}$ is still an $\varepsilon$-regret set
  since $R_{\varepsilon}(p)=\varnothing$, which implies that $Q^\prime$ is not locally minimal.
  Then, we arrange all points of $Q$ in a counterclockwise direction
  as $\{ q_1,\ldots,q_{|Q|} \}$.
  We prove that there exists an edge $(q_i \rightarrow q_{i+1})$ in $G$
  for every $i \in [|Q|]$ by contradiction.
  If $(q_i \rightarrow q_{i+1}) \notin E$, then either
  (1) $ \theta(q_{i+1})-\theta(q_i) > \pi $ or
  (2) $ 1 - \frac{\langle q_i,x^* \rangle}{\langle t,x^* \rangle} > \varepsilon$
  when $ \langle q_i,x^* \rangle = \langle q_{i+1},x^* \rangle $.
  In the previous case, we have $l(Q) > 1$ and $Q$ is not an $\varepsilon$-regret set.
  In the latter case, if there does not exist any $i^\prime < i$ or $i^{\prime\prime} > i+1$
  such that the regret ratio of
  $\{q_{i},q_{i+1}\} \cup \{q_{i^\prime}\}$ or $\{q_{i},q_{i+1}\} \cup \{q_{i^{\prime\prime}}\}$
  for $x^*$ is at most $\varepsilon$, we will have $l_{x^*}(Q) > \varepsilon$
  and $Q$ is not an $\varepsilon$-regret set; Otherwise,
  if there exists such $i^\prime < i$ or $i^{\prime\prime} > i+1$,
  we will have $q_i$ or $q_{i+1}$ is redundant and $Q$ is not locally minimal.
  Based on all above results, we conclude that there exists
  a cycle $C$ of $G$ corresponding to $Q$ as long as $Q$
  is a locally minimal $\varepsilon$-regret set of $P$.
\end{proof}

\end{appendix}

\end{document}